\documentclass[11pt]{article}
\usepackage{geometry}
\usepackage{amsfonts,amssymb}
\usepackage{amsmath, amsthm}
\usepackage{bigstrut}
\usepackage{multirow}
\usepackage{array}
\usepackage{tabularx}
\usepackage{colortbl}
\usepackage{todonotes,fancyhdr}
\usepackage{color}
\usepackage[utf8]{inputenc}
\usepackage{enumitem}
\usepackage{soul}
\usepackage{url}
\usepackage{algorithm}
\usepackage{algpseudocode}

\usepackage{pifont}
\usepackage{mathrsfs}
\usepackage{graphicx}
\usepackage[affil-it]{authblk}
\usepackage{lmodern}

\usepackage{tikz}
\usetikzlibrary{calc,matrix,fit,arrows,shapes,trees,backgrounds,arrows.meta}
\tikzset{highlight/.style={rectangle,
		                   fill=red!15,
		                   blend mode = multiply,
		                   rounded corners = 0.5 mm,
		                   inner sep=1pt,
		                   fit = #1}}
\newcommand{\tikzmark}[1]{\tikz[overlay,remember picture] \node (#1) {};}

\usepackage{forest}
\usepackage{subcaption}
\usepackage[T1]{fontenc}
\usepackage{rotating}
\usepackage{multimedia}
\usepackage[skip=4pt]{caption}
\usepackage{booktabs}
\usepackage{xcolor}
\usepackage{array}
\usepackage{dcolumn}
\usepackage{nicematrix}
\usepackage{longtable}
\usepackage{tabularx}

\usepackage{makecell}

\usepackage[toc,page]{appendix}
\usepackage{mathtools}
\usepackage[lambda, advantage, operators, sets, adversary, landau, probability, notions, logic, ff, mm, primitives, events, complexity, oracles, asymptotics, keys]{cryptocode} 

\usepackage[final]{pdfpages}

\newtheorem{definition}{Definition}
\newtheorem{theorem}{Theorem}
\newtheorem{corollary}{Corollary}
\newtheorem{proposition}{Proposition}

\newtheorem{lemma}{Lemma}

\newcounter{example}[section]
\newenvironment{example}[1][]{\refstepcounter{example}\par\medskip
	\noindent \textbf{Example~\theexample. #1} \rmfamily}{\medskip}

\newtheorem{openroblem}{\bf Open Problem}

\newtheorem{conjecture}{\bf Conjecture}
\newcommand{\shorten}[1]{}

\begin{document}
	
	\title{Entropoid Based Cryptography}
	
	\author{Danilo Gligoroski\thanks{Department of Information Security and Communication Technologies, Norwegian University of Science and Technology - NTNU}}
	
	\maketitle

	\begin{abstract}
		 
		The algebraic structures that are non-commutative and non-associative known as entropic groupoids that satisfy the \emph{"Palintropic"} property i.e., $x^{\mathbf{A} \mathbf{B}} = (x^{\mathbf{A}})^{\mathbf{B}} = (x^{\mathbf{B}})^{\mathbf{A}} = x^{\mathbf{B} \mathbf{A}}$ were proposed by Etherington in '40s from the 20th century. Those relations are exactly the Diffie-Hellman key exchange protocol relations used with groups. The arithmetic for non-associative power indices known as Logarithmetic was also proposed by Etherington and later developed by others in the 50s-70s. However, as far as we know, no one has ever proposed a succinct notation for exponentially large non-associative power indices that will have the property of fast exponentiation similarly as the fast exponentiation is achieved with ordinary arithmetic via the consecutive rising to the powers of two. 
		
		In this paper, we define ringoid algebraic structures $(G, \boxplus, *)$ where $(G, \boxplus) $ is an Abelian group and $(G, *)$ is a non-commutative and non-associative groupoid with an entropic and palintropic subgroupoid which is a quasigroup, and we name those structures as Entropoids. We further define succinct notation for non-associative bracketing patterns and propose algorithms for fast exponentiation with those patterns. 
		
		Next, by an analogy with the developed cryptographic theory of discrete logarithm problems, we define several hard problems in Entropoid based cryptography, such as Discrete Entropoid Logarithm Problem (DELP), Computational Entropoid Diffie-Hellman problem (CEDHP), and Decisional Entropoid Diffie-Hellman Problem (DEDHP). We post a conjecture that DEDHP is hard in Sylow $q$-subquasigroups. Next, we instantiate an entropoid Diffie-Hellman key exchange protocol. Due to the non-commutativity and non-associativity, the entropoid based cryptographic primitives are supposed to be resistant to quantum algorithms. At the same time, due to the proposed succinct notation for the power indices, the communication overhead in the entropoid based Diffie-Hellman key exchange is very low: for 128 bits of security, 64 bytes in total are communicated in both directions, and for 256 bits of security, 128 bytes in total are communicated in both directions. 
		
		Our final contribution is in proposing two entropoid based digital signature schemes. The schemes are constructed with the Fiat-Shamir transformation of an identification scheme which security relies on a new hardness assumption: computing roots in finite entropoids is hard. If this assumption withstands the time's test, the first proposed signature scheme has excellent properties: for the classical security levels between 128 and 256 bits, the public and private key sizes are between 32 and 64, and the signature sizes are between 64 and 128 bytes. The second signature scheme reduces the finding of the roots in finite entropoids to computing discrete entropoid logarithms. In our opinion, this is a safer but more conservative design, and it pays the price in doubling the key sizes and the signature sizes.
		
		We give a proof-of-concept implementation in SageMath 9.2 for all proposed algorithms and schemes in an appendix.		 
	\end{abstract}

	\newpage
	\tableofcontents

	\newpage

\section{Introduction}
The arithmetic of non-associative indices (shapes or patterns of bracketing with a binary operation $*$) has been defined as \emph{"Logarithmetic"} by Etherington in the '40s of the 20th century. One of the most interesting properties that have been overlooked by modern cryptography is the Discrete Logarithm problem and the Diffie-Hellman key exchange protocol in the non-commutative and non-associative logarithmetic of power indices. In light of the latest developments in quantum computing, Shor's quantum algorithm that can solve DL problem in polynomial time if the underlying algebraic structures are commutative groups, and the post-quantum cryptography, it seems that there is an opening for a \emph{rediscovery} of Logarithmetic and its applications in Cryptography. 

Etherington himself, as well as other authors later (\cite{harding1971probabilities}, \cite{dacey1973non}), noticed that the notation of shapes introduced in \cite{etherington1940xv} quickly gets complicated (and from our point of view for using them for cryptographic purposes, incapable of handling exponentially large indices). 

Inspired by the work of Etherington, in a series of works in '50s, '60s and '70s of the last century, many authors such as Robinson \cite{robinson1949non}, Popova \cite{popova1951logarithmetics}, Evans \cite{evans1957nonassociative}, Minc \cite{minc1959theorems}, Bollman \cite{bollman1967formal}, Harding \cite{harding1971probabilities}, Dacey \cite{dacey1973non}, Bunder \cite{bunder1976commutative}, Trappmann \cite{trappmann2007arborescent}, developed axiomatic number systems for  non-associative algebras which in many aspects resemble the axiomatics of ordinary number theory. While the results from that development are quite impressive such as the fundamental theorem of non-associative arithmetic (prime factorization of indices~\cite{evans1957nonassociative}), or the analogue of the last Fermat Theorem \cite{popova1951logarithmetics,evans1957nonassociative,minc1959theorems}, the construction of the shapes was essentially sequential. In cryptography, we need to operate with power indices of exponential sizes. Thus, we have to define shapes over non-associative and non-commutative groupoids that allow fast exponentiation, similarly as it is done with the consecutive rising to the powers of two in the standard modular arithmetic, while keeping the variety of possible outcomes of the calculations, as the flagship aspect of the non-commutative and non-associative structures.

\subsection{Our Contribution}
We first define a general class of groupoids $(G, *)$ (sets $G$ with a binary operation $*$) over direct products of finite fields with prime characteristics $(\mathbb{F}_p)^L$ that are \emph{"Entropic"} (for every four elements $x, y, z$ and $w$, if $x * y = z * w$ then $x * z = y * w$). Then, for $L=2$ we find instances of those operations $*$ that are nonlinear in $(\mathbb{F}_p)^2$, non-commutative and non-associative. In order to  compute the powers $x^a$ where $a \in \mathbb{Z}^+ $, of elements $x \in G$, due to the non-associativity, we need to know some exact bracketing shape $a_s$, and we denote the power indices as pairs $\mathbf{A} = (a, a_s)$.
Etherington defined the Logarithmetic of indices $\mathbf{A}$, $\mathbf{B}$, by defining their addition $\mathbf{A} + \mathbf{B} $ and multiplication $\mathbf{A} \mathbf{B} $ as $x^{\mathbf{A} + \mathbf{B}} = x^{\mathbf{A}} * x^{\mathbf{B}} $ and $x^{\mathbf{A} \mathbf{B}} = (x^{\mathbf{A}})^{\mathbf{B}} $. He further showed that the power indices of entropic groupoids, satisfy the \emph{"Palintropic"} property i.e., $x^{\mathbf{A} \mathbf{B}} = (x^{\mathbf{A}})^{\mathbf{B}} = (x^{\mathbf{B}})^{\mathbf{A}} = x^{\mathbf{B} \mathbf{A}}$ which is the exact form of the Diffie-Hellman key exchange protocol.

We further analyze the chosen instances of entropic groupoids for how to find left unit elements, how to compute the multiplicative inverses, how to define addition in those groupoids, and how to find generators $g \in G$ that generate subgroupoids with a maximal size of $(p - 1)^2$ elements. We show that these maximal multiplicative subgroupoids are quasigroups. We also define Sylow $q$-subquasigroups. Having all this, we define several hard problems such as Discrete Entropoid Logarithm Problem (DELP), Computational Entropoid Diffie-Hellman problem (CEDHP), and Decisional Entropoid Diffie-Hellman Problem (DEDHP). We post a conjecture that DEDHP is hard in Sylow $q$-subquasigroups. Next, we propose a new hard problem specific for the Entropoid algebraic structures: Computational Discrete Entropoid Root Problem (CDERP). 

We propose instances of Diffie-Hellman key exchange protocol in those entropic, non-abelian and non-associative groupoids. Due to the hidden nature of the bracketing pattern of the power index (if chosen randomly from an exponentially large set of possible patterns), it seems that the current quantum algorithms for finding the discrete logarithms, but also all classical algorithms for solving DLP (such as Pollard's rho and kangaroo, Pohlig-Hellman, Baby-step giant-step, and others) are not suitable to address the CLDLP.

Based on CDERP, we define two post-quantum signature schemes.

Our notation for power indices that we introduce in this paper differs from the notation that Etherington used in \cite{etherington1940xv}, and is adapted for our purposes to define operations of rising to the powers that have exponentially (suitable for cryptographic purposes) big values. However, for the reader's convenience we offer here a comparison of the corresponding notations: the shape $s$ in \cite{etherington1940xv} means a power index $\mathbf{A} = (a, a_s)$ here; degree $\delta(s)$ in \cite{etherington1940xv} means $a$ here; the notations of altitude $\alpha(s)$ and mutability $\mu(s)$ in \cite{etherington1940xv} do not have a direct interpretation in the notation of $\mathbf{A} = (a, a_s)$ but are implicitly related to $a_s$.

\section{Mathematical Foundations for Entropoid Based Cryptography}

\subsection{General definition of Logarithmetic}

\begin{definition}
	A groupoid $(G, *)$ is an algebraic structure with a set $G$ and a binary operation $*$ defined uniquely for all pairs of elements $x$ and $y$ i.e., 
	\begin{equation}\label{Eq:Groupoid}
		\forall x, y \in G,\ \ \  \exists! \  (x*y) \in G.
	\end{equation}
\end{definition}

The following definitions are taken and adapted for our purposes, from \cite{etherington1940xv}, \cite{etherington1945transposed}, \cite{etherington1958groupoids} and \cite{etherington1965quasigroups}.

\begin{definition}\label{Def:EntropicOperation}
We say the binary operation $*$ of the groupoid $(G,*)$ is entropic, if for every four elements $x, y, z, w \in G$, the following relation holds:
\begin{equation}\label{Eq:EntropicOperations}
	\text{If}\ \ \ x * y = z * w\ \ \ \text{then}\ \   \ x * z = y * w.
\end{equation}
\end{definition}

\begin{definition}\label{Def:General-Power-Index}
Let $x \in G$ is an element in the groupoid $(G, *)$ and let $a \in \mathbb{N}$ is a natural number. A bracketing shape (pattern) for multiplying $x$ by itself, $a$ times is denoted by $a_{s}$ i.e. 
\begin{equation}\label{Eq:a_{s}}
	a_{s}  \colon \underbrace{(x * \ldots (x*x) \ldots )}_{a \text{\ copies of } x}.
\end{equation}
The pair $\mathbf{A} = (a, a_{s})$ is called a power index.
\end{definition}

Let us denote by $S_a$ the set of all possible bracketing shapes $a_s$ that use $a$ instances of an element $x$ i.e. 
\begin{equation}\label{Eq:AllShapes}
	S_a(x) = \{ a_s\  |\  a_s \text{\ is a bracketing shape with}\ a \text{ instances of an element}\ x \}.
\end{equation}

\begin{proposition}
	If $\mathbf{A} = (a, a_{s})$ is a power index, then  
	\begin{equation}\label{Eq:Catalan}
		|S_a(x)| = C_{a-1} = \frac{1}{a}\binom{2a - 2}{a-1}, 
	\end{equation}
	where $C_{a-1}$ is the $(a-1)$-th Catalan number \cite[Sequence A000108]{sloane2007line}. \qed
\end{proposition}

Two of those bracketing shapes are characteristic since they can be described with an iterative sequential process starting with $x * x$ absorbing the factors $x$ one at a time either in the direction from left to right or from  right to left (in \cite{etherington1940xv} shapes that absorb the terms one at a time are called primary shapes).

\begin{definition}\label{Def:left-to-right-and-vice-cersa}
	We say that the power index $\mathbf{A} = (a, a_{s})$ has a primary left-to-right bracketing shape if
	\begin{equation}\label{Eq:a_{s}-primary-left-to-right}
		a_{s}  \colon \underbrace{(( \ldots ((x *^{\tikzmark{a}} x) * x) \ldots ) *^{\tikzmark{b}} x) }_{a \text{\ copies of } x}.
		\begin{tikzpicture}[overlay,remember picture,out=315,in=225,distance=0.4cm]
			\draw[>=stealth, ->,red, shorten <= -4pt, shorten >= -4pt] (a) edge [bend left=80] (b);
		\end{tikzpicture}
	\end{equation} 
    We shortly write it as: $\mathbf{A} \equiv x_{(a-1,x)\texttt{to-right}}$, denoting that it starts with $x$ and applies $a-1$ right multiplications by $x$. For some generic and unspecified bracketing shape $s$ the notation  $$s_{(k,x)\texttt{to-right}}$$ denotes a sequential extension of $s$ by $k$ right multiplications by $x$. This includes the formal notation  $s \equiv s_{(0,x)\texttt{to-right}}$ which means the shape $s$ itself extended with zero right multiplications by $x$.
	
	We say that the power index $\mathbf{A} = (a, a_{s})$ has a primary right-to-left bracketing shape if
	\begin{equation}\label{Eq:a_{s}-primary-right-to-left}
		a_{s}  \colon \underbrace{(x *^{\tikzmark{b}} ( \ldots (x * (x *^{\tikzmark{a}} x) \ldots ))}_{a \text{\ copies of } x}.
		\begin{tikzpicture}[overlay,remember picture,out=315,in=225,distance=0.4cm]
			\draw[>=stealth, ->,red, shorten <= -4pt, shorten >= -4pt] (a) edge [bend right=80] (b);
		\end{tikzpicture}
	\end{equation} 
	We shortly write it as: $\mathbf{A} \equiv x_{(x, a-1)\texttt{to-left}}$, denoting that it starts with $x$ and applies $a-1$ left multiplications by $x$. For some generic and unspecified bracketing shape $s$ the notation  $$s_{(x,k)\texttt{to-left}}$$ denotes a sequential extension of $s$ by $k$ left multiplications by $x$. This includes the formal notation  $s \equiv s_{(x,0)\texttt{to-right}}$ which means the shape $s$ itself extended with zero left multiplications by $x$.
   
\end{definition}

\begin{definition}\label{Def:Addition-and-Multiplication-of-Endomorphisms}
Every power index $\mathbf{A} = (a, a_{s})$ can be considered as an endomorphism on $G$ (a mapping of $G$ to itself)
\begin{equation}\label{Eq:Endomorphism}
	\mathbf{A} \colon x \rightarrow x^{\mathbf{A}}.
\end{equation}
If $\mathbf{A}  = (a, a_{s})$ and  $\mathbf{B}  = (b, b_{s})$ are two mappings, we define their sum $\mathbf{A} \mathbf{+} \mathbf{B}$ 
as the power index of the product of two powers i.e.:
\begin{equation}\label{Eq:SumOfEndomorphisms}
	x^{\mathbf{A} \mathbf{+} \mathbf{B}} = x^{\mathbf{A}} * x^{\mathbf{B}},
\end{equation}
and we define their product $\mathbf{A} \mathbf{\times} \mathbf{B}$ (or shortly $\mathbf{A} \mathbf{B}$) as the power index of an expression obtained by replacing every factor in the expression of the shape $b_{s}$ by a complete shape $a_{s}$ i.e.:
\begin{equation}\label{Eq:ProductOfEndomorphisms}
	x^{\mathbf{A} \mathbf{B}} = (x^{\mathbf{A}})^{\mathbf{B}}.
\end{equation}
\end{definition}

\begin{definition}
	If the sum of any two endomorphisms of $G$ is also and endomorphism of $G$, we say that $G$ has additive endomorphisms.
\end{definition}

\begin{definition}\label{Def:Equated-Indices}
	Let $(G, *)$ be a given groupoid. Two power indices $\mathbf{A}  = (a, a_{s})$ and  $\mathbf{B}  = (b, b_{s})$ are equal if and only if $x^{\mathbf{A}} = x^{\mathbf{B}} $ for all $x \in G$.
\end{definition}

\begin{definition}
	The Logarithmetic $(L(G), \mathbf{+}, \mathbf{\times})$ is the algebra over the equated indices from Definition \ref{Def:Equated-Indices} with operations $\mathbf{+} $ and $\mathbf{\times}$ as defined in Definition \ref{Def:Addition-and-Multiplication-of-Endomorphisms}.
\end{definition}

\begin{definition}
	If $x^{\mathbf{A} \mathbf{B}} = x^{\mathbf{B} \mathbf{A}}$ for all $x \in G$ and for all power indices $\mathbf{A}$ and $\mathbf{B}$, we say that the groupoid $(G, *)$ is palintropic.
\end{definition}

\begin{theorem}[Etherington, \cite{etherington1958groupoids}]
	If the groupoid $(G,*)$ has additive endomorphisms, then (i) power indices are endomorphisms of $G$, (ii) $G$ is palintropic. \qed
\end{theorem}

\begin{theorem}[Murdoch \cite{murdoch1939quasi}, Etherington, \cite{etherington1958groupoids}]\label{Thm:Basic-theorem-for-powers-in-entropic-groupoids}
	If the groupoid $(G,*)$ is entropic, then for every $x, y \in G$ 
	\begin{equation}\label{Eq:Power-of-product-of-x-and-y}
		(x*y)^\mathbf{A} = x^\mathbf{A}*y^\mathbf{A},
	\end{equation}
	and 
	\begin{equation}\label{Eq:Power-of-power-of-x-general}
		x^{\mathbf{A} \mathbf{B}} = (x^\mathbf{A})^\mathbf{B} = (x^\mathbf{B})^\mathbf{A} = x^{\mathbf{B} \mathbf{A}}.
	\end{equation}
	\qed
\end{theorem}

\begin{definition}\label{Def:generator}
	We say that $g \in G$ is the generator of the groupoid $(G, *)$ if 
	\begin{equation}
		\forall x \in G, \ \ \ \exists \mathbf{A},\ \ \text{such that}\ \ \ x = g^\mathbf{A}.
	\end{equation}
	In that case we write \begin{equation}\langle g \rangle = G.\end{equation}
\end{definition}

\textbf{Note:} In the following subsection and the rest of the paper, we will overload the operators $+$ and $\times$ for addition and multiplication of power indices defined in the previous definitions, with the same operators for the operations of addition and multiplication in the finite field $\mathbb{F}_p$. However, there will be no confusion since the operations act over different domains.

\subsection{Entropic Groupoids Over $(\mathbb{F}_p)^L$ }
Let a set $G$ is a direct product of $L$ instances of the finite field $\mathbb{F}_p$ with $p$ elements, where $p$ is a prime number i.e., 
\begin{equation}\label{Eqn:G_Direct_product}
	G = \underbrace{\mathbb{F}_p \times \ldots \mathbb{F}_p}_L.
\end{equation}

Let $G = (\mathbb{F}_p)^L$ for $L \ge 2$, and let us represent elements $x, y, z, w \in G$ as symbolic tuples $x = (x_1, \ldots x_L)$, $y = (y_1, \ldots y_L)$, $z = (z_1, \ldots z_L)$ and $w = (w_1, \ldots w_L)$.

\begin{definition}[Design criteria]\label{Def:Design-Criteria-for-the-operation}
The binary operation $*$ over $G$ should meet the following:
\begin{enumerate}[leftmargin=*,label=Design criterium \arabic*:]
	\item \textbf{Entropic.} The operation $*$ should be entropic as defined in (\ref{Eq:EntropicOperations}).
	\item \textbf{Nonlinear.} The operation $*$ should be nonlinear with the respect of the addition and multiplication operations in the finite field $\mathbb{F}_p$.
	\item \textbf{Non-commutative.} The operation $*$ should be non-commutative.
	\item \textbf{Non-associative.} The operation $*$ should be non-associative.
\end{enumerate}

\end{definition}

One simple way to find an operation that satisfies the design criteria from Definition \ref{Def:Design-Criteria-for-the-operation} would be to define it for $L=2$ for the most general quadratic $2L$-variate polynomials as follows:
\begin{align}\label{Eq:Generic-2L-variate-1}
	x * y = (x_1, x_2) * (y_1, y_2) = & (P_1(x_1, x_2, y_1, y_2), P_2(x_1, x_2, y_1, y_2)),\\
	\label{Eq:Generic-2L-variate-2}
	P_1(x_1, x_2, y_1, y_2) = & a_1 + a_2 x_1 + a_3 x_2 + a_4 y_1 + a_5 y_2 + a_6 x_1 y_1 + a_7 x_1 y_2 + a_8 x_2 y_1 + \nonumber \\
	                          &+ a_9 x_2 y_2 + a_{10} x_1^2 + a_{11} x_1 x_2 + a_{12} x_2^2 + a_{13} y_1^2 + a_{14} y_1 y_2 + a_{15} y_2^2,\\
	\label{Eq:Generic-2L-variate-3}
	P_2(x_1, x_2, y_1, y_2) = & b_1 + b_2 x_1 + b_3 x_2 + b_4 y_1 + b_5 y_2 + b_6 x_1 y_1 + b_7 x_1 y_2 + b_8 x_2 y_1 + \nonumber \\
	&+ b_9 x_2 y_2 + b_{10} x_1^2 + b_{11} x_1 x_2 + b_{12} x_2^2 + b_{13} y_1^2 + b_{14} y_1 y_2 + b_{15} y_2^2,
\end{align}
where 30 variables $a_1, \ldots, a_{15}$ and $b_1, \ldots, b_{15}$ are from $\mathbb{F}_p$, and their relations are to be determined such that (\ref{Eq:EntropicOperations}) holds. It turns out that searching for those relations with 30 symbolic variables is not a trivial task for the modern computer algebra systems such as Sage \cite{stein2005sage} and Mathematica \cite{wolfram1991mathematica}. There are many ways how to simplify further the $2L$-variate polynomials (\ref{Eq:Generic-2L-variate-2}) and (\ref{Eq:Generic-2L-variate-3}) by removing some of their monomials. We present one such a simplification approach by defining the operation $*$ as follows:
\begin{align}\label{Eq:Simple-2L-variate-1}
	x * y = (x_1, x_2) * (y_1, y_2) = & (P_1(x_1, x_2, y_1, y_2), P_2(x_1, x_2, y_1, y_2)),\\
	\label{Eq:Simple-2L-variate-2}
	P_1(x_1, x_2, y_1, y_2) = & a_1 + a_2 x_1 + a_3 x_2 + a_4 y_1 + a_6 x_1 y_1 + a_8 x_2 y_1,\\
	\label{Eq:Simple-2L-variate-3}
	P_2(x_1, x_2, y_1, y_2) = &  b_1 + b_2 x_1 + b_3 x_2 + b_5 y_2 + b_7 x_1 y_2 + b_9 x_2 y_2.
\end{align}

\begin{openroblem}
	Define a generic and systematic approach for finding solutions that will satisfy the design criteria of Definition \ref{Def:Design-Criteria-for-the-operation} for higher dimensions ($L>2$) and higher nonlinearity (the degree of the polynomials to be higher than 2).
\end{openroblem}

For the simplified system (\ref{Eq:Simple-2L-variate-1}) - (\ref{Eq:Simple-2L-variate-3}) one solution that satisfies all design criteria from Definition \ref{Def:Design-Criteria-for-the-operation} is the following:
\begin{definition}\label{Def:OperationStar}
	Let $x = (x_1, x_2), y = (y_1, y_2)$ be two elements of the set $G = \mathbb{F}_p \times \mathbb{F}_p$. The operation $*$, i.e. $x * y$ is defined as:
	{\footnotesize 
	\begin{equation}\label{Eq:OperationStar}
		(x_1, x_2) * (y_1, y_2) = \left( \frac{{a_3} ({a_8} {b_2}-{b_7})}{{a_8} {b_7}}+{a_3}
		{x_2}+\frac{{a_8} {b_2} {y_1}}{{b_7}}+{a_8} {x_2} {y_1},-\frac{{b_2}
			({a_8}-{a_3} {b_7})}{{a_8} {b_7}}+\frac{{a_3} {b_7}
			{y_2}}{{a_8}}+{b_2} {x_1}+{b_7} {x_1} {y_2}\right),
	\end{equation}
	}where $a_3, a_8, b_2, b_7 \in \mathbb{F}_p$, $a_8 \neq 0$ and $b_7 \neq 0$, and the operations $-$ and $/$ are the operations of subtraction and division in $\mathbb{F}_p$.	
\end{definition}

\begin{openroblem}
	Find efficient operations $*$ that satisfy the design criteria of Definition \ref{Def:Design-Criteria-for-the-operation} and have as little as possible operations of addition, subtraction, multiplication and division in $\mathbb{F}_p$.
\end{openroblem}

The next Corollary can be easily proven by simple expression replacements.
\begin{corollary}\label{Cor:Basic-properties-of-multiplication}
	Let $G = (\mathbb{F}_p)^2$, and let the operation $*$ is defined with Definition \ref{Def:OperationStar}. Then: 
	\begin{enumerate}
		\item The groupoid $(G, *)$ is entropic groupoid i.e., for every $x, y, z, w \in G$ if $ x*y = z * w $ then $x*z = y*w. $
		\item The element $\mathbf{0}_* = \left( -\frac{a_3}{a_8}, -\frac{b_2}{b_7} \right) $ is the multiplicative zero for the groupoid $(G, *)$, i.e.
		\begin{equation}\label{Eq:Multiplicative-Zero-In-G}
			\forall x \in G,\ \ \  x * \mathbf{0}_* = \mathbf{0}_* * x = \mathbf{0}_* \ \ .
			\end{equation}
		\item The element $\mathbf{1}_* = \left( \frac{1}{b_7} - \frac{a_3}{a_8}, \frac{1}{a_8} -\frac{b_2}{b_7} \right) $ is the multiplicative left unit for the groupoid $(G, *)$, i.e.
		\begin{equation}\label{Eq:Multiplicative-Left-Unit-In-G}
			\forall x \in G,\ \ \  \mathbf{1}_* * x = x \ \ .
		\end{equation}
		\item For every $x = (x_1, x_2) \neq \mathbf{0}_*$ its inverse multiplicative element $x^{-1}_*$ with the respect of the left unit element $\mathbf{1}_*$ is given by the following equation 
		\begin{equation}\label{Eq:Multiplicative-Inverse-In-G}
			x^{-1}_* =  \left( \frac{1 - a_3 b_2 - a_3 b_7 x_2}{a_8 (b_2 + b_7 x_2)}, 
			\frac{1 - a_3 b_2 - a_8 b_2 x_1}{b_7 (a_3 + a_8 x_1)} \right) \ \ ,
		\end{equation}
		for which it holds
		\begin{equation}\label{Eq:Multiplicative-Inverses-and-Left-Unit}
			x * x^{-1}_* = x^{-1}_* * x = \mathbf{1}_* \ \ .
		\end{equation}
	\end{enumerate} \qed
\end{corollary}

\begin{proposition}
	There are $p-1$ distinct square roots of the left unity $\mathbf{1}_*$, i.e. 
	\begin{equation}
		\mathbb{S}(p) = \{ x\  |\  x * x = \mathbf{1}_*  \}, \ \ \text{and}\ \ \ |\mathbb{S}(p)| = p-1.
	\end{equation} \qed
\end{proposition}

\begin{definition}\label{Def:Addition-in-Ringoid}
	Let $x = (x_1, x_2), y = (y_1, y_2)$ be two elements of the set $G = \mathbb{F}_p \times \mathbb{F}_p$. The additive operation $\boxplus$, i.e. $x \boxplus y$ is defined as:
	\begin{equation}\label{Eq:OperationPlus}
		(x_1, x_2) \boxplus (y_1, y_2) = \left( x_1 + y_1 + \frac{a_3}{a_8}, x_2 + y_2 + \frac{b_2}{b_7} \right) ,
	\end{equation}
	where $a_3, a_8, b_2, b_7 \in \mathbb{F}_p$, are defined in Definition \ref{Def:OperationStar}. 
	
	Let us denote by $\boxminus$ the corresponding "inverse" of the additive operation $\boxplus$. Its definition is given by the following expression:
	\begin{equation}\label{Eq:OperationMinus}
		(x_1, x_2) \boxminus (y_1, y_2) = \left( x_1 - y_1 - \frac{a_3}{a_8}, x_2 - y_2 - \frac{b_2}{b_7} \right).
	\end{equation}
	We can use the operation $\boxminus$ also as a unary operator. If we write $\boxminus x$ then we mean 
\begin{equation}
		\boxminus x \stackrel{def}{=} \mathbf{0}_* \boxminus x = \left( -\frac{a_3}{a_8} - x_1, -\frac{b_2}{b_7} - x2 \right). 
\end{equation}
\end{definition}

One can check that "minus one" i.e. $\boxminus \mathbf{1}_*$ is also a square root of the left unity i.e that it holds: $\boxminus \mathbf{1}_* * \boxminus \mathbf{1}_* = \mathbf{1}_*$. 

A consequence of Corollary \ref{Cor:Basic-properties-of-multiplication} and Definition \ref{Def:Addition-in-Ringoid} is the following
\begin{corollary}
	The algebraic structure $(G, \boxplus, *)$ is a ringoid where $(G, \boxplus) $ is an Abelian group with a neutral element $\mathbf{0}_*$, $(G, *)$ is a non-commutative and non-associative groupoid with a zero element $\mathbf{0}_*$, a left unit element $\mathbf{1}_*$ and $*$ is distributive over $\boxplus$ i.e. 
	\begin{equation}\label{Eq:Distributive-laws}
		x * (y \boxplus z) = (x*y) \boxplus (x*z) \ \ \text{and}\ \ (x \boxplus y) * z = (x*z) \boxplus (y*z).
	\end{equation}
\end{corollary}

\begin{definition}\label{Def:Entropi-ring-Entropoid}
	The ringoid $\mathbb{E}_{p^2} = (G, \boxplus, *)$ with operation $*$ defined with Definition \ref{Def:OperationStar} and operation $\boxplus$ defined with Definition \ref{Def:Addition-in-Ringoid} is called a \emph{Finite Entropic Ring} or \emph{Finite Entropoid} with $p^2$ elements. For given values of $p, a_3, a_8, b_2$ and $b_7$ it will be denoted as $\mathbb{E}_{p^2}( a_3, a_8, b_2, b_7)$. \footnote{The ringoid $(G, \boxplus, *)$ is neither a neofield (since $(G, *)$ is not a group), nor a Lie ring (since the Jacobi identity is not satisfied), but is built with the operation $*$ given by the entropic identity (\ref{Eq:EntropicOperations}). We are aware of the work of J.D.H. Smith and A.B. Romanowska in \cite{romanowska2015hopf} that refer to the entropic Jonsson-Tarski algebraic varieties, but for our purposes, and to be more narrow with the definition of the algebraic structure that we will use,  we give a formal name of this ringoid as "Entropoid".}
\end{definition}

In the rest of the text we will use interchangeably the notations $\mathbb{E}_{p^2}( a_3, a_8, b_2, b_7)$ and if in text  context, the constants $a_3, a_8, b_2, b_7$ are not important, just $\mathbb{E}_{p^2}$. We will also assume that the choice for the parameters is $a_3 \neq 0$ and $b_2 \neq 0$. To shorten the mathematical expressions, when the meaning is clear from the context, we will also overload the symbol $\mathbb{E}_{p^2}$ with two interpretations: as a set $\mathbb{E}_{p^2} = G = \mathbb{F}_p \times \mathbb{F}_p$ and as an algebraic structure $\mathbb{E}_{p^2} = (G, \boxplus, *)$.

\begin{definition}\label{Def:Multiplicative-subgroupoid}
	Let us define the subset $\mathbb{E}_{(p-1)^2}^* \subset \mathbb{E}_p$ as
	\begin{equation}\label{Eq:Multiplicative-subgroupoid}
		\mathbb{E}_{(p-1)^2}^* = \bigg( \Big(\mathbb{F}_p\setminus \Big\{ -\frac{a_3}{a_8} \Big\} \Big) \times  \Big(\mathbb{F}_p\setminus \Big\{ -\frac{b_2}{b_7} \Big\} \Big)  \bigg).
	\end{equation}
	We say that the groupoid $(\mathbb{E}_{(p-1)^2}^*, *)$ is the maximal multiplicative subgroupoid of $\mathbb{E}_{p^2}$.
\end{definition}

It is clear from Definition \ref{Def:Multiplicative-subgroupoid} that the multiplicative subgroupoid $\mathbb{E}_{(p-1)^2}^*$ has $(p-1)^2$ elements $x = (x_1, x_2)$ such that $x_1$ is not the first component and $x_2$ is not the second component of the multiplicative zero $\mathbf{0}_*$. It has one additional property: it is a quasigroup, and that is stated in the following Lemma.

\begin{lemma}
	The maximal multiplicative groupoid $(\mathbb{E}_{(p-1)^2}^*, *)$ defined in Definition \ref{Def:Multiplicative-subgroupoid} with the operation $*$ defined in Definition \ref{Def:OperationStar} is a non-commutative and non-associative quasigroup with a left unit element $\mathbf{1}_*$.
\end{lemma}
\begin{proof}
	Non-commutativity, non-associativity and the proof about the left unit element come directly by the definition of the operation $*$, and Corollary \ref{Cor:Basic-properties-of-multiplication}. The only remaining part is to prove that for every $c = (c_1, c_2) \in \mathbb{E}_{(p-1)^2}^*$ and every $d = (d_1, d_2) \in \mathbb{E}_{(p-1)^2}^*$, the equations $c * x = d$ and $x * c = d$ have always solutions. It is left as an exercise for the reader to replace the values of $c$ and $d$, to apply the definition of the operation $*$ annd with simple polynomial algebra to get get two equations for $x_1$ and $x_2$ with a unique solution $x = (x_1, x_2)$.
\end{proof}


We will use the notation $(\mathbb{E}^*_{\nu}, *)$ for the subgroupoids of $(\mathbb{E}_{(p-1)^2}^*, *)$, with $\nu$ elements. That means that if we are given $(\mathbb{E}^*_{\nu}, *)$, then $\mathbb{E}^*_{\nu} \subseteq \mathbb{E}_{(p-1)^2}^*$, $\forall x, y \in \mathbb{E}^*_{\nu}$, $x * y \in \mathbb{E}^*_{\nu}$ and $|\mathbb{E}^*_{\nu}| = \nu$. Using the Etherington terminology, we will say that the quasigroup $(\mathbb{E}_{(p-1)^2}^*, *)$ and all of its subquasigroups $(\mathbb{E}^*_{\nu}, *)$ are \emph{entropic quasigroups}. 
\begin{proposition}\label{Prop:G-nu-subgroupoid}
	If $(\mathbb{E}^*_{\nu}, *)$ is a subgroupoid of $(\mathbb{E}_{(p-1)^2}^*, *)$, then $\nu$ divides $(p-1)^2$ i.e. $\nu | (p-1)^2$ and $(\mathbb{E}^*_{\nu}, *)$ is a subquasigroup of $(\mathbb{E}_{(p-1)^2}^*, *)$. \qed
\end{proposition}

The following Proposition is a connection between the subgroups of the multiplicative group of a finite field, and the subgroupoids in the finite entropoid. It is a consequence of Proposition \ref{Prop:G-nu-subgroupoid} and the Lagrange's theorem for groups:
\begin{proposition}\label{Prop:Fields-Entropoid-connection-via-subgroups-in-F}
	Let $\mathbb{F}_p$ be the finite field over which an entropoid $\mathbb{E}_{p^2}$ is defined. Let $\Gamma \subseteq \mathbb{F}_p^*$ be a cyclic subgroup of order $|\Gamma|$ of the multiplicative group $\mathbb{F}_p^*$ and let $\gamma \neq  -\frac{a_3}{a_8}$ and $\gamma \neq -\frac{b_2}{b_7}$ is its generator, then $g = (\gamma, \gamma)$ is a generator of a subgroupoid $(\mathbb{E}^*_{\nu}, *)$ and $\nu$ divides $|\Gamma|^2$.
\end{proposition}


Next, for our subquasigroups $(\mathbb{E}^*_{\nu}, *)$ we will partially use the Smith terminology in his proposed Sylow theory for quasigroups \cite{smith2015sylow}. 

\begin{definition}
	Let $(\mathbb{E}^*_{\nu}, *)$ be a subquasigroup of $(\mathbb{E}^*_{(p-1)^2}, *)$ and let $(p-1)^2$ be represented as product of the powers of its prime factors i.e.$(p-1)^2 = 2^{e_1}q_1^{e_2}\ldots q_{k}^{e_k}$. We say that $(\mathbb{E}^*_{\nu}, *)$ is Sylow $q_i$-subquasigroup if $\nu = q_i^{e_i}$ for $i \in \{2, \ldots, k\}$. 
\end{definition}

Before we continue, let us give one example.
\begin{example}\label{Ex:GF(7)}
Let the finite entropoid be defined as $\mathbb{E}_{7^2}( a_3=6, a_8=3, b_2=3, b_7=4)$. In that case the operation $*$ becomes: $$x * y = (x_1, x_2)* (y_1, y_2) = (6 + 6 x_2 + 4 y_1 + 3 x_2 y_1, \ \ \ y_2 + 3 x_1 + 4 x_1 y_2). $$ All elements from $\mathbb{E}_{7^2}$ are presented in Table \ref{Table:Example01}.

From Corollary \ref{Cor:Basic-properties-of-multiplication} we get $\mathbf{0}_* = (5, 1)$ and $\mathbf{1}_* = (0, 6)$. The elements that are highlighted in Table \ref{Table:Example01} (elements in the row and column where the zero element $\mathbf{0}_*$ is positioned) are excluded from the multiplicative subgroupoid $\mathbb{E}_{6^2}^*$.


\begin{table}[h]
	\centering
	\includegraphics{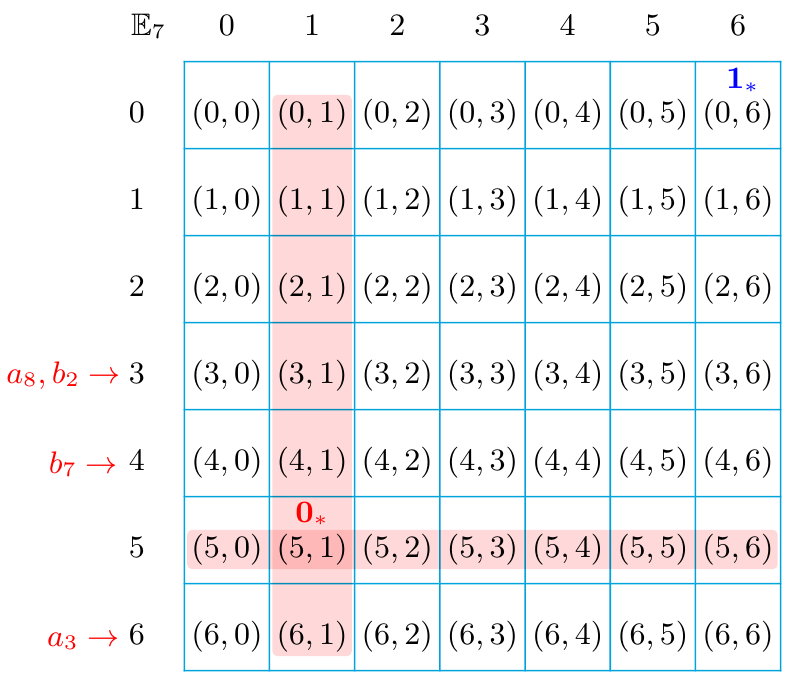}
	\caption{Representation of all 49 elements of $\mathbb{E}_{7^2}$.}\label{Table:Example01}
\end{table}

\begin{table}[h]
	\centering
	\includegraphics{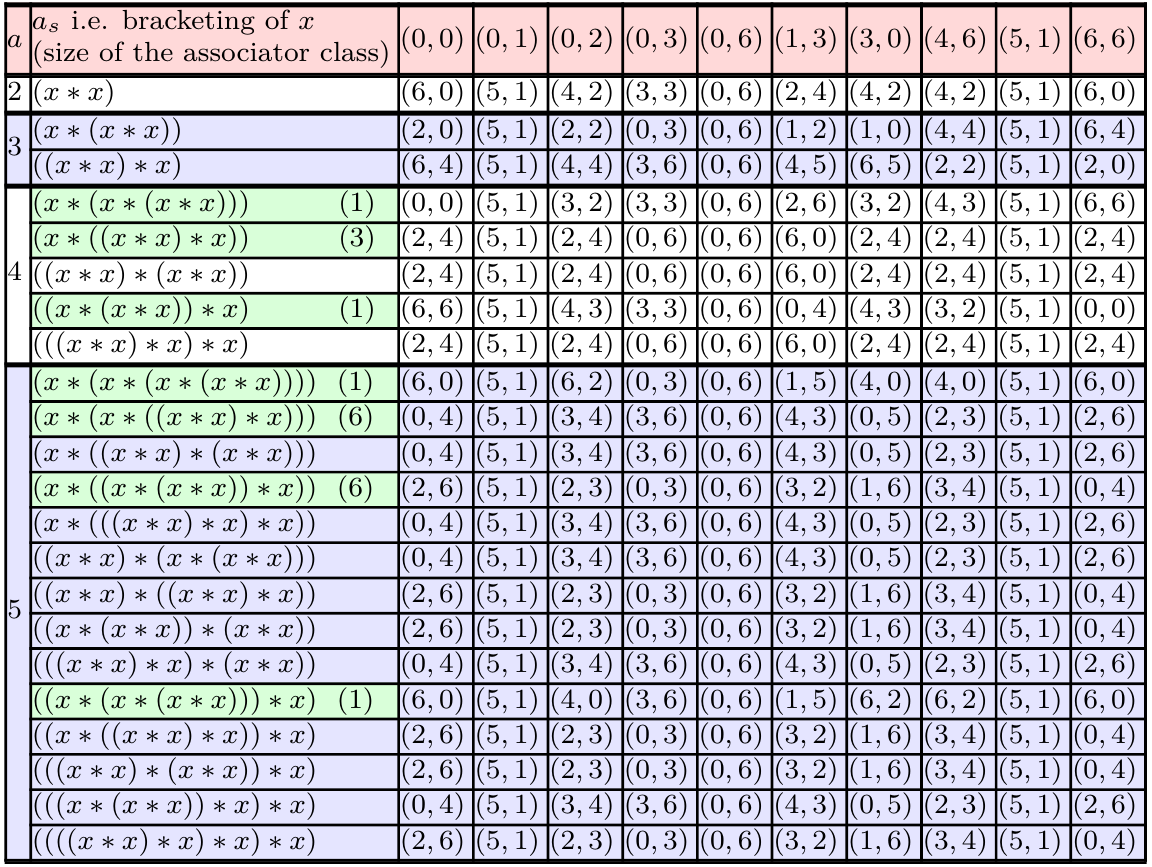}
	\caption{Results of rising different values $x$ to the powers from 2 to 5 with different bracketing patterns $a_s$. The green highlighted cells are for the representatives of a class of expressions that give the same result, and in the brackets is the size of that class.}    \label{Table:Bracketing-2-5}
\end{table}

\shorten{

} 

\begin{table}[ht]
	\centering
	\includegraphics{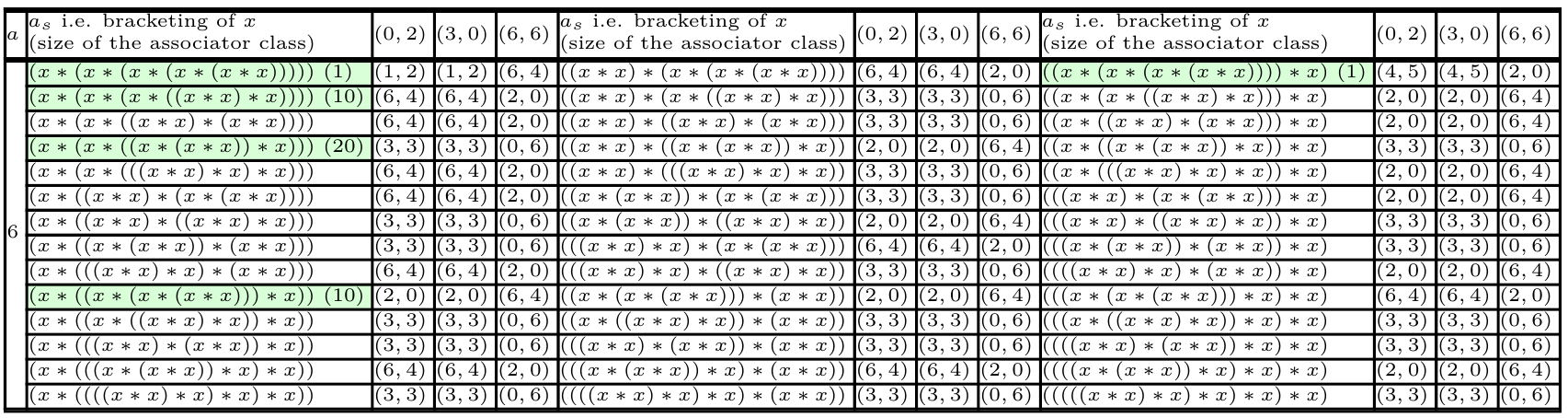}
	\caption{Results of rising different values $x$ to the power of 6 with different bracketing patterns $a_s$.}\label{Table:Bracketing-6}
\end{table}


We can make several observations from the values in Table \ref{Table:Bracketing-2-5} and Table \ref{Table:Bracketing-6}:
\begin{itemize}
	\item As we would expect, $\mathbf{0}_*^a = \mathbf{0}_*$ and $\mathbf{1}_*^a = \mathbf{1}_*$ for all $a$ and for all bracketing patterns $a_s$.
	\item Some values of $x$, raised to different patterns $a_s$ for the same value $a$ tend to produce more different values, while some generate less different values. For example, for $a=4$, for $x = (0, 2)$ there are three different values $(3,2)$, $(2, 4)$ and $(4, 3)$, while for $x = (0, 3)$ there are two different values $(3, 3)$ and $(0, 6)$.
	\item For any $x$, for a given fixed $a$, the number of different values generated with all patterns $a_s$ do not exceed $a-1$.
	\item When the number of different values is $a-1$, the distribution of the patterns over those $a-1$ values follows the distribution of Narayana numbers $N(n,k)$\cite[Sequence A001263]{sloane2007line}: 
	\begin{description}
		\item[$a = 4$:] 1, 3, 1
		\item[$a = 5$:] 1, 6, 6, 1
		\item[$a = 6$:] 1, 10, 20, 10, 1
	\end{description}
\end{itemize}

One can raise to different powers with different patterns all elements of $\mathbb{E}_{7^2}$ and can count the total number of generated elements. In Table \ref{Table:GeneratedSetsSize} we show the size of the generated sets with all elements of $\mathbb{E}_{7^2}$. The green highlighted elements with the value 36 generate the maximal multiplicative subgroupoid $(\mathbb{E}^*_{36},*)$. For example, $(0, 2)$ is the generator of $(\mathbb{E}^*_{36},*)$.

The yellow highlighted cells in Table \ref{Table:GeneratedSetsSize} are for the elements that do not belong to the multiplicative subgroupoid $((\mathbb{E}^*_{36},*)$.

Notice that the prime number $p = 7$ can be represented as $ p = 2 q + 1 = 2 * 3 + 1$, and that all elements of $\mathbb{E}_{7^2}$ can belong to different classes with cardinalities that are divisors of $(p - 1)^2 = 2^2 q^2 = 2^2 3^2$, i.e. the cardinalities are $\{1, 2, 3, 4, 6, 9, 18, 36\}$.
\begin{table}[h!]
	\centering
	\includegraphics{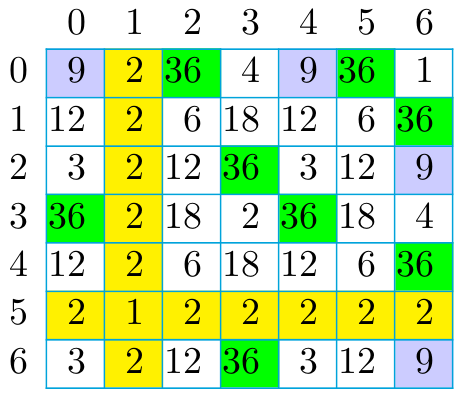}
	\caption{The size of the sets $\langle (x_1, x_2) \rangle$ for $x = (x_1, x_2) \in \mathbb{E}_{7^2}$ where $x_1$ represents the row number and $x_2$ represents the column number.}\label{Table:GeneratedSetsSize}
\end{table}
\qed
\end{example}

\forestset{
	dot tree/.style={
		/tikz/>=Latex,
		for tree={
			anchor=center,
			base=bottom,
			inner sep=1pt,
			fill=yellow,
			draw,
			circle,
			calign=fixed edge angles,
		},
		baseline,
		before computing xy={
			where n children>=4{
				tempcounta/.option=n children,
				tempdima/.option=!1.s,
				tempdimb/.option=!l.s,
				tempdimb-/.register=tempdima,
				tempdimc/.process={RRw2+P {tempcounta}{tempdimb}{##2/(##1-1)}},
				for children={
					if={>On>OR<&{n}{1}{n}{tempcounta}}{
						s/.register=tempdima,
						s+/.process={ORw2+P  {n} {tempdimc} {(##1-1)*##2} }
					}{},
				},
			}{},
		},
	},
	dot tree spread/.style={
		dot tree,
		for tree={fit=rectangle},
	},
	add arrow/.style={
		tikz+={
			\draw [thick, blue!15!gray]  (current bounding box.east) ++(2.5mm,0) edge [->] ++(10mm,0) ++(2.5mm,0) coordinate (o);
		}
	}
}
\begin{table}[h!]
	\centering
	\small
	\begin{tabular}{|c||c|c|c|}
		\hline
		$a$ & $i=1$, $N(3, 1) = 1$  & $i=2$, $N(3, 2) = 3$ & $i=3$, $N(3, 3) = 1$ \\
		\hline
		\multirow{5}{1em}{ } &
		\scalebox{0.5}{
			\begin{forest}
				dot tree,
				where n children=0{
					before computing xy={l*=1.0, s*=1.0},
					edge+={densely dashed},
				}{},
				[[$x$][[$x$][[$x$][$x$]]]]
		\end{forest} }
		&
		\scalebox{0.5}{\begin{forest}
				dot tree,
				where n children=0{
					before computing xy={l*=1.0, s*=1.0},
					edge+={densely dashed},
				}{},
				[[$x$][[[$x$][$x$]][$x$]]]
		\end{forest}}
		&
		\scalebox{0.5}{\begin{forest}
				dot tree,
				where n children=0{
					before computing xy={l*=1.0, s*=1.0},
					edge+={densely dashed},
				}{},
				[[[$x$][[$x$][$x$]]][$x$]]	
		\end{forest}} \\
		& $(x*(x*(x*x)))$  &  $(x*((x*x)*x))$ &  $((x*(x*x))*x)$ \\
		& &
		\scalebox{0.5}{\begin{forest}
				dot tree,
				where n children=0{
					before computing xy={l*=1.0, s*=1.0},
					edge+={densely dashed},
				}{},
				[[[$x$][$x$]][[$x$][$x$]]]
		\end{forest}}
		& \\
		4 & &  $((x*x)*(x*x))$ & \\
		& &
		\scalebox{0.5}{\begin{forest}
				dot tree,
				where n children=0{
					before computing xy={l*=1.0, s*=1.0},
					edge+={densely dashed},
				}{},
				[[[[$x$][$x$]][$x$]][$x$]]
		\end{forest}}
		& 	\\ 
		& &  $(((x*x)*x)*x)$ & \\
		\hline
	\end{tabular}
	\caption{Representation of all possible non-associative powers of $a=4$ as planar binary trees, and the corresponding clastering in subsets with Narayana numbers of elements. }\label{Table:PBT(4)}
\end{table}

\begin{definition}\label{Def:AssociativeClassRepresentatives}
	Let $2 \leq a$ be an integer, and let $x$ be an element in $G$. An ordered list $R_a(x)$ of $a-1$ bracketing shapes is called the list of associative class representatives and is defined as: 
	\begin{align}\label{Eq:Representatives}
		R_a(x) = [ R_a(x)[0], R_a(x)[1], \ldots, & R_a(x)[a-3], R_a(x)[a-2] ], \text{where} \nonumber \\
		R_a(x)[0]   & = x_{(x, a-1)\texttt{to-left}}, \nonumber  \\
		 R_a(x)[1]	& = (x_{(x, 1)\texttt{to-left}} * x)_{(x, a - 3)\texttt{to-left}}, \nonumber  \\
		 R_a(x)[2]	& = (x_{(x, 2)\texttt{to-left}} * x)_{(x, a - 4)\texttt{to-left}}, \nonumber  \\
					& \ldots, \\
		 R_a(x)[a-3]& = (x_{(x, a-3)\texttt{to-left}} * x)_{(x, 0)\texttt{to-left}}, \nonumber  \\
		 R_a(x)[a-2]& = (x_{(x, a-2)\texttt{to-left}} * x), \nonumber
	\end{align}
\end{definition}
\textbf{Note:} We use the zero-based indexing style for the list members.

\begin{lemma}\label{Lemma:Recurrent-Relations-For-Representatives}
	The bracketing shapes for $R_a(x)$ and $R_{a+1}(x)$ are related with the following recurrent relations:
	\begin{align}\label{Eq:Recurrent-Relations-For-Representatives}
		R_{a+1}(x)[0]   & = (x*R_{a}(x)[0]), \nonumber  \\
		R_{a+1}(x)[1]	& = (x*R_{a}(x)[1]), \nonumber  \\
		R_{a+1}(x)[2]	& = (x*R_{a}(x)[2]), \nonumber  \\
		& \ldots, \\
		R_{a+1}(x)[a-2]& = (x*R_{a}(x)[a-2]), \nonumber  \\
		R_{a+1}(x)[a-1]  & = (R_{a}(x)[0]*x), \nonumber		
	\end{align}
\end{lemma}
\begin{proof}
	We will prove the lemma with induction by $a$. Let us first note that $R_2(x) = [ R_2(x)[0] \equiv (x*x) ]$. For $a=3$ we have $R_3(x) = [ R_3(x)[0], R_3(x)[1] ]$, where $R_3(x)[0] = x_{(x, 3-1)\texttt{to-left}}  = (x*(x*x)) = (x*R_2(x))$, and $R_3(x)[1] = (x_{(x, 3-2)\texttt{to-left}} * x) = ((x*x)*x) =  (R_{2}(x)[0]*x)$.
	
	If we suppose that the recurrent relations (\ref{Eq:Recurrent-Relations-For-Representatives}) are true for $a$, then for $a+1$ we have 
	\begin{itemize}
		\item $R_{a+1}(x)[0] \stackrel{def}{=} x_{(x, a)\texttt{to-left}} = (x*x_{(x, a-1)\texttt{to-left}}) = (x*R_{a}(x)[0])$,
		\item $R_{a+1}(x)[1] \stackrel{def}{=} (x_{(x, 1)\texttt{to-left}} * x)_{(x, a + 1 - 3)\texttt{to-left}} = (x*x_{(x, a-2)\texttt{to-left}}) = (x*R_{a}(x)[1])$,
		\item $\ldots$,
		\item $R_{a+1}(x)[a-2] \stackrel{def}{=} (x_{(x, a-2)\texttt{to-left}} * x)_{(x, 0)\texttt{to-left}} = (x*R_{a}(x)[a-2]) $,
		\item $R_{a+1}(x)[a-1] \stackrel{def}{=} (x_{(x, a-1)\texttt{to-left}} * x) = (R_{a}(x)[0]*x)$.
	\end{itemize}
\end{proof}

\begin{example}\label{Ex:PBTrees}
	Let us present the bracketing shapes highlighted in green in Table \ref{Table:Bracketing-2-5} and in Table \ref{Table:Bracketing-6}, in Example \ref{Ex:GF(7)} with the notation introduced in Definition \ref{Def:AssociativeClassRepresentatives}.
	
	\begin{description}
		\item[$a = 4$ : ] Out of 5 bracketing shapes, the following 3 are highlighted:
			\begin{description}
				\item[$i = 1$:] $(x*(x*(x*x))) = x_{(x, 3)\texttt{to-left}}$
				\item[$i = 2$:] $(x*((x*x)*x)) = (x_{(x, 1)\texttt{to-left}} * x)_{(x, 0)\texttt{to-left}}$
				\item[$i = 3$:] $((x*(x*x))*x) = (x_{(x, 2)\texttt{to-left}} * x)$
			\end{description}
		\item[$a = 5$ : ] Out of 14 bracketing shapes, the following 4 are highlighted:
			\begin{description}
				\item[$i = 1$:] $(x*(x*(x*(x*x)))) = x_{(x, 4)\texttt{to-left}}$
				\item[$i = 2$:] $(x*(x*((x*x)*x))) = (x_{(x, 1)\texttt{to-left}} * x)_{(x, 2)\texttt{to-left}}$
				\item[$i = 3$:] $(x*((x*(x*x))*x)) = (x_{(x, 2)\texttt{to-left}} * x)_{(x, 1)\texttt{to-left}}$
				\item[$i = 4$:] $((x*(x*(x*x)))*x) = (x_{(x, 3)\texttt{to-left}} * x)$
			\end{description}
		\item[$a = 6$ : ] Out of 42 bracketing shapes, the following 5 are highlighted:
			\begin{description}
				\item[$i = 1$:] $(x*(x*(x*(x*(x*x))))) = x_{(x, 5)\texttt{to-left}}$
				\item[$i = 2$:] $(x*(x*(x*((x*x)*x)))) = (x_{(x, 1)\texttt{to-left}} * x)_{(x, 3)\texttt{to-left}}$
				\item[$i = 3$:] $(x*(x*((x*(x*x))*x))) = (x_{(x, 2)\texttt{to-left}} * x)_{(x, 2)\texttt{to-left}}$
				\item[$i = 4$:] $(x*((x*(x*(x*x)))*x)) = (x_{(x, 3)\texttt{to-left}} * x)_{(x, 1)\texttt{to-left}}$
				\item[$i = 5$:] $((x*(x*(x*(x*x))))*x) = (x_{(x, 4)\texttt{to-left}} * x)$
			\end{description}	
	\end{description} 

	Let us also present in Table \ref{Table:PBT(4)} and Table \ref{Table:PBT(5)} the bracketing shapes for $a=4$ and $a=5$ as planar binary trees (for $a=6$ it would be an impractically tall table with 20 trees in one column). We see that the first rows in the tables are exactly the highlighted green shapes from Table \ref{Table:Bracketing-2-5}.

\begin{table}[h]
	\centering
	\small
	\begin{tabular}{|c||c|c|c|c|}
		\hline
		$a$ & $i=1$, $N(4, 1) = 1$  & $i=2$, $N(4, 2) = 6$ & $i=3$, $N(4, 3) = 6$ & $i=4$, $N(4, 4) = 1$  \\
		\hline
		\multirow{5}{1em}{ } &
		\scalebox{0.5}{
			\begin{forest}
				dot tree,
				where n children=0{
					before computing xy={l*=1.0, s*=1.0},
					edge+={densely dashed},
				}{},
				[[$x$][[$x$][[$x$][[$x$][$x$]]]]]
		\end{forest} }
		&
		\scalebox{0.5}{\begin{forest}
				dot tree,
				where n children=0{
					before computing xy={l*=1.0, s*=1.0},
					edge+={densely dashed},
				}{},
				[[$x$][[$x$][[[$x$][$x$]][$x$]]]]
		\end{forest}}
		&
		\scalebox{0.5}{\begin{forest}
				dot tree,
				where n children=0{
					before computing xy={l*=1.0, s*=1.0},
					edge+={densely dashed},
				}{},
				[[$x$][[[$x$][[$x$][$x$]]][$x$]]]	
		\end{forest}} 
		&
		\scalebox{0.5}{\begin{forest}
				dot tree,
				where n children=0{
					before computing xy={l*=1.0, s*=1.0},
					edge+={densely dashed},
				}{},
				[[[$x$][[$x$][[$x$][$x$]]]][$x$]]	
		\end{forest}} 
		
		\\
		& {\scriptsize $(x*(x*(x*(x*x))))$}  &  {\scriptsize $(x*(x*((x*x)*x)))$} &  {\scriptsize $(x*((x*(x*x))*x))$} & {\scriptsize $((x*(x*(x*x)))*x)$}\\
		& &
		\scalebox{0.5}{\begin{forest}
				dot tree,
				where n children=0{
					before computing xy={l*=1.0, s*=1.0},
					edge+={densely dashed},
				}{},
				[[$x$][[[$x$][$x$]][[$x$][$x$]]]]
		\end{forest}}
		& 
		\scalebox{0.5}{\begin{forest}
				dot tree,
				where n children=0{
					before computing xy={l*=1.0, s*=1.0},
					edge+={densely dashed},
				}{},
				[[[$x$][$x$]][[[$x$][$x$]][$x$]]]
		\end{forest}}	    
		& \\
		& &  {\scriptsize $(x*((x*x)*(x*x)))$} & {\scriptsize $((x*x)*((x*x)*x))$}    & \\
		& &
		\scalebox{0.5}{\begin{forest}
				dot tree,
				where n children=0{
					before computing xy={l*=1.0, s*=1.0},
					edge+={densely dashed},
				}{},
				[[$x$][[[[$x$][$x$]][$x$]][$x$]]]
		\end{forest}}
		& 
		\scalebox{0.5}{\begin{forest}
				dot tree,
				where n children=0{
					before computing xy={l*=1.0, s*=1.0},
					edge+={densely dashed},
				}{},
				[[[$x$][[$x$][$x$]]][[$x$][$x$]]]
		\end{forest}}
		&	\\ 
		5 & &  {\scriptsize $(x*(((x*x)*x)*x))$} & {\scriptsize $((x*(x*x))*(x*x))$} & \\
		& &
		\scalebox{0.5}{\begin{forest}
				dot tree,
				where n children=0{
					before computing xy={l*=1.0, s*=1.0},
					edge+={densely dashed},
				}{},
				[[[$x$][$x$]][[$x$][[$x$][$x$]]]]
		\end{forest}}
		& 
		\scalebox{0.5}{\begin{forest}
				dot tree,
				where n children=0{
					before computing xy={l*=1.0, s*=1.0},
					edge+={densely dashed},
				}{},
				[[[$x$][[[$x$][$x$]][$x$]]][$x$]]
		\end{forest}}
		&	\\ 
		& &  {\scriptsize $((x*x)*(x*(x*x)))$} & {\scriptsize $((x*((x*x)*x))*x)$} & \\
		& &
		\scalebox{0.5}{\begin{forest}
				dot tree,
				where n children=0{
					before computing xy={l*=1.0, s*=1.0},
					edge+={densely dashed},
				}{},
				[[[[$x$][$x$]][$x$]][[$x$][$x$]]]
		\end{forest}}
		& 
		\scalebox{0.5}{\begin{forest}
				dot tree,
				where n children=0{
					before computing xy={l*=1.0, s*=1.0},
					edge+={densely dashed},
				}{},
				[[[[$x$][$x$]][[$x$][$x$]]][$x$]]
		\end{forest}}
		&	\\ 
		& &  {\scriptsize $(((x*x)*x)*(x*x))$} & {\scriptsize $(((x*x)*(x*x))*x)$} & \\
		& &
		\scalebox{0.5}{\begin{forest}
				dot tree,
				where n children=0{
					before computing xy={l*=1.0, s*=1.0},
					edge+={densely dashed},
				}{},
				[[[[$x$][[$x$][$x$]]][$x$]][$x$]]
		\end{forest}}
		& 
		\scalebox{0.5}{\begin{forest}
				dot tree,
				where n children=0{
					before computing xy={l*=1.0, s*=1.0},
					edge+={densely dashed},
				}{},
				[[[[[$x$][$x$]][$x$]][$x$]][$x$]]
		\end{forest}}
		&	\\ 
		& &  {\scriptsize $(((x*(x*x))*x)*x)$} & {\scriptsize $((((x*x)*x)*x)*x)$} & \\
		\hline
	\end{tabular}
	\caption{Representation of all possible non-associative powers of $a=5$ as planar binary trees, and the corresponding clastering in subsets with Narayana number of elements. }\label{Table:PBT(5)}
\end{table}

	\qed
\end{example}

Since for any fixed integer $a$, the set $S_a(x)$ of all bracketing shapes $\mathbf{A}= (a, a_s)$ have $C_{a-1} = \frac{1}{a}\binom{2a - 2}{a-1}$ elements (equation (\ref{Eq:Catalan})), for a generic non-commutative and non-associative groupoid $(G, *)$, with a multiplicative operation $*$ we would expect that for some elements $x\in G$ there would be up to $C_{a-1}$ different power values in the set $\{x ^\mathbf{A} \}$. However, for entropic groupoids $(\mathbb{E}^*_{p^2}$ defined by Definition \ref{Def:Design-Criteria-for-the-operation}  and by the equation (\ref{Eq:OperationStar}) that is not the case, and the size of the set $\{x ^\mathbf{A} \}$ is limited to $a-1$ as we see it in Example \ref{Ex:GF(7)} and Example \ref{Ex:PBTrees}. 

\begin{definition}
    Let  $\mathbb{E}_{p^2}$ be a given finite entropoid. We define an integer $b_{max}$ by the following expression:
    \begin{equation}
        b_{max} = \min \left\lbrace  b \ | \ \frac{1}{b}\binom{2b - 2}{b-1} > (p - 1)^2 \right\rbrace .
    \end{equation}
\end{definition}

The value $b_{max}$ is the smallest integer for which the Catalan number $C_{b-1} = \frac{1}{b}\binom{2b - 2}{b-1}$ surpasses the number of elements in the maximal multiplicative subgroupoid $(G^*, g)$. We will need it in the proof of the next theorem.

\begin{theorem}[Equivalent classes and their representatives]\label{Thm:Equivalent-Classes}
	Let  $\mathbb{E}_{p^2}$ be given, and let $g$ is a generator of its maximal multiplicative subgroupoid $(\mathbb{E}^*_{(p-1)^2}, *)$, where $*$ is defined by Definition \ref{Def:Design-Criteria-for-the-operation} and by the equation (\ref{Eq:OperationStar}). For every $3 \leq a < b_{max}$, evaluating the bracketing shapes of the set $S_a(g)$ gives a partitioning in $a-1$ equivalent classes $S_{1,a(g)}$, $S_{2,a(g)}$, $\ldots$, $S_{a-2,a(g)}$, $S_{a-1,a(g)}$, whose corresponding representatives are given by the corresponding elements of $R_a(g)$ defined by the relation (\ref{Eq:Representatives}) and the cardinality  of the sets $S_{i,a(g)}$ for $i = 1, \ldots, a-1$ is given by the following expression 
	\begin{equation}\label{Eq:Narayana_S_a(g)}
		|S_{i,a(g)}| = N(a-1, i) = \frac{1}{a-1}\binom{a-1}{i}\binom{a-1}{i-1},
	\end{equation}
	where $N(a-1, i)$ are the Narayana numbers.
\end{theorem}
\begin{proof}
	We will prove the theorem by induction on $a$ for the range $3 \leq a < b_{max}$. For $a \geq b_{max}$ there are not enough elements to be classified in classes which total number surpass  $(p-1)^2$ (as it is the number of elements of $(\mathbb{E}^*_{(p-1)^2}, *)$).
    
     For $a=3$ the theorem is trivially true since there are just $a-1 = 2$ associative shapes. The first non-trivial case is for $a=4$. Taking $g = (g_1, g_2)$, and using the definition for the operation $*$ given in equation (\ref{Eq:OperationStar}), by basic symbolic and algebraic expression replacements we get that for all three shapes the end result is: 
	\begin{eqnarray*}
		 & & R_4(g)[1] \equiv (g*((g*g)*g)) = ((g*g)*(g*g)) = (((g*g)*g)*g)= \\ 
		 & = \Bigg(& \frac{a_8^3 g_1^2 \left(b_7 g_2+b_2\right)^2+2 a_3 a_8^2 g_1 \left(b_7 g_2+b_2\right)^2+a_3
		 	\left(2 a_3 a_8 b_7 b_2 g_2+a_3 a_8 b_7^2 g_2^2+a_3 a_8 b_2^2-b_7\right)}{a_8 b_7},\\
	 	& & \frac{a_3^2 b_7^3
		 	g_2^2+2 a_3^2 b_2 b_7^2 g_2+a_8^2 b_7 g_1^2 \left(b_7 g_2+b_2\right)^2+2 a_3 a_8 b_7 g_1 \left(b_7
		 	g_2+b_2\right)^2+b_2 \left(a_3^2 b_2 b_7-a_8\right)}{a_8 b_7}\Bigg) 
	\end{eqnarray*}
	We can also check that $R_4(g)[0] \equiv (g*(g*(g*g))) \neq (g*((g*g)*g)) \equiv R_4(g)[1]$, $R_4(g)[2] \equiv ((g*(g*g))*g) \neq (g*((g*g)*g)) \equiv R_4(g)[1]$, and $R_4(g)[0] \equiv (g*(g*(g*g))) \neq ((g*(g*g))*g) \equiv R_4(g)[2]$. Thus for the case $a=4$ we have $S_{1,4(g)} = \{R_4(g)[0] \}$, $S_{2,4(g)} = \{R_4(g)[1] \equiv (g*((g*g)*g)), ((g*g)*(g*g)), (((g*g)*g)*g)  \}$ and $S_{3,4(g)} = \{R_4(g)[2] \}$.
	
	Let us now suppose that the claims of the theorem are true for $a$, i.e. that the expressions (\ref{Eq:Partitioning-S_a(g)}) and (\ref{Eq:Narayana_S_a(g)}) hold. Then, for $a+1$ we have the set $S_{a+1}(g)$ that has $C_{a} = \frac{1}{a+1}\binom{2a}{a}$ elements, where $C_{a}$ is the $(a)$-th Catalan number. We can now partition the set $S_{a+1}(g)$ as follows:
	 \begin{equation}\label{Eq:Partitioning-S_a(g)}
		S_{a+1}(g) = S_{1,a(g)} \bigcup S_{2,a(g)} \bigcup \ldots \bigcup S_{a-1,a(g)} \bigcup S_{a,a(g)},
	\end{equation}
	 where $x_{(x, a-1)\texttt{to-left}} \in S_{1,a(g)}$,  $(x_{(x, 1)\texttt{to-left}} * x)_{(x, a - 3)\texttt{to-left}} \in S_{2,a(g)}$, $\ldots$, $(x_{(x, a-3)\texttt{to-left}} * x)_{(x, 1)\texttt{to-left}} \in S_{a-2,a(g)}$ and $((x_{(x, a-2)\texttt{to-left}} * x) \in S_{a-1,a(g)}$. 	
	
	 The conditions that $g$ is a generator of the $(\mathbb{E}^*_{(p-1)^2}, *)$, and that $*$ is non-commutative and non-associative operation are essential for obtaining that the element $R_{a+1}(x)[a-1] = (R_{a}(x)[0]*x)$ of the subset $S_{a,a(g)}$ is different from the element $R_{a+1}(x)[0] = (x*R_{a}(x)[0])$ of the subset $S_{1,a(g)}$.
	 
	 Then,if we represent the Catalan number $ C_a$ as a sum of Narayana numbers $N(a, i)$:
	 $$ C_a= N(a, 1) + N(a, 2) + \ldots + N(a, a-1) + N(a, a), $$
	 and if we use the recurrent relations for the representatives of the equivalent classes given in Lemma \ref{Lemma:Recurrent-Relations-For-Representatives} to get expressions for the representatives of the partitions of $S_{a+1}(g)$, we obtain that the relations (\ref{Eq:Partitioning-S_a(g)}) and (\ref{Eq:Narayana_S_a(g)}) hold also for $a+1$. 
\end{proof}

\begin{openroblem}
	In the proof of the Theorem \ref{Thm:Equivalent-Classes}, as a jump-start for the induction we used the concrete instance for the entropic operation $*$ defined by the equation (\ref{Eq:OperationStar}). We do not know does the same partitioning in equivalent classes hold for general entropic operations $*$, and it would be an interesting mathematical problem to further investigate.
\end{openroblem}

\subsection{Succinct Notation for Exponentially Large Bracketing Shapes}
\textbf{Notation:} In the rest of the text we assume that we work in finite entropoid $\mathbb{E}_{p^2}$ and $g\in \mathbb{E}_{p^2}$ is a generator of some multiplicative subgroupoid $(\mathbb{E}^*_{\nu}, *)$, where $*$ is defined by Definition \ref{Def:Design-Criteria-for-the-operation} and by the equation (\ref{Eq:OperationStar}). In Definition \ref{Def:AssociativeClassRepresentatives}, Lemma \ref{Lemma:Recurrent-Relations-For-Representatives} and Theorem \ref{Thm:Equivalent-Classes} we use the symbol $a$ to represent any number (possibly exponentially large) of multiplicative factors in the process of rising to the power of $a$. In our pursuit for succinct representation of exponentially large indices, we will keep the symbol $a$, but for smaller bracketing patterns of size $\mathfrak{b}$ the role of the symbol $a$ in Definition \ref{Def:AssociativeClassRepresentatives}, Lemma \ref{Lemma:Recurrent-Relations-For-Representatives} and Theorem \ref{Thm:Equivalent-Classes} will be replaced by $\mathfrak{b} \geq 2$. Further, we will denote the indices either as pairs $(\mathbf{A}, \mathfrak{b})$ or as triplets $(a, a_s, \mathfrak{b})$ meaning $(\mathbf{A}, \mathfrak{b}) = (a, a_s, \mathfrak{b}) $. The set of all power indices will be denoted by $\mathbb{L}$ in resemblance to the logarithmetic of power indices $(L(G), \mathbf{+}, \mathbf{\times})$ $)$. We will write $(\mathbf{A}, \mathfrak{b})  \xleftarrow{\mathbf{r}} \mathbb{L}$ to denote that $(\mathbf{A}, \mathfrak{b})$ was chosen uniformly at random from the set $\mathbb{L}$ if $\mathfrak{b} \xleftarrow{\mathbf{r}} \mathbb{Z}_{>1}$, $a \xleftarrow{\mathbf{r}} \mathbb{Z}_p^*$, and $a_s  \xleftarrow{\mathbf{r}} \mathbb{L}[a, \mathfrak{b}]$ where $\mathbb{L}[a, \mathfrak{b}]$ is the set of all bracketing shapes defined in Definition \ref{Def:Power-indices}.

\begin{definition}[Succinct power indices]\label{Def:Power-indices}
	Let integer $\mathfrak{b} \geq 2$ be a base. The power index $ (\mathbf{A}, \mathfrak{b}) = (a, a_s, \mathfrak{b})$ is defined as a triplet consisting of two lists and an integer $\mathfrak{b}$. Let $a \in \mathbb{Z}^+$ be represented in base $\mathfrak{b}$ as $a = A_0 + A_1 \mathfrak{b} + \ldots + A_k \mathfrak{b}^k$, or in little-endian notation with the list of digits $0 \leq A_i \leq \mathfrak{b}-1$, as $a = [A_0, A_1, \ldots, A_k]_\mathfrak{b}$. Let the bracketing pattern $a_s$ be represented with a list of digits $a_s = [P_0, P_1, \ldots, P_k]_{\mathfrak{b}-1}$ where 
    $0 \leq P_i \leq \mathfrak{b}-2$ for $i = 0, \ldots, k$. For base $\mathfrak{b}$ and $a = [A_0, A_1, \ldots, A_k]_\mathfrak{b} \in \mathbb{Z}^+$ we denote the set of all possible patterns $a_s$ by
    \begin{equation}
    	\mathbb{L}[a, \mathfrak{b}] = \{ a_s \ | a_s = [P_0, \ldots, P_k]_{\mathfrak{b}-1}, \text{ for all } P_j \in \mathbb{Z}_{\mathfrak{b}-1}, j = 0, \ldots, k  \}.
    \end{equation}
\end{definition}
    
\begin{definition}[Non-associative and non-commutative exponentiation]\label{Def:Raising-to-power}
	For any $x \in \mathbb{E}_{p}$ we define $x^{(\mathbf{A}, \mathfrak{b})}$ as follows:
	\begin{align}
		w_0 = & x, \nonumber \\
		\label{Eq:Raising-to-power-v1}
		w_i = & R_b(w_{i-1})[P_i], \text{\ for \ } i = 1, \ldots, k, \\
		j = & \text{index of the first nonzero digit } A_j \nonumber \\
		x_j = & \left\{
		\begin{array}{ll}
			w_j & \text{, if } A_j = 1,\\
			\label{Eq:Raising-to-power-v2}
			R_{A_j}(w_j)[ P_j \mod (A_j - 1) ] & \text{, if } A_j > 1.
		\end{array} \right.\\
		\text{for } i=j+1, \ldots, k & \nonumber \\
		\label{Eq:Raising-to-power-v3}
		x_i = & \left\{
		\begin{array}{ll}
			x_{i-1} & \text{, if } A_i = 0,\\
			w_i * x_{i-1} & \text{, if } A_i = 1 \text{ and } P_{i-1} \text{ is even,}\\
			x_{i-1} * w_i & \text{, if } A_i = 1 \text{ and } P_{i-1} \text{ is odd,}\\
			R_{A_i}(w_i)[ P_i \mod (A_i - 1) ] * x_{i-1} & \text{, if } A_i > 1  \text{ and } P_{i-1}  \text{ is even,}\\
			x_{i-1} * R_{A_i}(w_i)[ P_i \mod (A_i - 1) ] & \text{, if } A_i > 1  \text{ and } P_{i-1}  \text{ is odd,}		
		\end{array} \right.\\
		\label{Eq:Raising-to-power-v4}
		x^{(\mathbf{A}, \mathfrak{b})} = & x_k.
	\end{align}	
\end{definition}


\begin{proposition}\label{Prop:Well-defiined-powers}
	For every base $\mathfrak{b} \geq 2$, every $x \in \mathbb{E}_{p^2}$, every $a \in \mathbb{Z}^+$ and every bracketing pattern $a_s = [P_0, P_1, \ldots, P_k]_{\mathfrak{b}-1}$, the value $x^{(\mathbf{A}, \mathfrak{b})}$ is a product of $a$ multiplications of $x$ 
	$$x^{(\mathbf{A}, \mathfrak{b})} = x^{(a, a_s, \mathfrak{b})} = \underbrace{(x * \ldots (x*x) \ldots )}_{a \text{\ copies of } x}. $$ If $O^*_{\mathfrak{b}}$ denotes the number of operations $*$ used to compute the result $x^{(\mathbf{A}, \mathfrak{b})}$, then its value is given by the following expression 
	\begin{equation}\label{Eq:Upper-bound-number-of-multiplications-in-power}
		O^*_{\mathfrak{b}} = k (\mathfrak{b}-1) - 1 + \sum_{A_i \neq 0}A_i  \ \ .
	\end{equation}
\end{proposition}
\begin{proof}
	For the first part, let us define a counter $M(i)$ for $i = 0, 1, \ldots, k$, that counts the number of times $x$ was multiplied in the procedure described in equations (\ref{Eq:Raising-to-power-v1}) - (\ref{Eq:Raising-to-power-v4}) for a number $m_i = [A_0, A_1, \ldots, A_i]_\mathfrak{b}$ consisting of the first $i$ digits of $a$. 
	
	Let us first notice that every value $w_i$, for $i = 0, 1, \ldots, k$ in (\ref{Eq:Raising-to-power-v1}) is a result of $\mathfrak{b}^k$ multiplications of $x$. Now, our counting should start at (\ref{Eq:Raising-to-power-v2}) with the value $x_j$ where  $j$ is the index of the first nonzero digit $A_j$. This means that if $j=0$, $M(0) = A_0 = m_0$, and if $j>0$, $M(j) = A_j \mathfrak{b}^j = m_j$.  Then, for every next digit $A_i$ for $i=j+1, \ldots, k$ we have that $M(j) = M(j-1) + A_j \mathfrak{b}^j$, where the factor $\mathfrak{b}^j$ comes from the fact that $w_j$ is used either as $1 \mathfrak{b}^j$ or as $A_j \mathfrak{b}^j$ if $A_j > 1$ in the part $ [ P_i \mod (A_i - 1) ] $ of (\ref{Eq:Raising-to-power-v3}). The final result is that $M(k) = \sum_{i=0}^{k}A_j \mathfrak{b}^j = a$.
	
	The second part can be proved by counting the number of multiplications performed in parts (\ref{Eq:Raising-to-power-v1}) and in parts (\ref{Eq:Raising-to-power-v2})-(\ref{Eq:Raising-to-power-v3}). In the expression (\ref{Eq:Raising-to-power-v1}) we have a fixed value $\mathfrak{b}$ - so total number of multiplications is exactly $k (\mathfrak{b}-1)$. Next, in (\ref{Eq:Raising-to-power-v2}) we have either $0$ multiplications (if $A_j = 0$), or $A_j - 1$ multiplications (if $A_j > 1$). Finally, in (\ref{Eq:Raising-to-power-v3}), for $i=j+1, \ldots, k$ we apply $1$ multiplication with $x_{i-1}$ and $A_i - 1$ multiplications in the expression $R_{A_i}(w_i)[ P_i \mod (A_i - 1) ]$, which proves the equation (\ref{Eq:Upper-bound-number-of-multiplications-in-power}).
\end{proof}

Proposition \ref{Prop:Well-defiined-powers} ensures that Definition \ref{Def:Raising-to-power} defines a succinct notation for exponentially large power indices with their bracketing shapes since:
\begin{enumerate}
    \item It is consistent procedure of rising to a power where the number of times that $x$ is multiplied is exactly $a$, and
    \item An efficient procedure where the number of performed operations $*$ is $O(\mathfrak{b} \log_\mathfrak{b} a)$.
\end{enumerate}

The contribution of the non-associativity of the operation $*$ to the final result comes from the expressions $w_i = R_b(w_{i-1})[P_i]$ of (\ref{Eq:Raising-to-power-v1}), and from $R_{A_j}(w_j)[ P_j \mod (A_j - 1) ] $ of (\ref{Eq:Raising-to-power-v2}) and $R_{A_i}(w_i)[ P_i \mod (A_i - 1) ]$ of (\ref{Eq:Raising-to-power-v3}). On the other hand, the contribution of the non-commutativity of the operation $*$ comes from the multiplication by $x_{i-1}$ from left or from right in (\ref{Eq:Raising-to-power-v3}).

The next Theorem is a direct consequence from Theorem \ref{Thm:Basic-theorem-for-powers-in-entropic-groupoids} and the Proposition \ref{Prop:Well-defiined-powers}.
\begin{theorem}[Basic theorem for exponentiation]\label{Thm:Basic-theorem-for-risinng-to-power}
	Let $x, y \in \mathbb{E}_{(p-1)^2}^*$. For every $(\mathbf{A},{\mathfrak{b}_1}), (\mathbf{B},{\mathfrak{b}_2}) \xleftarrow{\mathbf{r}} \mathbb{L}$
	\begin{equation}\label{Eq:Power-of-product-of-x-and-y-mixed-bases}
		(x*y)^{(\mathbf{A},{\mathfrak{b}_1})} = x^{(\mathbf{A},{\mathfrak{b}_1})} * y^{(\mathbf{A},{\mathfrak{b}_1})},
	\end{equation}
	and 
	\begin{equation}\label{Eq:Power-of-power-of-x-succinct}
		x^{(\mathbf{A},{\mathfrak{b}_1}) (\mathbf{B},{\mathfrak{b}_2})} = \big(x^{(\mathbf{A},{\mathfrak{b}_1})}\big)^{(\mathbf{B},{\mathfrak{b}_2})} = \big(x^{(\mathbf{B},{\mathfrak{b}_2})}\big)^{(\mathbf{A},{\mathfrak{b}_1})} = x^{(\mathbf{B},{\mathfrak{b}_2}) (\mathbf{A},{\mathfrak{b}_1})}.
	\end{equation}
	\qed	
\end{theorem}

Note that different values of bases $\mathfrak{b}_1$ and $\mathfrak{b}_2$ in Theorem \ref{Thm:Basic-theorem-for-risinng-to-power} do not affect the application of Theorem \ref{Thm:Basic-theorem-for-powers-in-entropic-groupoids}, since Theorem \ref{Thm:Basic-theorem-for-powers-in-entropic-groupoids} holds true for any bracketing shapes $\mathbf{A}$ and $\mathbf{B}$.

The arithmetic (Logarithmetic) for indices $(\mathbf{A},{\mathfrak{b}_1})$ and $(\mathbf{B},{\mathfrak{b}_2})$ is thus implicitly defined with Theorem \ref{Thm:Basic-theorem-for-powers-in-entropic-groupoids} and Definition  \ref{Def:Addition-and-Multiplication-of-Endomorphisms}. Namely, $(\mathbf{C},{\mathfrak{b}_3}) = (\mathbf{A},{\mathfrak{b}_1}) + (\mathbf{B},{\mathfrak{b}_2})$ iff $x^{(\mathbf{C},{\mathfrak{b}_3})} = x^{(\mathbf{A},{\mathfrak{b}_1})} * 
x^{(\mathbf{B},{\mathfrak{b}_2})}$, and $(\mathbf{D},{\mathfrak{b}_4}) = (\mathbf{A},{\mathfrak{b}_1}) (\mathbf{B},{\mathfrak{b}_2})$ iff $x^{(\mathbf{D},{\mathfrak{b}_4})} = (x^{(\mathbf{A},{\mathfrak{b}_1})})^{(\mathbf{B},{\mathfrak{b}_2})}$ for all $x \in \mathbb{E}_{p^2}$. 

\begin{proposition}\label{Prop:Addition-and-multiplication-of-integer-parts}
	Let $(\mathbf{A},{\mathfrak{b}_1}) = (a, a_s, \mathfrak{b}_1)$ and $(\mathbf{B},{\mathfrak{b}_2}) = (b, b_s, \mathfrak{b}_2)$ be two power indices where $a, b \in \mathbb{Z}^+$, and let $(c, c_s, \mathfrak{b}_3) = (\mathbf{C},{\mathfrak{b}_3}) = (\mathbf{A},{\mathfrak{b}_1}) + (\mathbf{B},{\mathfrak{b}_2})$ and $(d, d_s, \mathfrak{b}_4) = (\mathbf{D},{\mathfrak{b}_3}) = (\mathbf{A},{\mathfrak{b}_1}) (\mathbf{B},{\mathfrak{b}_2})$. Then, $c = a + b$ and $d = a b$.
\end{proposition}
\begin{proof}
	The proof is by direct application of Theorem \ref{Thm:Basic-theorem-for-powers-in-entropic-groupoids} for powers in entropic groupoids.
\end{proof}

While we know the values of $c$ and $d$ in $(\mathbf{C},{\mathfrak{b}_3})$ and $(\mathbf{D},{\mathfrak{b}_4})$, we do not know explicitly the shapes $c_s$ and $d_s$.
\begin{openroblem}\label{Open-problem:FindExplixitRulesForLogarithmetic}
	For a given entropoid $\mathbb{E}_{p^2}$ and two power indices $(\mathbf{A},{\mathfrak{b}_1})$ and $(\mathbf{B},{\mathfrak{b}_2})$ as in Definition \ref{Def:Power-indices} find the explicit forms for  $(\mathbf{C},{\mathfrak{b}_1})$ and $(\mathbf{D},{\mathfrak{b}_2})$ such that $(\mathbf{C},{\mathfrak{b}_3}) = (\mathbf{A},{\mathfrak{b}_1}) + (\mathbf{B},{\mathfrak{b}_2})$ and $(\mathbf{D},{\mathfrak{b}_3}) = (\mathbf{A},{\mathfrak{b}_1}) (\mathbf{B},{\mathfrak{b}_2})$.
\end{openroblem}

Despite the efficiency of the procedure for non-associative and non-commutative exponentiation given in Definition \ref{Def:Raising-to-power}, we have to notice that for a fixed base $\mathfrak{b}$ there are many bracketing shapes that are omitted and can not be produced with the expressions  (\ref{Eq:Raising-to-power-v1}) - (\ref{Eq:Raising-to-power-v4}). Apparently, for $\mathfrak{b}=2$ the pattern $a_s = [P_0, P_1, \ldots, P_k]_{\mathfrak{b}-1}$ where $0 \leq P_i \leq \mathfrak{b}-2$ for $i = 0, \ldots, k$ becomes a constant pattern $a_s = [0, 0, \ldots, 0]$, (that we denote shortly by $\mathbf{[0]}$) and only the conditions for even $P_{i-1}$ apply in (\ref{Eq:Raising-to-power-v3}). 

Still, we can use the limitations of Definition \ref{Def:Raising-to-power} for good: the fixed pattern related to the base $\mathfrak{b}=2$ can help us propose heuristics for efficient finding of generators for $(\mathbb{E}^*_{(p-1)^2}, *)$. For that purpose, let us first give several definitions and propositions about the subgroupoids generated by elements of $\mathbb{E}_{p^2}$ with ${(\mathbf{A}, 2)} = (a, a_s, 2) = (a, \mathbf{[0]}, 2)$ indices and with general indices $\mathbf{A}$.

\begin{definition}\label{Def:Generated-sets-and-}
	For every $x = (x_1, x_2) \in \mathbb{E}_{p^2}$, we define the set $\langle x \rangle_2$ as 
	\begin{equation}
		\langle x \rangle_2 = \{ x^{(\mathbf{A}, 2)}\ | \ (\mathbf{A}, 2) = (a, \mathbf{[0]}, 2), a \in \mathbb{Z}^+ \},
	\end{equation}
	and the set $\langle x \rangle$ as
	\begin{equation}
		\langle x \rangle = \{ x^{\mathbf{A}}\ | \ \mathbf{A} = (a, a_s), a \in \mathbb{Z}^+, a_s \text{\ a bracketing shape} \}.
	\end{equation}
	We say the subgroupoid $(\langle x \rangle_2, *) = (\mathbb{E}^*_{|\langle x \rangle_2|}, *)$ is a multiplicative cyclic subgroupoid of $\mathbb{E}_{p^2}$ of order $|\langle x \rangle_2|$, and $(\langle x \rangle, *) = (\mathbb{E}^*_{|\langle x \rangle|}, *)$ is a multiplicative subgroupoid of $\mathbb{E}_{p^2}$ of order $|\langle x \rangle|$.	
\end{definition}

Without a proof which is left as an exercise, we give the following two propositions

\begin{proposition}\label{Prop:Cyclic-subgroupoid}
	If $s(x)_2 = | \langle x \rangle_2 |$ is the size of the set $\langle x \rangle_2$, then 
	\begin{equation}
		s(x)_2\ |\ 2(p - 1),
	\end{equation}
	and
	\begin{equation}
		s_{max}(x)_{2} = \max_{x \in \mathbb{E}_{p^2}} (s(x)_2) = 2 (p - 1). 
	\end{equation}\qed
\end{proposition}

\begin{proposition}\label{Prop:All-generated-subgroupoids}
	If $s(x) =\ |\langle x \rangle |$ is the size of the set $\langle x \rangle$, then 
	\begin{equation}
		s(x)\ |\ (p - 1)^2
	\end{equation}
	and
	\begin{equation}
		s_{max}(x) = \max_{x \in \mathbb{E}_p} s(x) = (p - 1)^2. 
	\end{equation}\qed
\end{proposition}

The problem of finding an explicit analytical expression for the group generators is still an open problem. One version of that problem is the famous Artin's conjecture on primitive roots \cite{ARTINS-CONJECTURE-FOR-PRIMITIVE-ROOTS}. However, there are efficient heuristic algorithms that find generators of the multiplicative group of a finite field (for example, see \cite{menezes2018handbook}[Alg. 4.84, Note 4.82, Alg. 4.86]). We especially point out the efficient algorithm Alg 4.86 in \cite{menezes2018handbook} for finding a generator of $\mathbb{F}^*_p$ where $p$ is a safe prime. Inspired by that algorithm, we propose here an efficient heuristic algorithm for finding a generator $g$ of the maximal multiplicative quasigroup $(\mathbb{E}^*_{(p-1)^2}, *)$.

\begin{definition}
	A safe prime $p$ is a prime of the form $p = 2 q + 1$ where $q$ is also a prime.
\end{definition}

\begin{conjecture}\label{Conj:Generator}
	Let $\mathbb{E}_{p^2}$ be an entropoid defined with a $\lambda$ bit safe prime $p$, and let $g \in \mathbb{E}^*_{(p-1)^2}$. If the following conditions are true
	\begin{align}
		g & \neq g^{(p, \mathbf{[0]}, 2)}\\
		g*g & \neq g^{(p-1, \mathbf{[0]}, 2)}\\
		(g*(g*g)) & \neq g^{(p-2, \mathbf{[0]}, 2)}\\
		(g*(g*g)) & \neq ((g*g)*g) \\
		(g*(g*(g*g))) & \neq ((g*(g*g))*g),
	\end{align}
	then the probability that $g$ is the generator of $\mathbb{E}^*_{(p-1)^2}$ is 
	\begin{align}
		Pr[g \text{ is generator of } G^*] & > 1 - \epsilon, \ \text{where }	\epsilon < \frac{1}{2^\lambda}.
	\end{align}
\end{conjecture}

The corresponding algorithm coming from Conjecture \ref{Conj:Generator} is Algorithm \ref{Alg:Generator}.
\begin{algorithm}
	\small
	\caption{$\mathtt{Gen}$: Find a generator $g$ for the multiplicative quasigroup $(\mathbb{E}^*_{(p-1)^2}, *)$.
		\newline
		\textbf{Input:} $\lambda$ bit safe prime $p$ and $a_3, a_8, b_2, b_7 \in \mathbb{F}_{p}$ that define $\mathbb{E}_{p^2}$;
		\newline
		\textbf{Output:} Generator $g$ for $(\mathbb{E}^*_{(p-1)^2}, *)$.}
	\begin{algorithmic}[1]
		\Repeat
		\State Set $Success \leftarrow \mathtt{True}$;
		\State Choose random element $g \in G^*$;
		\State Set $Success \leftarrow \big(Success \mathtt{\ and\ } (g \neq g^{(p, \mathbf{[0]}, 2)} ) \big)$
		\State Set $Success \leftarrow \big(Success \mathtt{\ and\ } (g*g \neq g^{(p-1, \mathbf{[0]}, 2)} ) \big)$
		\State Set $Success \leftarrow \big(Success \mathtt{\ and\ } (\ (g*(g*g)) \neq g^{(p-2, \mathbf{[0]}, 2)} ) \big)$
		\State Set $Success \leftarrow \big(Success \mathtt{\ and\ } (\ (g*(g*g)) \neq ((g*g)*g) ) \big)$
		\State Set $Success \leftarrow \big(Success \mathtt{\ and\ } (\ (g*(g*(g*g))) \neq ((g*(g*g))*g) ) \big)$
		\Until{$Success$}
		\State Return $g$.
	\end{algorithmic}\label{Alg:Generator}
\end{algorithm}

For a detailed demonstration of the Proposition \ref{Prop:Cyclic-subgroupoid}, Proposition \ref{Prop:All-generated-subgroupoids} and Conjecture \ref{Conj:Generator} see the Appendix \ref{App:Examples-for-conjecture}. 

\begin{openroblem}
	For a given entropoid $\mathbb{E}_{p^2}$ prove the Conjecture \ref{Conj:Generator} or construct another exact or probabilistic efficient algorithm for finding generator of $\mathbb{E}^*_{(p-1)^2}$.
\end{openroblem}	

A direct consequence from Proposition \ref{Prop:Cyclic-subgroupoid} and Proposition \ref{Prop:All-generated-subgroupoids} is the following corollary.
\begin{corollary}\label{Cor:Probability-To-Be-Power-Of-2}
	Let $g$ obtained with the Algorithm \ref{Alg:Generator} be a generator of $\mathbb{E}^*_{(p-1)^2}$, where $p = 2q + 1$ is a safe prime. Let $y  \xleftarrow{\mathbf{r}} \mathbb{E}^*_{(p-1)^2}$ be a randomly chosen element from $\mathbb{E}^*_{(p-1)^2}$. Then the probability that there exists a power index $(\mathbf{A},2)$ such that $y = g^{(\mathbf{A},2)}$ is
	\begin{equation}
		Pr(\{\exists (\mathbf{A},2) \text{ and } y = g^{(\mathbf{A},2)} \}) = \frac{1}{q}.
	\end{equation}
\end{corollary}
\begin{proof}
	The proof of this Corollary is just a direct ratio between the size of the cyclic subgroupoid and the size of the maximal entropic quasigroup: $$\frac{s_{max}(x)_2}{s_{max}(x)} = \frac{2(p-1)}{(p-1)^2} = \frac{1}{q}.$$
\end{proof}

Having a heuristics for finding generators of the maximal quasigroup $\mathbb{E}^*_{(p-1)^2}$, where $p = 2 q + 1$ is a safe prime, it would be also very beneficial if we can find generators for the Sylow $q$-subquasigroup $\mathbb{E}^*_{q^2}$.

\begin{conjecture}\label{Conj:Generator-Sylow-q}
	Let $g$ be a generator of $\mathbb{E}^*_{(p-1)^2}$, where $p \geq 11$. Then 
	\begin{equation}
		g_q = (g*(g*(g*((g*g)*g)))),
	\end{equation} is the generator of the Sylow $q$-subquasigroup $\mathbb{E}^*_{q^2}$.
\end{conjecture}

The corresponding algorithm coming from Conjecture \ref{Conj:Generator-Sylow-q} is Algorithm \ref{Alg:Generator-Sylow-q}.
\begin{algorithm}
	\small
	\caption{$\mathtt{GenQ}$: Find a generator $g_q$ for the Sylow $q$-subquasigroup $\mathbb{E}^*_{q^2}$.
		\newline
		\textbf{Input:} $\lambda$ bit safe prime $p$ and $a_3, a_8, b_2, b_7 \in \mathbb{F}_{p}$ that define $\mathbb{E}_{p^2}$;
		\newline
		\textbf{Output:} Generator $g_q$ for $\mathbb{E}^*_{q^2}$.}
	\begin{algorithmic}[1]
		\State Set $g \leftarrow \mathtt{Gen}(\lambda, p, a_3, a_8, b_2, b_7)$;
		\State Set $g_q \leftarrow (g*(g*(g*((g*g)*g)))) $;
		\State Return $g_q$.
	\end{algorithmic}\label{Alg:Generator-Sylow-q}
\end{algorithm}

\begin{proposition}\label{Prop:Even-vs-Odd-powers-Distinguisher}
	Let $g$ be a generator of $\mathbb{E}^*_{(p-1)^2}$. Then for every bracketing shape $(\mathbf{A},\mathfrak{b}) = (a, a_s, \mathfrak{b}) \in \mathbb{L}$ the following relation hold:
	\begin{equation}\label{Eq:Even-vs-Odd-powers-Distinguisher}
		\bigg(g^{(\mathbf{A},\mathfrak{b})}\bigg)^{(p-1, \mathbf{[0]}, 2)} = \left\{
		\begin{array}{rl}
			\mathbf{1}_*, & \text{ if } a \text{ is even}, \\
			\boxminus \mathbf{1}_*, & \text{ if } a \text{ is odd}.
		\end{array} \right.\\
	\end{equation}
\end{proposition}
\begin{proof}
	We use the Theorem \ref{Thm:Basic-theorem-for-risinng-to-power} and Proposition \ref{Prop:Addition-and-multiplication-of-integer-parts} to deduce that when working with base $b=2$ we work with a cyclic groupoid. Then from Proposition \ref{Prop:Cyclic-subgroupoid} we know that the order of $x = g^{(\mathbf{A},\mathfrak{b})}$ is $2(p - 1)$. That means that the order of $y = x^{(p-1, \mathbf{[0]}, 2)}$ is 2 i.e. $y^2 = y * y = \mathbf{1}_*$. Thus $y$ belongs to the set of square roots of the left unit $\mathbb{S}(p)$. Now, if $y$ takes any value other than $\mathbf{1}_*$ or $\boxminus \mathbf{1}_*$ it will lead to a result that the order of $x$ is greater than $2(p - 1)$ which is a contradiction. This leaves only two possibilities: $y = \mathbf{1}_*$ or $y = \boxminus \mathbf{1}_*$.
	Again from Theorem \ref{Thm:Basic-theorem-for-risinng-to-power} and Proposition \ref{Prop:Addition-and-multiplication-of-integer-parts} we get that there must exist an index $(\mathbf{B},2) = (b, \mathbf{[0]}, 2)$ such that $(\mathbf{B},2) = (\mathbf{A},\mathfrak{b}) (p-1, \mathbf{[0]}, 2) $. Thus $b = a (p - 1)$, and the result $y$ will depend on the parity of $a$.
\end{proof}

\begin{proposition}\label{Prop:Non-distinguisher-in-Sylow-subquasigroups}
	Let $g_q$ be a generator of the Sylow $q$ quasigroup $\mathbb{E}^*_{q^2}$. Then for every bracketing shape $(\mathbf{A},\mathfrak{b}) = (a, a_s, \mathfrak{b}) \in \mathbb{L}$ the following relations hold:
	\begin{align}
		\label{Eq:Non-distinguisher-in-Sylow-subquasigroups}
		\bigg(g_q^{(\mathbf{A},\mathfrak{b})}\bigg)^{(p-1, \mathbf{[0]}, 2)} & = \mathbf{1}_*\\
		\bigg(g_q^{(\mathbf{A},\mathfrak{b})}\bigg)^{(q, \mathbf{[0]}, 2)} & = y, 
	\end{align}
	where $y$ belongs in a subset of the set of square roots of the left unit i.e. $y \in \mathbb{S}_q(p) \subseteq \mathbb{S}(p)$, such that $|\mathbb{S}_q(p)| = q$.
\end{proposition}
\begin{proof}
	Similar arguments hold for this situation with a distinction that now, the set of all values $x = g_q^{(\mathbf{A},\mathfrak{b})}$ have order $q$ instead of $2q$ that was in the previous case. That, means $x^{(p-1, \mathbf{[0]}, 2)} = x^{(2q, \mathbf{[0]}, 2)} = \mathbf{1}_*$ and again putting $y = x^{(q, \mathbf{[0]}, 2)}$ we have that $y^2 = y * y = \mathbf{1}_*$. So, $y$ must belong to $\mathbb{S}(p)$. If $\mathbb{S}_q(p) \subseteq \mathbb{S}(p)$  is the subset from where $y$ receives its values, its cardinality must be $q$, otherwise it will generate the whole multiplicative quasigroup $\mathbb{E}^*_{(p-1)^2}$. 
\end{proof}

\subsection{Dichotomy between odd and even bases}
There are a few more issues with the bracketing shapes defined in Definition \ref{Def:Raising-to-power} that need to be discussed. As we see in Proposition \ref{Prop:Cyclic-subgroupoid} and Proposition \ref{Prop:All-generated-subgroupoids}, the choice of the base $b$ influence the size of the generated sets when raising to powers. It also influence the probability an element $x \in \mathbb{E}^*_{(p-1)^2}$ to be an image of some generator raised to some power index i.e. $x = g^{(\mathbf{A}, \mathfrak{b})} = g^{(a, a_s, \mathfrak{b})}$. This means, since $g$ is a generator, there is certainly some power index $(\mathbf{A}, \mathfrak{b}) = (a, a_s, \mathfrak{b})$ for which $x = g^{(\mathbf{A}, \mathfrak{b})}$, but when distributed over all patterns $a_s = [P_0, P_1, \ldots, P_k]_{\mathfrak{b}-1}$, for some elements $x$ there will be many patterns that give $x = g^{(a, a_s, \mathfrak{b})}$, while for other elements $x$ there will be a few. To formalize this discussion we adapt the approach that Smith proposed in \cite{smith2001some} about the Shannon entropy of completely partitioned sets, and the relations between Shannon entropy and several instances of R\'{e}nyi entropy (such as collision entropy and min-entropy) studied by Cachin in \cite{cachin1997entropy} (see also the work of Sk{\'o}rski \cite{skorski2015shannon}).

For a given base $2 \leq \mathfrak{b} \leq b_{max}$ and given generator $g \in  \mathbb{E}^*_{(p-1)^2}$, let us investigate how the sequence of sets of shapes
\begin{align}
	\mathbb{L}(i) = & \{ a_s \ | a_s = [P_0, \ldots, P_i]_{\mathfrak{b}-1}, \text{ for all } P_j \in \mathbb{Z}_{\mathfrak{b}-1}, j = 0, \ldots, i  \},
\end{align}
each with $(\mathfrak{b}-1)^i$ elements, are partitioned into $r_i$ sets of mutually exclusive subsets 
\begin{align}
	\xi_i = & \{ C_{i,1}, \ldots, C_{i, r_i}  \}.
\end{align}

The partitioning is done due to the following conditions:
\begin{align}
	\forall a_{s_1}, a_{s_2} \in C_{i,j}, & \ \ \ g^{(\mathfrak{b}^{i-1}, a_{s_1}, \mathfrak{b})} = g^{(\mathfrak{b}^{i-1}, a_{s_2}, \mathfrak{b})} = g_{i,j},\\
	\forall a_{s_1} \in C_{i,j_1}, \text{ and } \forall a_{s_2} \in C_{i,j_2}, & \ \ \ g^{(\mathfrak{b}^{i-1}, a_{s_1}, \mathfrak{b})} \neq g^{(\mathfrak{b}^{i-1}, a_{s_2}, \mathfrak{b})}, \text{ when } j_1 \neq j_2.
\end{align}

As a short notation we write $g^{(\mathfrak{b}^{i-1}, \mathbb{L}(i), \mathfrak{b})}$ the set of all powers of $g$ to $\mathfrak{b}^{i-1}$ with all possible shapes $\mathbb{L}(i)$.

If a shape $a_s$ is sampled uniformly at random from $\mathbb{L}(i)$, i.e. if $a_s \xleftarrow{\mathbf{r}} \mathbb{L}(i)$, then $g^{(\mathfrak{b}^{i-1}, a_s, \mathfrak{b})}$ determines the set $C_{i,j}$ where it belongs. The probability that $a_s$ belongs to $C_{i,j}$ depends on the number of elements $n_{ij} = |C_{i,j}|$ and is calculated as
\begin{align}
	p_{ij} = p(C_{i, j}) = \frac{n_{ij}}{(b-1)^i}.
\end{align}

The Shannon entropy $H_1$ of the partitioned set $\xi_i$ is defined with
\begin{equation}
	H_1(\xi_i) = -\sum_{j=1}^{r_i}p(C_{i, j}) \log p(C_{i, j}).
\end{equation}

Similarly, the R\'{e}nyi entropy of order $\alpha$ for $\xi_i$  is defined as
\begin{equation}
	H_{\alpha}(\xi_i) = \frac{1}{1 - \alpha} \log \sum_{j=1}^{r_i}p(C_{i, j})^{\alpha},
\end{equation}
with the special instance for $\alpha = 2$ which is called the Collision entropy
\begin{equation}
	H_{2}(\xi_i) = - \log \sum_{j=1}^{r_i}p(C_{i, j})^{2}.
\end{equation}
 
Min entropy is defined as
\begin{equation}\label{Eqn:Min-Entropy}
		H_{\infty}(\xi_i) = - \log \max_{C_{i, j} \in \xi_i}p(C_{i, j}) = - \log \frac{\max n_{ij}}{(\mathfrak{b}-1)^i} .
\end{equation}

The ordering relation (see \cite[Lemma 3.2.]{cachin1997entropy}) among different entropies is 
\begin{equation}
	H_{\infty}(\xi_i) \leq H_{2}(\xi_i) \leq H_1(\xi_i)
\end{equation} 
 
Let us point to the fact that the grouping of the shapes is governed by the Narayana numbers given in equation (\ref{Eq:Narayana_S_a(g)}) in Theorem \ref{Thm:Equivalent-Classes}. Thus, for even bases $\mathfrak{b}$ there are $\mathfrak{b}-1$ classes that partition the set of all possible shapes with cardinality expressed by the Narayana numbers $ N(\mathfrak{b}-1, i) = \frac{1}{\mathfrak{b}-1}\binom{\mathfrak{b}-1}{i}\binom{\mathfrak{b}-1}{i-1}$. Since $\mathfrak{b}-1$ in that case is odd, there is one central dominant Narayana number, and there will be one class of shapes that will have a dominant number of members. For example, for $\mathfrak{b}=6$, the five classes have cardinality $\{1, 10, 20, 10, 1\}$, so the central class is a dominant one with 20 elements. For $\mathfrak{b}=8$ the seven classes have cardinality $\{1, 21, 105, 175, 105, 21, 1\}$, so the central class is a dominant one with 175 elements. On the other hand, for odd bases $\mathfrak{b}$, we have grouping in even number of $\mathfrak{b}-1$ classes, the sequences of Narayana numbers are completely symmetrical and there is not one but two dominant classes. For example, for $\mathfrak{b}=7$, the six classes have cardinality $\{1, 15, 50, 50, 15, 1\}$, so two central classes are dominant with 50 elements.

For bigger bases $\mathfrak{b}$ we have observed the same pattern: for even bases $\mathfrak{b} = 2 \mathfrak{b}_1$ the values of $\max n_{ij}$ that determine the min entropy $H_{\infty}$ are increasing with the same exponential speed as the values of $(\mathfrak{b}-1)^i$ increase. Looking at the equation (\ref{Eqn:Min-Entropy}) it makes $\frac{\max n_{ij}}{(\mathfrak{b}-1)^i}$ to trend to 1 i.e. $H_{\infty}$ trends to 0.
	
For odd bases $\mathfrak{b} = 2 \mathfrak{b}_1 + 1$, while there is increase of $\max n_{ij}$ as $i$ increases, the ratio $\frac{\max n_{ij}}{(\mathfrak{b}-1)^i}$ is actually decreasing, which makes  $H_{\infty}$ to increase.

See Appendix \ref{App:Even-Odd-Example} for details about this observed dichotomy between even and odd bases.

We summarize this discussion with the following Conjecture
\begin{conjecture}\label{Conj:Conjecture-even-b}
    Let $g \in \mathbb{E}^*_{(p-1)^2}$ be a generator of the maximal multiplicative subgroupoid of a given entropoid $\mathbb{E}_{p^2}$. For every even base $\mathfrak{b} = 2 b_1 < b_{max}$ the values $n(\mathfrak{b}, i) = \max n_{ij}$ are given by the following relation:
    \begin{equation}\label{Eqn:Conjecture-even-b}
    	n(\mathfrak{b}, i) = (\mathfrak{b} - 1) \bigg( (\mathfrak{b} - 1)^{i-1} - (\mathfrak{b} - 2)^{i-1} \bigg),
    \end{equation}
	which makes the following relation about the min entropy $H_{\infty}$
    \begin{equation}\label{Eqn:Min-entropy-even-b}
		H_{\infty}(\xi_i) = 1 - \left( \frac{\mathfrak{b} - 2}{\mathfrak{b} - 1}\right) ^{i-1}.
\end{equation}
	
\end{conjecture} 

The discussion so far was about the entropy of the pattern sets $\mathbb{L}(i)$ partitioned to sets $\xi_i$ induced by rising a generator $g$ to a special forms of powers $g^{(a, a_{s_1}, \mathfrak{b})}$ where $a = \mathfrak{b}^{i-1}$. This basically means that in the procedure for rising to a power we use only the equations (\ref{Eq:Raising-to-power-v1}) and (\ref{Eq:Raising-to-power-v2}), since the numbers $a$ in that case have in the little-endian notation the following forms $a = [0, 0, \ldots, 1]_\mathfrak{b}$. One might hope that the entropy of $\xi$ for even bases will improve significantly if we work with generic numbers $a$ with a lot of non-zero digits $a = [A_0, A_1, \ldots, A_k]_\mathfrak{b}$. However, that is not the case as it is showed in Figure \ref{Figure:EvenBases} in Appendix \ref{App:Even-Odd-Example}. 

The experiments were conducted by generating a random entropoid $\mathbb{E}_{p^2}$ with safe prime number $p$ with $\lambda$ bits. After finding a generator $g$ for the multiplicative quasigroup $\mathbb{E}^*_{(p-1)^2}$, we generated one random number $a \in \mathbb{Z}_{(p-1)^2}$, and then we run the procedure of rising to a power $g^{(a, a_s, \mathfrak{b})}$ for random shapes $a_s = [P_0, \ldots, P_i]_{\mathfrak{b}-1}$ until the first collision. 

As the size of the finite entropoid $\mathbb{E}_{p^2}$ increases, the collision entropy for different even bases remains constant. On the other hand, with the odd bases the situation is completely different. We see in Figure \ref{Figure:OddBases} that as the size of the entropoid increases, the collision entropy increases as well. A loose observation is that for $p$ being $\lambda$ bits long, the collision entropy is $H_2(\xi) \approx \frac{\lambda}{2}$.
\begin{openroblem}
	For odd bases $\mathfrak{b}$ find proofs and find tighter bounds for the collision entropy $H_2(\xi) $.
\end{openroblem}

\section{Hard Problems in Entropoid Based Cryptography}\label{Sec:Hard-problems}

We now have enough mathematical understanding and heuristic evidence to precisely formulate several hard problems in entropoid based cryptography, in a similar fashion as the discrete logarithm problem, and computational and decisional Diffie-Hellman problems are defined within the group theory. 
We will use the notion of negligible function $\mathtt{negl}: \mathbb{N} \mapsto \mathbb{R}$ for the function that for every $c \in \mathbb{N}$ there is an integer $n_c$ such that $\mathtt{negl}(n) \leq n^{-c}$ for all $n\geq n_c$.

\begin{definition}[Discrete Entropoid Logarithm Problem (DELP)]
	An entropoid $\mathbb{E}_{p^2}$ and a generator $g_{q_i}$ of one of its Silow subquasigroups $\mathbb{E}^*_{\nu}$ are publicly known. Given an element $y \in \mathbb{E}^*_{\nu}$ find a power index $(\mathbf{A}, \mathfrak{b})$ such that $y = g_{q_i}^{(\mathbf{A}, \mathfrak{b})}$.
\end{definition}

\begin{definition}[Computational Entropoid Diffie–Hellman Problem (CEDHP)] An entropoid $\mathbb{E}_{p^2}$ and a generator $g_{q_i}$ of one of its Silow subquasigroups $\mathbb{E}^*_{\nu}$ are publicly known. Given $g_{q_i}^{(\mathbf{A}, \mathfrak{b}_1)}$ and $g_{q_i}^{(\mathbf{B}, \mathfrak{b}_2)}$, where $(\mathbf{A},{\mathfrak{b}_1}), (\mathbf{B},{\mathfrak{b}_2}) \xleftarrow{\mathbf{r}} \mathbb{L}$, compute $g_{q_i}^{(\mathbf{A},{\mathfrak{b}_1}) (\mathbf{B},{\mathfrak{b}_2})}$.
\end{definition}

The similar reduction as with DLP and CDH is true for DELP and CEDHP: CEDHP $\leq$ DELP i.e. CEDHP is no harder than DELP. Namely, if an adversary can solve DELP, it can find $(\mathbf{A},{\mathfrak{b}_1})$ and $(\mathbf{B},{\mathfrak{b}_2})$ and compute $g_{q_i}^{(\mathbf{A},{\mathfrak{b}_1}) (\mathbf{B},{\mathfrak{b}_2})}$.

\begin{definition}[Decisional Entropoid Diffie–Hellman Problem (DEDHP)] An entropoid $\mathbb{E}_{p^2}$ and a generator $g_{q_i}$ of one of its Silow subquasigroups $\mathbb{E}^*_{\nu}$ are publicly known. Given $g_{q_i}^{(\mathbf{A}, \mathfrak{b}_1)}$, $g_{q_i}^{(\mathbf{B}, \mathfrak{b}_2)}$ and $g_{q_i}^{(\mathbf{C}, \mathfrak{b}_3)}$, where $(\mathbf{A},{\mathfrak{b}_1}), (\mathbf{B},{\mathfrak{b}_2}), (\mathbf{C}, \mathfrak{b}_3) \xleftarrow{\mathbf{r}} \mathbb{L}$, decide if $(\mathbf{C}, \mathfrak{b}_3) = (\mathbf{A},{\mathfrak{b}_1}) (\mathbf{B},{\mathfrak{b}_2})$ or $ (\mathbf{C}, \mathfrak{b}_3) \xleftarrow{\mathbf{r}} \mathbb{L}$.
\end{definition}

Again, the similar reduction as with classical CDH and DDH, holds here: DEDHP $\leq$ CEDHP i.e. DEDHP is no harder than CEDHP, since if an adversary can solve CEDHP, it will compute $g_{q_i}^{(\mathbf{A},{\mathfrak{b}_1}) (\mathbf{B},{\mathfrak{b}_2})}$ and will compare it with $g_{q_i}^{(\mathbf{C}, \mathfrak{b}_3)}$.

As with the classical DDH for the multiplicative group $\mathbb{F}_p^*$ where DDH is easy problem, but for its quadratic residues subgroup $QR(p)$, DDH is hard, we have a similar situation for DEDHP which is stated in the following Lemma.
\begin{lemma}
	Let $\mathbb{E}_{p^2}$  be an entropoid and $g$ be a generator of its maximal quasigroup $\mathbb{E}^*_{(p-1)^2}$. Then there is an efficient algorithm that solves DEDHP in $\mathbb{E}^*_{(p-1)^2}$.
\end{lemma}
\begin{proof}
	The distinguishing algorithm can be built based on the distinguishing property described in Proposition \ref{Prop:Even-vs-Odd-powers-Distinguisher}.
\end{proof}

On the other hand, based on Proposition \ref{Prop:Non-distinguisher-in-Sylow-subquasigroups} we can give the following plausible conjecture.
\begin{conjecture}\label{Conj:DEDHP-is-hard-in-Sylow-q}
	Let $\mathbb{E}_{p^2}$ be an entropoid, where $p = 2 q + 1$ is a safe prime and $g_{q}$ is the generator of its Sylow $q$-subquasigroup $\mathbb{E}^*_{q^2}$. Then there is no algorithm $\mathcal{A}$ that solves DEDHP in $\mathbb{E}^*_{q^2}$ with significantly higher advantage over the strategy of uniformly random guesses for making the decisions.
\end{conjecture}

\begin{theorem}
	If the DEDHP conjecture is true, then a Diffie-Hellman key exchange protocol over finite entropoids is secure in the Canetti-Krawczyk  model of passive adversaries \cite{Canetti2001a}.
\end{theorem}

We will make here a slight digression, and will relate the classical DLP with another problem over the classical group theory: finding roots. Then we will just translate it for the case of finite entropoids.
\begin{definition}[Computational Discrete Root Problem (CDRP)] A group $G$ of order $N$ and its generator $g$ are publicly known. Given $y = x^b$ and $b$, where $b \xleftarrow{\mathbf{r}} \mathbb{Z}_N$, compute $x = \sqrt[b]{y}$.
\end{definition}

In general, CDRP is an easy problem. However, there are instances where this problem is still hard, and we will discuss those instances now. 

One of the best generic algorithms for solving CDRP is by Johnston  \cite{conf/soda/Johnston99}. As mentioned there, CDRP can be reduced to solving the DLP in $G$, i.e. CDRP $\leq$ DLP. First of all $b$ is supposed to divide $N$, otherwise due to the cyclic nature of the group $G$, it is a straightforward technique that finds the $b$-th root: $x = \sqrt[b]{y} = y ^ {\frac{1}{b}\mod N}$. Let us denote $g_1 = g^b$, where $g$ is the generator of $G$. If we have a DLP solver for $G$, and if there exists a solution for the equation $y = g_1^a$, then DLP will find $a$ efficiently. Then we can compute $x$ as $x = g^a$. Johnston noticed that if $\frac{N}{b}$ is small, then DLP solver will be efficient with a complexity $O(\sqrt{\frac{N}{b}})$. On the other hand if  $\frac{N}{b}$ is not that small, Johnston made a reduction to  another DLP solver, by heavily using the reach algebraic structure of the finite cyclic groups generated by a gennerator $g$. The other DLP solver computes a discrete log of $x^{\frac{N}{b^k}}$ using the generator $g^\frac{N}{b^{k-1}}$, where $k$ is the largest power of $b$ such that $b^k$ still divides $N$. The  generic complexity of this DLP solver is $O((k-1)\sqrt{b})$. Now let us work in finite field $\mathbb{F}_p$ with the following prime number: $p = 2 q^3 + 1$ where $q$ has $\lambda$ bits. So, if fix the $b$-th root to be exactly $b = q$, then for the second DLP solver in the Johnson technique we have that $k = 3$, and it has an exponential complexity of $O((k-1)2^\frac{\lambda}{2})$.



\begin{definition}[Computational Discrete Entropoid Root Problem (CDERP)] An entropoid $\mathbb{E}_{p^2}$ and a generator $g$ of its multiplicative quasigroups $\mathbb{E}^*_{(p-1)^2}$ are publicly known. Given $y = x^{(\mathbf{B}, \mathfrak{b})}$ and $(\mathbf{B}, \mathfrak{b})$, where $x, y \in \mathbb{E}^*_{(p-1)^2}$ and $(\mathbf{B},\mathfrak{b}) \xleftarrow{\mathbf{r}} \mathbb{L}$, compute $x = \sqrt[(\mathbf{B},\mathfrak{b})]{y}$. \end{definition}

A similar discussion applies for CDERP that it is not harder than DELP, i.e., CDERP $\leq$ DELP. However, notice that it is not possible directly to extend the Johnson technique to finite entropoids due to the lack of the associative law and because $\mathbb{E}^*_{(p-1)^2}$ is not a cyclic structure. So, at this moment, we can make the following plausible conjecture.
\begin{conjecture}\label{Conj:CDERP-is-hard-in-Sylow-q}
	Let $\mathbb{E}_{p^2}$ be an entropoid, where $p = 2 q + 1$ is a safe prime with $\lambda$ bits and $g$ is the generator of its multiplicative quasigroups $\mathbb{E}^*_{(p-1)^2}$. Let $\mathcal{A}$ be an algorithm that solves CDERP in $\mathbb{E}^*_{(p-1)^2}$. Then the probability, over uniformly chosen $(\mathbf{B},\mathfrak{b}) \xleftarrow{\mathbf{r}} \mathbb{L}$ that  $\mathcal{A}(x^{(\mathbf{B}, \mathfrak{b})}, (\mathbf{B},\mathfrak{b}) ) = x$ is $\mathtt{negl}(\lambda)$.
\end{conjecture}

We want to emphasize one essential comparison between CDRP in cyclic groups $G$ of order $N$ and CDERP in the multiplicative quasigroups $\mathbb{E}^*_{(p-1)^2}$. CDERP is easy problem for almost all root values $b$ except when $b = q$ in groups that have orders $N$ divisible by $q^k$ where $k \geq 2$. CDERP is conjectured that is hard in $\mathbb{E}^*_{(p-1)^2}$ (that has order $4 q^2$) for every randomly selected root $(\mathbf{B},\mathfrak{b}) \xleftarrow{\mathbf{r}} \mathbb{L}$. The conjecture is based on the fact that currently, there is no developed Logarithmetic for the succinct power indices, but more importantly, that $\mathbb{E}^*_{(p-1)^2}$ is neither a group nor a cyclic structure.


Continuing with the comparisons, let us now compare DELP with the classical DLP in finite groups or in finite fields. Several differences are notable:
\begin{enumerate}
	\item Operations for DELP are non-associative and non-commutative operations in an entropic quasigroup $(\mathbb{E}^*_{\nu}, *)$. At the same time, DLP is exclusively defined in groups $G$ that are mostly commutative (there are also DLPs over non-commutative groups, such as the isogenies between elliptical curves defined over the finite fields).
	\item All generic algorithms for solving DLP, exclusively without exceptions, use the fact that the group $G$ is cyclic of order $N$, generated by some generator element $g$ and that for every element $y \in G$ there is a unique index $i \in \{1, \ldots, N\}$ such that $y = g^i$ (in multiplicative notion). In DELP, there are generators for the multiplicative quasigroup $\mathbb{E}^*_{(p-1)^2}$ which has order $(p-1)^2$, but $\mathbb{E}^*_{(p-1)^2}$ is not a cyclic structure since for every $y \in \mathbb{E}^*_{(p-1)^2}$ there are many indices $(\mathbf{A}, \mathfrak{b})$ such that $y = g^{(\mathbf{A}, \mathfrak{b})}$, and finding only one of them will solve the DELP. Thus, at first sight, it might seem DELP is an easier task than DLP. However, with Proposition \ref{Prop:DLP-reduces-to-DELP} (given below), we show that DLP is no harder than DELP, i.e., DLP $\leq$ DELP. 
	\item We can take a conservative approach for modeling the complexity of solving DELP, and assume that eventually, an arithmetic (logarithmetic) for the power indices in finite entropoids will be developed (Open problem \ref{Open-problem:FindExplixitRulesForLogarithmetic}). In that case, an adaptation of the generic algorithms for solving DLP, such as Baby-step giant-step, Pollard rho, Pollard kangaroo, or Pohlig–Hellman for solving DELP, will address a problem with a search space size $N \approx p^2$. Since the complexity of a generic DLP algorithm is $O(\sqrt{N})$ we get that solving DELP with classical algorithms could possibly reach a complexity as low as $O(p)$. Extending this thinking for potential quantum algorithms that will solve DELP, we estimate that their complexity could potentially be as low as $O(p^{1/2})$.
\end{enumerate}

\begin{proposition}\label{Prop:DLP-reduces-to-DELP}
	Let $p = 2 q + 1$ is a safe prime number with $\lambda$ bits, and let $\mathbb{E}_{p^2}(a_3, a_8, b_2, b_7)$ is a given entropoid. If $\mathcal{A}$ is an efficient algorithm that solves DELP, then there exist an efficient algorithm $\mathcal{B}$  that solves DLP for every subgroup $\Gamma$ of $\mathbb{F}_p^*$.
\end{proposition}
\begin{proof}
	Let us use the algorithm $\mathcal{A}$ for an entropoid $\mathbb{E}_{p}(a_3 = 0, a_8 = 1, b_2 = 0, b_7 = 1)$. In that entropoid the operation $*$ becomes $$(x_1, x_2) * (y_1, y_2) = (x_2 y_1, x_1 y_2) .$$ 
	Let $\Gamma \subseteq \mathbb{F}_p^*$ is a nontrivial subgroup. Then, since $p = 2 q + 1$ where $q$ is a prime number, $\Gamma$ is either the quadratic residue group $\Gamma = QR(p)$ of order $q$ or $\Gamma = \mathbb{F}_p^*$ of order $(p-1)$. Let $\gamma$ be a generator of $\Gamma$. Apparently $\gamma \neq 0$ i.e. $\gamma \neq  -\frac{a_3}{a_8}$ and $\gamma \neq -\frac{b_2}{b_7}$. Then from Proposition \ref{Prop:Fields-Entropoid-connection-via-subgroups-in-F} it follows that $g = (\gamma, \gamma)$ is a generator of some subgroupoid  $(\mathbb{E}^*_{\nu}, *)$,  and the operation of exponentiation of $g$ in the entropoid, reduces to exponentiation in a finite field i.e. $$g^{(\mathbf{A}, b)} = g^{(a, a_s, b)} = (\gamma^a, \gamma^a).$$
	Thus, for every received $\gamma^a$, the algorithm $\mathcal{B}$ constructs the pair $(\gamma^a, \gamma^a)$ and asks the algorithm $\mathcal{A}$ to solve it. $\mathcal{A}$ solves it efficiently and returns $(\mathbf{A}, \mathfrak{b})$, from which $\mathcal{B}$ extracts the discrete logarithm $a$.
\end{proof}

So, in its generality, and currently without the arithmetic for the succinct power indices defined with Definition \ref{Def:Power-indices}, the best algorithms for solving DELP are practically the generic algorithms for random function inversion, i.e., the generic guessing algorithms. Two of them are given as Algorithm \ref{Alg:DELP-Solver} and Algorithm \ref{Alg:DELP-Solver-Iterative}.
\begin{algorithm}
	\small
	\caption{Randomized search solver for (DELP)
		\newline
		\textbf{Input:} Entropoid $\mathbb{E}_{p^2}$, generator $g$ of $\mathbb{E}_{(p-1)^2}^*$ and $y \in \mathbb{E}_{(p-1)^2}^*$.
		\newline
		\textbf{Output:}$(\mathbf{A}, \mathfrak{b})$ such that $y = g^{(\mathbf{A}, \mathfrak{b})}$.}
	\begin{algorithmic}[1]
		\Repeat
		\State Set $(\mathbf{A}, \mathfrak{b})  \xleftarrow{\mathbf{r}} \mathbb{L}$, where $\mathfrak{b} \geq 3$;
		\Until{$y = g^{(\mathbf{A}, \mathfrak{b})}$}
		\State Return $(\mathbf{A}, \mathfrak{b})$.
	\end{algorithmic}\label{Alg:DELP-Solver}
\end{algorithm}

\begin{algorithm}
	\small
	\caption{Brute force search solver for (DELP)
		\newline
		\textbf{Input:} Entropoid $\mathbb{E}_{p^2}$, generator $g$ of $\mathbb{E}_{p^2}^*$ and $y \in \mathbb{E}_{p^2}^*$.
		\newline
		\textbf{Output:}$(\mathbf{A}, \mathfrak{b})$ such that $y = g^{(\mathbf{A}, \mathfrak{b})}$.}
	\begin{algorithmic}[1]
		\State Set $\mathfrak{b} = 2 \mathfrak{b}' + 1$, and $\mathfrak{b} \leq b_{max}$;
		\For{$a = 2$ \texttt{to} $(p-1)^2$}
			\For{$a_s \in \mathbb{L}[a, \mathfrak{b}]$}
				\If{$y = g^{(a, a_s, \mathfrak{b})}$}
					\State Return $(\mathbf{A}, \mathfrak{b}) = (a, a_s, \mathfrak{b})$.
				\EndIf
			\EndFor	 
		\EndFor	
	\end{algorithmic}\label{Alg:DELP-Solver-Iterative}
\end{algorithm}

\subsection{DELP is secure  against Shor's quantum algorithm for DLP}\label{Sec:DELP-i-secure-against-quantum-attackers}
Shor's quantum algorithm breaks algorithms that rely on the difficulty of DLP defined over finite commutative groups. One of Shor'salgorithm's crucial components is the part of its quantum circuit that calculates the modular arithmetic for raising $g$ to any power, with the repeated squaring. That part of the Shor's quantum circuit for the repeated squaring works if the related group multiplication operation is associative and commutative. There are no variants of Shor's algorithm or any other quantum algorithm that will work if the underlying algebraic structure is non-commutative. Additionally, DELP is defined over entropoids that are both non-associative and non-commutative.

A designer of a quantum algorithm for solving DELP faces two challenges:
\begin{enumerate}
	\item Build quantum circuits that implement non-commutative operations of multiplication $*$.
	\item Build quantum circuits that perform an unknown pattern of non-associative and non-commutative multiplications $*$, where the number of possible patterns is exponentially high.
\end{enumerate}

We have to note that if the used base is $\mathfrak{b}=2$, the bracketing pattern is known, and there is a possibility to "reuse" the Shor's circuit. However, as we see from Corollary \ref{Cor:Probability-To-Be-Power-Of-2} the probability that the answer from that circuit will be correct is $\frac{1}{q}$. Thus, for $q$ being 128 or 256 bits, it would be a very inefficient quantum algorithm.

\section{Concrete instances of Entropoid Based Key Exchange and Digital Signature Algorithms}

\subsection{Choosing Parameters For a Key Exchange Algorithm Based on DELP}
Based on the discussion in Section \ref{Sec:Hard-problems} for achieving post-quantum security levels of $2^{64}$ and $2^{128}$ qubit operations we propose finite entropoids $\mathbb{E}_{p^2}$ to use safe prime numbers $p$ with 128 and 256 bits. For estimating the number of $*$ operations for performing one power operation, we use the equation (\ref{Eq:Upper-bound-number-of-multiplications-in-power}). We see that the number depends on the odd base $\mathfrak{b}$. Additionally, $*$ operation with a pre-computation of expressions that involve $ a_3, a_8, b_2, b_7$ can be computed with six modular additions and six modular multiplications in $\mathbb{F}_p$. The expected number of modular operations and the total communication cost for two security levels are given in Table \ref{Tab:Efficiency-128-256}. We can see that the number of modular operations increases with $\mathfrak{b}$, while the communication costs in both directions in total are 64 and 128 bytes, respectively. 

A formal description of an unauthenticated Diffie-Hellman protocol over finite entropoids is given as follows:  
\begin{description}
	\item[Agreed public parameters]\ \\
	\vspace{-0.7cm}
	\begin{enumerate}
		\item Alice and Bob agree on parameters for $\mathbb{E}_{p^2}( a_3, a_8, b_2, b_7)$, where $p = 2 q + 1$ is a prime number with a size of $\lambda = 128$ or $\lambda = 256$ bits, $q$ is also a prime number, the values $a_3, a_8, b_2, b_7 \in \mathbb{F}_p$ and operations for non-commutative and non-associative multiplication and exponentiation are defined as in Definition \ref{Def:OperationStar} and Definition \ref{Def:Raising-to-power}
		\item Alice and Bob agree on the generator $g_q = \mathtt{GenQ}(\lambda, p, a_3, a_8, b_2, b_7)$
		\item Alice and Bob agree on odd base $\mathfrak{b}$
	\end{enumerate}
	\item[Ephemeral key exchange phase]\ \\
	\vspace{-0.7cm}
	\begin{enumerate}
		\item Alice generates a random power index $(\mathbf{A}, \mathfrak{b}) = (a, a_s, \mathfrak{b})$ where $a \in \mathbb{Z}^*_{p}$
		\item Alice computes $K_a = g_q^{(\mathbf{A}, \mathfrak{b})}$ and sends it to Bob
		\item Bob generates a random power index $(\mathbf{B}, \mathfrak{b}) = (b, b_s, \mathfrak{b})$ where $b \in \mathbb{Z}^*_{p}$
		\item Bob computes $K_b = g_q^{(\mathbf{B}, \mathfrak{b})}$ and sends it to Alice
		\item Alice computes $K_{ab} = K_b^{(\mathbf{A}, \mathfrak{b})}$
		\item Bob computes $K_{ba} = K_a^{(\mathbf{B}, \mathfrak{b})}$
		\item $K_{ab} = K_{ba}$.
	\end{enumerate}
\end{description}

\begin{table}[htbp]
  \centering
  \includegraphics{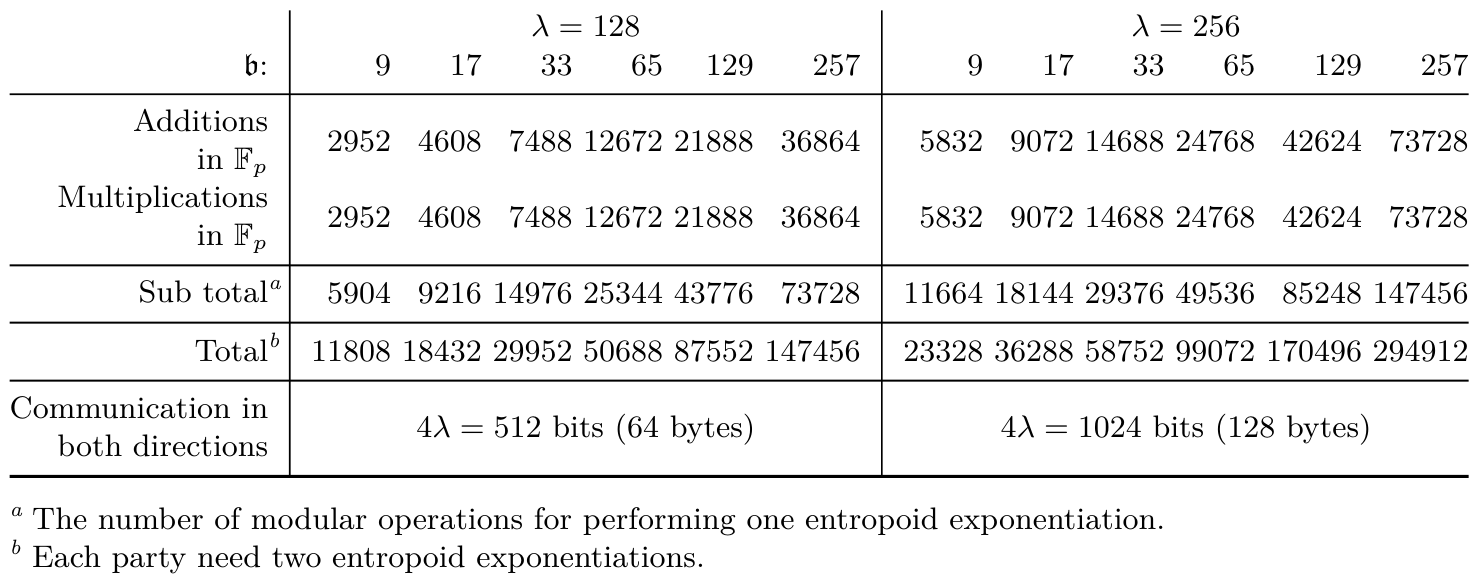}
  \caption{Some efficiency metrics for entropoid based Diffie-Hellman key exchange.}
  \label{Tab:Efficiency-128-256}%
\end{table}%

\subsection{Digital Signature Scheme Based on CDERP}
For defining a digital signature over finite entropoids, let us first fix the base $\mathfrak{b} = 257$. The reason for this is the fact that in that case for any power index $(\mathbf{B}, 257) = (b, b_s, 257)$, the bracketing pattern part $b_s = [P_0, \ldots P_k]_{256}$ has a very convenient interpretation as a list of bytes. \textbf{Remark:} Choosing a base $\mathfrak{b}=257$ might be too conservative and expensive - making the computations of exponentiation very slow. All proposed algorithms in this and the next sub-section can be carried out with $\mathfrak{b}=17$, which is also convenient since the bracketing patterns, in that case, become 4-bit nibbles. However, since this is the first introduction of a new signature scheme based on a new hardness problem, we propose more conservative parameters.

In the rest of this sub-section for the power indices instead of writing $(\mathbf{B}, 257)$ we will omit the base part and will simply write the bold letter $\mathbf{B}$ meaning $\mathbf{B} = (b, b_s, 257)$. We also assume that the public parameters $\mathbb{E}_{p^2}( a_3, a_8, b_2, b_7)$, where $p = 2 q + 1$ is a prime number with a size of $\lambda \in \{128, 192, 256\}$ bits, $q$ is also a prime number and the values $a_3, a_8, b_2, b_7 \in \mathbb{F}_p$ are known and fixed.


We will use the NIST standardized cryptographic hash function $\mathtt{SHAXXX}$ where $\mathtt{XXX} \in \{256, 384, 512\}$ and will use the following notation. Let $M$ be any message, and let $\mathtt{SHAXXX}(M) = h_1 || h_2$ where $|h_1| = |h_2| = \frac{\mathtt{XXX}}{2}$ bits. We define $truncate_\mathbb{L}(\mathtt{SHAXXX}(M)) = (h, h_s)$ where $h = h_1$ and the sequence of bits of $h_1$ are interpreted as a little-endian encoding for a number $h \in \mathbb{Z}_{2^\frac{\mathtt{XXX}}{2}}$. We partition the bits of $h_2$ in a list of $k_{max} \in \{16, 24, 32\}$ bytes $h_2 =  [P_0, \ldots P_{k_{max}-1}]_{256}$. Then we compute $k = \log_{257}h - 1$. Finally we interpret $h_s$ as a list of the first $k+1$ bytes $h_s = [P_0, \ldots P_k]_{256}$. If $k < k_{max} - 1$ we truncate the remaining bytes. The set of all such power indices  $\mathbb{L}[h, 257]$ is shortly denoted as $\mathbb{L}_{\mathtt{257}}$. With this we defined a mapping $Hash_{\mathtt{XXX}} :  \{0, 1\}^* \mapsto \mathbb{L}_{\mathtt{257}}$ as follows:
\begin{equation}
	Hash_{\mathtt{XXX}} (M) = truncate_\mathbb{L}(\mathtt{SHAXXX}(M)) = (h, h_s, 257).
\end{equation}

The signature scheme is designed by a Fiat-Shamir transformation \cite{fiat87how} of an identification scheme, and it looks similar to Schnorr identification scheme \cite{Schnorr1990}, but the security is based on the hardness of CDERP. 

\textbf{Note:} In this first version of the paper, we are not giving clear formal proof of the scheme's security in the EUF-CMA security model (Existential Unforgeability under Chosen Message Attack). That proof will be given either as a separate work or in the updated versions of this paper. Instead, we are giving here an initial discussion about the security properties of the proposed signature scheme.

Let us call the key generation algorithm $\mathtt{GenKey}$. It is given in Algorithm \ref{Alg:GenKey}. We assume that a generator $g$ for the quasigroup $\mathbb{E}^*_{(p-1)^2}$ is publicly known and standardized and that a power index $\mathbf{B} =  Hash_{\mathtt{XXX}} (\mathtt{``abc"})$ is also predetermined and fixed. The message $\mathtt{``abc"}$ can be any string such as \texttt{``This is a seed message for fixing the value of the public root."}

\begin{algorithm}
	\small
	\caption{$\mathtt{GenKey}$: Generate a key pair $(\mathtt{PrivateKey}, \mathtt{PublicKey})$.
		\newline
		\textbf{Input:} ;
		\newline
		\textbf{Output:} $(\mathtt{PrivateKey}, \mathtt{PublicKey})$.}
	\begin{algorithmic}[1]
		\State Set $x \xleftarrow{\mathbf{r}} \mathbb{E}^*_{(p-1)^2}$
		\State Set $\mathtt{PrivateKey} = x$
		\State Set $y = x^{\mathbf{B}}$
		\State Set $\mathtt{PublicKey} = y$
		\State Return $(\mathtt{PrivateKey}, \mathtt{PublicKey})$
	\end{algorithmic}\label{Alg:GenKey}
\end{algorithm}

Let us now describe the following identification scheme:

\begin{figure}[h]
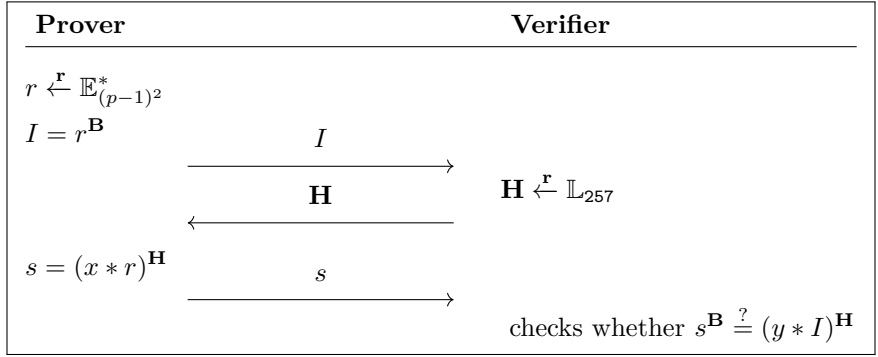

	\begin{center}
		\fbox{
			\pseudocode{%
				\textbf{ Prover} \< \< \textbf{ Verifier}  \\[0.1\baselineskip][\hline]
				\<\< \\[-0.5\baselineskip]
				r \xleftarrow{\mathbf{r}} \mathbb{E}^*_{(p-1)^2} \< \< \\
				I = r^{\mathbf{B}} \<\< \\[-4ex]
				\< \sendmessageright*{I} \< \\[-2ex]
				\<\< \mathbf{H} \xleftarrow{\mathbf{r}} \mathbb{L}_{\mathtt{257}} \\[-4ex]
				\< \sendmessageleft*{\mathbf{H}} \< \\
				s = (x * r)^\mathbf{H} \<\< \\[-4ex]
				\< \sendmessageright*{s} \< \\[-2ex]
				\<\< \text{ checks whether } s^\mathbf{B} \stackrel{?}{=} (y * I)^\mathbf{H}
			}
		}
	\end{center}
	\caption{An ID scheme based on the hardness of CDERP}\label{Fig:ID-scheme-based-on-CDERP}
\end{figure}

Since $I$ is a uniformly random element from $ \mathbb{E}^*_{(p-1)^2}$, and $\mathbf{H}$ is a uniformly random element from  $\mathbb{L}_{\mathtt{257}}$, the distribution of $s$ is also uniformly random. Thus, an attacker can simulate the transcripts of honest executions by randomly producing triplets $(I, \mathbf{H}, s)$, without a knowledge of the private key. However, since $s = \sqrt[\mathbf{B}]{(y * I)^\mathbf{H}}$, if the produced transcripts are verified and true, it implies that the attacker can compute the $\mathbf{B}$-th root, i.e., can solve the CDERP. From this discussion, we give (without proof) the following Theorem:

\begin{theorem}
	If the computational discrete entropoid root problem is hard in $ \mathbb{E}^*_{(p-1)^2}$, then the identification scheme given in Figure \ref{Fig:ID-scheme-based-on-CDERP} is secure.
\end{theorem}

The Fiat-Shamir transformation of the identification scheme presented in Figure \ref{Fig:ID-scheme-based-on-CDERP} gives a digital signature scheme with three algorithms: $\mathtt{GenKey}$ for key generation (already presented in Algorithm \ref{Alg:GenKey}), $\mathtt{Sign}$ for digital signing (given in Algorithm \ref{Alg:Sign}) and $\mathtt{Verify}$ for signature verification (given in Algorithm \ref{Alg:Verify}). 


\begin{algorithm}
	\small
	\caption{$\mathtt{Sign}$: Sign a message.
		\newline
		\textbf{Input:} A message $M$, and $\mathtt{PrivateKey} = x$;
		\newline
		\textbf{Output:} $(M, \sigma)$ where $\sigma = (I, s)$ is the  digital signature of the message $M$.}
	\begin{algorithmic}[1]
		\State Set $r \xleftarrow{\mathbf{r}} \mathbb{E}^*_{(p-1)^2}$
		\State Set $I = r^{\mathbf{B}}$
		\State Set $\mathbf{H} = Hash_{\mathtt{XXX}} (I || M)$
		\State Set $s = (x * r)^\mathbf{H}$
		\State Set $\sigma = (I, s)$
		\State Return $(M, \sigma)$
	\end{algorithmic}\label{Alg:Sign}
\end{algorithm}
\begin{algorithm}
	\small
	\caption{$\mathtt{Verify}$: Verify a digital signature.
		\newline
		\textbf{Input:} A pair $(M, \sigma)$, and $\mathtt{PublicKey} = y$;
		\newline
		\textbf{Output:} \texttt{True} or \texttt{False}.}
	\begin{algorithmic}[1]
		\State Set $\mathbf{H} = Hash_{\mathtt{XXX}} (I || M)$
		\If{$s^\mathbf{B} \stackrel{?}{=} (y * I)^\mathbf{H}$}
		\State Return \texttt{True}
		\Else
		\State Return \texttt{False}
		\EndIf
	\end{algorithmic}\label{Alg:Verify}
\end{algorithm}

We check the correctness of the signature scheme as follows: 
$$s^\mathbf{B} = \big((x * r)^\mathbf{H}\big)^\mathbf{B} =  \big(x^\mathbf{H}\big)^\mathbf{B} * \big(r^\mathbf{H}\big)^\mathbf{B} = \big(x^\mathbf{B}\big)^\mathbf{H} * \big(r^\mathbf{B}\big)^\mathbf{H} =  y^\mathbf{H} * I^\mathbf{H} = (y * I)^\mathbf{H}.$$

An attacker can forge signatures in one of the following ways
\begin{enumerate}
	\item For an existing pair $(M, \sigma)$, find a second preimage $I || M'$ such that $\mathbf{H} = Hash_{\mathtt{XXX}} (I || M')$.  In that case $(M', \sigma)$ is a valid pair.
	\item Compute a discrete entropoid $\mathbf{B}$-root of $y$. In that case the attacker will know the $\mathtt{PrivateKey} = x = \sqrt[\mathbf{B}]{y}$.
	\item Generate random $I$, random message $M$, and compute the corresponding $\mathbf{H} = Hash_{\mathtt{XXX}} (I || M)$. Then compute a discrete entropoid $\mathbf{B}$-root $s = \sqrt[\mathbf{B}]{z}$ where $z = (y * I)^\mathbf{H}$.
\end{enumerate}

We want to emphasize that finding a collision $\mathbf{H} = Hash_{\mathtt{XXX}} (I_1 || M_1) =  Hash_{\mathtt{XXX}} (I_2 || M_2)$ is not enough to forge a signature since the attacker in order to produce $s$ has to perform the operation of computing a discrete entropoid $\mathbf{B}$-root for $x = \sqrt[\mathbf{B}]{y}$ or for $s = \sqrt[\mathbf{B}]{z}$.

However, there is a collision finding strategy that will help the attacker to simulate the $\mathbf{B}$-root computation  and that is the classical Diffie and Hellman Meet-in-the-middle attack \cite{diffie1977special}. For achieving the goal  of having a probability 1/2 for finding  $\mathbf{B}$-root $s = \sqrt[\mathbf{B}]{z}$ where $z = (y * \sigma_1)^\mathbf{H}$ the attacker needs to build two tables $T_1$ and $T_2$ where $T_1$ will contain pairs of elements from $\mathbb{E}^*_{(p-1)^2}$ and $T_2$ will contain quadruples as described below:
\small{
	\begin{align}
		T_1  = &  [(z_i, z_i^\mathbf{B}) \ | \ \text{for } z_i  \xleftarrow{\mathbf{r}}  \mathbb{E}^*_{(p-1)^2}, i \in \{1, \ldots p-1\} ] \\
		\text{and} \nonumber \\
		T_2  = &  [(I_i, M_i, \mathbf{H}_i, (y * I_i)^{\mathbf{H}_i}) \ | \ \text{for } I_i \xleftarrow{\mathbf{r}} \mathbb{E}^*_{(p-1)^2}, M_i \xleftarrow{\mathbf{r}} \{0, 1\}^*, \mathbf{H}_i = Hash_{\mathtt{XXX}} (I_i || M_i), i \in \{1, \ldots p-1\} ]
	\end{align}
}

During the build-up of the tables if the attacker is lucky, it can even find an entry in $T_1$ that has an item $z_i^\mathbf{B}  = y$. In that case it found $\mathbf{B}$-root for $x = \sqrt[\mathbf{B}]{y} = z_i$. For the size of $T_1$ being $p-1$, the probability of this event is $\frac{p-1}{(p-1)^2} = \frac{1}{p-1}$. The attacker can also search for collisions $z_i^\mathbf{B}  = (y * I_j)^{\mathbf{H}_j}$. If that happens, it would found $\mathbf{B}$-root for $s = \sqrt[\mathbf{B}]{z_i}$. For the size of $T_1$ and $T_2$ being $p-1$, the probability of this event is around 0.5. So the memory complexity for this attack is $O(2p) = O(2^{\lambda+1})$ and the time complexity is also $O(2^{\lambda+1})$.

\begin{table}[htbp]
	\centering
	
	\begin{tabular}{|c|c|c|c|c|c|}
		\hline
		\multicolumn{1}{|p{4.5em}|}{\centering {\footnotesize $\lambda$, $\mathbb{F}_p$, $p = 2 q +1$,  $\lambda = \lceil  \log p \rceil$} } & \multicolumn{1}{p{5.5em}|}{\centering EUF-CMA classical security} & \multicolumn{1}{p{5.5em}|}{\centering EUF-CMA quantum security} & \multicolumn{1}{p{5em}|}{\centering PublicKey size (bytes)} & \multicolumn{1}{p{5em}|}{\centering PrivateKey size (bytes)} & \multicolumn{1}{p{5em}|}{\centering Signature size (bytes)} \bigstrut\\
		\hline
		128             & $2^{128}$         & $2^{85}$          & 32              & 32              & 64 \bigstrut\\
		\hline
		192             & $2^{192}$         & $2^{128}$         & 48              & 48              & 96 \bigstrut\\
		\hline
		256             & $2^{256}$         & $2^{171}$         & 64              & 64              & 128 \bigstrut\\
		\hline
	\end{tabular}%
	\caption{A summary table for the characteristics of the entropoid digital signature scheme based on CDERP}
	\label{Tab:EntropoidDigitalSignatures}%
\end{table}%

For a similar quantum collision search, we first have to assume that the attacker has overcome the challenges discussed at the end of Section \ref{Sec:DELP-i-secure-against-quantum-attackers}. While in this situation the associative pattern $\mathbf{B}$ used in table $T_1$ is known and fixed, for every entry in table $T_2$ the associative patterns $\mathbf{H}_i$ are entangled with the choices of $I_i$ and $M_i$, and the output of the hash function $\mathbf{H}_i = Hash_{\mathtt{XXX}} (I_i || M_i)$. So, the attacker faces again the two challenges: To build quantum circuits that implement non-commutative operations of multiplication $*$, and to build quantum circuits that perform an exponential number of non-associative and non-commutative multiplication patterns. If it overcomes those challenges, we can assume that the complexity for the quantum search \cite{tani2009claw} for entropoid collisions could be as low as $O(q^{\frac{2}{3}})$.

As a summary of all discussion in this section, we give a Table \ref{Tab:EntropoidDigitalSignatures}.

\subsection{Digital Signature Scheme Based on reducing CDERP to DELP for Specific Roots $\mathbf{B}$}
The signature scheme proposed in the previous sub-section is based on the new assumption about the hardness of CDERP. In case that assumption turns out to be false, we propose here an alternative and more conservative signature scheme that relies its security on the hardness of solving the discrete entropoid logarithm problem.

The more conservative scheme is similar to the one given in the previous sub-section, with the following differences. The entropoid $\mathbb{E}_{p^2}( a_3, a_8, b_2, b_7)$, is defined with a prime number $p = 2 q + 1$ where the bit size of $p$ is $\lambda \in \{256, 384, 512\}$ bits. The root $\mathbf{B}$ as a public parameter has the following format:\footnote{The same remark about the possibility to use base $\mathfrak{b}=17$ applies also for this signature scheme.} $\mathbf{B} = (q, b_s, 257)$, where $q$ is the prime number in the construction of $p$, and $b_s = \mathtt{SHAXXX} (\mathtt{``abc"})$. The mapping $Hash_{\mathtt{XXX}} :  \{0, 1\}^* \mapsto \mathbb{L}_{q, \mathtt{257}}$ is now defined as:
\begin{equation}
	Hash_{\mathtt{XXX}} (M) = (q, \mathtt{SHAXXX}(M), 257).
\end{equation}

Now, with this different hashing, the algorithms $\mathtt{GenKey}$, $\mathtt{Sign}$ and $\mathtt{Verify}$ are the same as in the previous case.

The same security analysis applies here, with one additional safety layer. Let us suppose that a Logarithmetic for our succinct power indices will be developed and that the Johnston root-finding algorithm will be adapted for the entropoids $\mathbb{E}_{p^2}$. Since the size of maximal multiplicative groupoid  $\mathbb{E}_{(p-1)^2}$ is $4 q^2$, and since we have a fixed value $q$ in the public root value $\mathbf{B} = (q, b_s, 257)$, even the hypothetical version of the Johnston algorithm will reduce to finding discrete entropoid logarithm in $\mathbb{E}_{(p-1)^2}$.

The consequences of doubling the bit sizes of $p$ are the doubling of the keys and signatures of the proposed signature scheme and are given in Table \ref{Tab:EntropoidDigitalSignatures02}.

\begin{table}[htbp]
	\centering
	
	\begin{tabular}{|c|c|c|c|c|c|}
		\hline
		\multicolumn{1}{|p{4.5em}|}{\centering {\footnotesize $\lambda$, $\mathbb{F}_p$, $p = 2 q +1$,  $\lambda = \lceil  \log p \rceil$} } & \multicolumn{1}{p{5.5em}|}{\centering EUF-CMA classical security} & \multicolumn{1}{p{5.5em}|}{\centering EUF-CMA quantum security} & \multicolumn{1}{p{5em}|}{\centering PublicKey size (bytes)} & \multicolumn{1}{p{5em}|}{\centering PrivateKey size (bytes)} & \multicolumn{1}{p{5em}|}{\centering Signature size (bytes)} \bigstrut\\
		\hline
		256             & $2^{128}$         & $2^{85}$          & 64              & 64              & 128 \bigstrut\\
		\hline
		384             & $2^{192}$         & $2^{128}$         & 96              & 96              & 192 \bigstrut\\
		\hline
		512             & $2^{256}$         & $2^{171}$         & 128             & 128              & 256 \bigstrut\\
		\hline
	\end{tabular}%
	\caption{A summary table for the characteristics of the entropoid digital signature scheme based on reduction of CDERP to DELP}
	\label{Tab:EntropoidDigitalSignatures02}%
\end{table}%


\section{Conclusions}

The algebraic structures that are non-commutative and non-associative known as entropic groupoids that satisfy the \emph{"Palintropic"} property i.e., $x^{\mathbf{A} \mathbf{B}} = (x^{\mathbf{A}})^{\mathbf{B}} = (x^{\mathbf{B}})^{\mathbf{A}} = x^{\mathbf{B} \mathbf{A}}$ were proposed by Etherington in '40s from the 20th century. Those relations are exactly the Diffie-Hellman key exchange protocol relations used with groups. The arithmetic for non-associative power indices known as Logarithmetic was also proposed by Etherington and later developed by others in the period of '50s-'70s. However, there was never proposed a succinct notation for exponentially large non-associative power indices that will have the property of fast exponentiation similarly as the fast exponentiation is achieved with ordinary arithmetic via the consecutive rising to the powers of two. 

In this paper, we defined ringoid algebraic structures $(G, \boxplus, *)$ where $(G, \boxplus) $ is an Abelian group and $(G, *)$ is a non-commutative and non-associative groupoid with an entropic and palintropic subgroupoid which is a quasigroup, and we named those structures as Entropoids. We further defined succinct notation for non-associative bracketing patterns and proposed algorithms for fast exponentiation with those patterns. 

Next, by analogy with the developed cryptographic theory of discrete logarithm problems, we defined several hard problems in Entropoid based cryptography, and based on that, we proposed an entropoid Diffie-Hellman key exchange protocol and an entropoid signature schemes. Due to the non-commutativity and non-associativity, the entropoid based cryptographic primitives are supposed to be resistant to quantum algorithms. At the same time, due to the proposed succinct notation for the power indices, the communication overhead in the entropoid based Diffie-Hellman key exchange is very low: for 128 bits of security, 64 bytes in total are communicated in both directions, and for 256 bits of security, 128 bytes in total are communicated in both directions. 

In this paper, we also proposed two entropoid based digital signature schemes. The schemes are constructed with the Fiat-Shamir transformation of an identification scheme which security relies on a new hardness assumption: computing roots in finite entropoids is hard. If this assumption withstands the time's test, the first proposed signature scheme has very attractive properties: for the classical security levels between 128 and 256 bits, the public and private key sizes are between 32 and 64, and the signature sizes are between 64 and 128 bytes. The second signature scheme reduces the finding of the roots in finite entropoids to computing discrete entropoid logarithms. In our opinion, this is a safer but more conservative design, and the price is in doubling the key sizes and the signature sizes.

We give a proof-of-concept implementation in SageMath 9.2 for all proposed algorithms and schemes in Appendix C.

We hope that this paper will initiate further research in Entropoid Based Cryptography.


\newpage
\appendix
\section{Examples for $\mathbb{E}_{{11}^2}$, $\mathbb{E}_{{13}^2}$, $\mathbb{E}_{{19}^2}$ and $\mathbb{E}_{{23}^2}$}\label{App:Examples-for-conjecture}

	Let us define the following finite entropoids: 
	\begin{enumerate}
		\item $\mathbb{E}_{{11}^2}(a_3=9, a_8=1, b_2=8, b_7=9)$, which has $\mathbf{0}_* =  (2, 4)$, $\mathbf{1}_* = (7, 5)$ and $x * y = (x_1, x_2) * (y_1, y_2) = ( x_2 y_1 + 9 x_2 + 7 y_1 + 10, \ \ \ 9 x_1 y_2 + 8 x_1 + 4 y_2 + 10)$;
		\item $\mathbb{E}_{{13}^2}(a_3=10, a_8=2, b_2=3, b_7=9)$, which has $\mathbf{0}_* =  (8, 4)$, $\mathbf{1}_* = (11, 11)$ and $x * y = (x_1, x_2) * (y_1, y_2) = ( 2 x_2 y_1 + 10 x_2 + 5 y_1 + 7, \ \ \ 9 x_1 y_2 + 3 x_1 + 6 y_2 + 6)$;
		\item $\mathbb{E}_{{19}^2}(a_3=18, a_8=11, b_2=14, b_7=10)$, which has $\mathbf{0}_* =  (7, 10)$, $\mathbf{1}_* = (9, 17)$ and $x * y = (x_1, x_2) * (y_1, y_2) = ( 11 x_2 y_1 + 18 x_2 + 4 y_1 + 17, \ \ \ 10 x_1 y_2 + 14 x_1 + 6 y_2 + 7)$;
		\item $\mathbb{E}_{{23}^2}(a_3=15, a_8=13, b_2=9, b_7=14)$, which has $\mathbf{0}_* =  (13, 1)$, $\mathbf{1}_* = (18, 17)$ and $x * y = (x_1, x_2) * (y_1, y_2) = ( 13 x_2 y_1 + 15 x_2 + 10 y_1 + 21, \ \ \ 14 x_1 y_2 + 9 x_1 + 2 y_2 + 22)$.
	\end{enumerate}
	
	We present their elements as square $p\times p$ arrays as in Example \ref{Ex:GF(7)} and in cells with coordinates $(x_1, x_2)$ we put the values that are the size of the sets $\langle x \rangle_2$ and $\langle x \rangle$.
	
	The colored cells has the following meaning:
	\begin{enumerate}
		\item The yellow highlighted elements do not belong to the multiplicative quasigroup $(\mathbb{E}_{{p}^2}^*, *)$;
		\item The green highlighted elements $x = (x_1, x_2) $ are generators for both the maximal length cyclic subgroupoid $\langle x \rangle_2$ with $2 (p - 1)$ elements, and are generators of the multiplicative quasigroup $(\mathbb{E}_{{p}^2}^*, *)$;
		\item The red highlighted elements $x = (x_1, x_2) $ are generators for a maximal length cyclic subgroupoid $\langle x \rangle_2$ with $2 (p - 1)$ elements, \textbf{but are not} generators of the multiplicative quasigroup $(\mathbb{E}_{{p}^2}^*, *)$.
		\item Blue highlighted element for  $\mathbb{E}_{{11}^2}$ and $\mathbb{E}_{{23}^2}$ denote the generators of the Sylow  $q$-subgroupoids with 25 and 121 elements (11 and 23 are "safe primes" i.e. $11 = 2 \times 5 + 1$ and $23 = 2 \times 11 + 1$). 
	\end{enumerate}
	
	
	\arrayrulecolor{cyan}	
	\begin{table}[!ht]
		\begin{minipage}{.5\linewidth}
			\centering
			{\footnotesize
			$
			\setlength{\extrarowheight}{4mm}
			\begin{NiceArray}{rrrrrrrrrrr}%
				[small, hvlines, first-row, first-col,code-before ={\cellcolor{red}{1-1, 2-3, 4-7, 5-9, 6-11, 7-2, 10-8, 11-10}
					\cellcolor{green}{1-4, 1-7, 1-11, 2-6, 2-8, 2-9, 4-1, 4-2, 4-4, 5-3, 5-6, 5-10, 6-1, 6-2, 6-4, 7-4, 7-7, 7-11, 8-1, 8-2, 8-7, 8-11, 9-3, 9-8, 9-9, 9-10, 10-3, 10-6, 10-10, 11-6, 11-8, 11-9} 	
					\rowcolor{yellow}{3}  
					\columncolor{yellow}{5}},
				columns-width = 4mm
				]%
				&  0 &  1 &  2 &  3  & 4 &  5 &  6 &  7 &  8 &  9 & 10 \\
				0  & 20 &  4 & 10 & 20  & 2 & 10 & 20 &  2 &  5 & 10 & 20 \\
				1  & 10 & 10 & 20 & 10  & 2 & 20 & 10 & 20 & 20 &  4 &  2 \\
				2  &  2 &  2 &  2 &  2  & 1 &  2 &  2 &  2 &  2 &  2 &  2 \\
				3  & 20 & 20 &  5 & 20  & 2 & 10 & 20 & 10 & 10 &  2 &  4 \\
				4  & 10 &  2 & 20 & 10  & 2 & 20 & 10 &  4 & 20 & 20 & 10 \\
				5  & 20 & 20 &  2 & 20  & 2 & 10 &  4 & 10 & 10 &  5 & 20 \\
				6  &  4 & 20 & 10 & 20  & 2 & 10 & 20 &  5 &  2 & 10 & 20 \\
				7  & 20 & 20 & 10 &  4  & 2 &  1 & 20 & 10 & 10 & 10 & 20 \\
				8  & 10 & 10 & 20 &  2  & 2 &  4 & 10 & 20 & 20 & 20 & 10 \\
				9  &  2 & 10 & 20 & 10  & 2 & 20 & 10 & 20 &  4 & 20 & 10 \\
				10  & 10 & 10 &  4 & 10  & 2 & 20 &  2 & 20 & 20 & 20 & 10 \\
			\end{NiceArray}
			$
			}
			\caption{The size of the sets $\langle x \rangle_2$ for $x \in \mathbb{E}_{{11}^2}$}\label{Table-E11-a}
		\end{minipage}%
		\begin{minipage}{.5\linewidth}
			\centering
			{\footnotesize			
			$
			\setlength{\extrarowheight}{4mm}
			\begin{NiceArray}{rrrrrrrrrrr}%
				[small, hvlines, first-row, first-col,code-before ={\cellcolor{red}{1-1, 2-3, 4-7, 5-9, 6-11, 7-2, 10-8, 11-10}
					\cellcolor{green}{1-4, 1-7, 1-11, 2-6, 2-8, 2-9, 4-1, 4-2, 4-4, 5-3, 5-6, 5-10, 6-1, 6-2, 6-4, 7-4, 7-7, 7-11, 8-1, 8-2, 8-7, 8-11, 9-3, 9-8, 9-9, 9-10, 10-3, 10-6, 10-10, 11-6, 11-8, 11-9}
					\cellcolor{blue!20}{1-3, 1-6, 1-10, 4-6, 4-8, 4-9, 6-6, 6-8, 6-9, 7-3, 7-6, 7-10, 8-3, 8-8, 8-9, 8-10} 	
					\rowcolor{yellow}{3}  
					\columncolor{yellow}{5}},
				columns-width = 4mm
				]%
				&   0 &   1 &   2 &   3  &  4 &   5 &   6 &   7 &   8 &   9 &  10 \\
				0   &  20 &  20 &  25 & 100  &  2 &  25 & 100 &   5 &   5 &  25 & 100 \\
				1   &  50 &  50 &  20 &  50  &  2 & 100 &  10 & 100 & 100 &  20 &  10 \\
				2   &   2 &   2 &   2 &   2  &  1 &   2 &   2 &   2 &   2 &   2 &   2 \\
				3   & 100 & 100 &   5 & 100  &  2 &  25 &  20 &  25 &  25 &   5 &  20 \\
				4   &  10 &  10 & 100 &  50  &  2 & 100 &  50 &  20 &  20 & 100 &  50 \\
				5   & 100 & 100 &   5 & 100  &  2 &  25 &  20 &  25 &  25 &   5 &  20 \\
				6   &  20 &  20 &  25 & 100  &  2 &  25 & 100 &   5 &   5 &  25 & 100 \\
				7   & 100 & 100 &  25 &   4  &  2 &   1 & 100 &  25 &  25 &  25 & 100 \\
				8   &  50 &  50 & 100 &   2  &  2 &   4 &  50 & 100 & 100 & 100 &  50 \\
				9   &  10 &  10 & 100 &  50  &  2 & 100 &  50 &  20 &  20 & 100 &  50 \\
				10  &  50 &  50 &  20 &  50  &  2 & 100 &  10 & 100 & 100 &  20 &  10 \\
			\end{NiceArray}
			$
			}
			\caption{The size of the sets $\langle x \rangle$ for $x \in \mathbb{E}_{{11}^2}$}\label{Table-E11-b}
		\end{minipage} 
	\end{table}
	\arrayrulecolor{black}
	
	\arrayrulecolor{cyan}	
	\begin{table}[!ht]
		\begin{minipage}{.5\linewidth}
			\centering
			{\footnotesize			
			$
			\setlength{\extrarowheight}{4mm}
			\begin{NiceArray}{rrrrrrrrrrrrr}%
				[small, hvlines, first-row, first-col,code-before ={\cellcolor{red}{1-3, 1-7, 2-10, 2-13, 3-10, 3-13, 4-3, 4-7, 5-4, 5-6, 8-2, 8-8, 10-2, 10-8, 13-4, 13-6}
					\cellcolor{green}{1-11, 1-12, 2-11, 2-12, 3-11, 3-12, 4-11, 4-12, 5-1, 5-9, 6-2, 6-4, 6-6, 6-8, 7-3, 7-7, 7-10, 7-13, 8-1, 8-9, 10-1, 10-9, 11-3, 11-7, 11-10, 11-13, 12-2, 12-4, 12-6, 12-8, 13-1, 13-9} 
					\rowcolor{yellow}{9}  
					\columncolor{yellow}{5}},
				columns-width = 4mm
				]%
				&  0 &  1 &  2 &  3  & 4 &  5 &  6 &  7 &  8 &  9 & 10 & 11 & 12 \\
				0  & 12 &  6 & 24 &  2 &  2 &  4 & 24 & 12 &  6 &  8 & 24 & 24 &  8 \\
				1  & 12 &  2 &  8 &  6 &  2 & 12 &  8 &  4 &  6 & 24 & 24 & 24 & 24 \\
				2  &  6 &  4 &  8 & 12 &  2 &  6 &  8 &  2 & 12 & 24 & 24 & 24 & 24 \\
				3  &  6 & 12 & 24 &  4 &  2 &  2 & 24 &  6 & 12 &  8 & 24 & 24 &  8 \\
				4  & 24 &  8 &  2 & 24 &  2 & 24 &  4 &  8 & 24 & 12 & 12 &  6 &  3 \\
				5  &  8 & 24 & 12 & 24 &  2 & 24 &  6 & 24 &  8 &  6 &  2 &  4 & 12 \\
				6  &  2 & 12 & 24 & 12 &  2 &  6 & 24 &  6 &  4 & 24 &  8 &  8 & 24 \\
				7  & 24 & 24 & 12 &  8 &  2 &  8 &  6 & 24 & 24 &  2 &  6 & 12 &  4 \\
				8  &  2 &  2 &  2 &  2 &  1 &  2 &  2 &  2 &  2 &  2 &  2 &  2 &  2 \\
				9  & 24 & 24 &  3 &  8 &  2 &  8 & 12 & 24 & 24 &  4 & 12 &  6 &  2 \\
				10  &  4 &  6 & 24 &  6 &  2 & 12 & 24 & 12 &  2 & 24 &  8 &  8 & 24 \\
				11  &  8 & 24 &  6 & 24 &  2 & 24 & 12 & 24 &  8 & 12 &  4 &  1 &  6 \\
				12  & 24 &  8 &  4 & 24 &  2 & 24 &  2 &  8 & 24 &  6 &  6 & 12 & 12 \\
			\end{NiceArray}
			$
			}
			\caption{The size of the sets $\langle x \rangle_2$ for $x \in \mathbb{E}_{{13}^2}$}\label{Table-E13-a}
		\end{minipage}%
		\begin{minipage}{.5\linewidth}
			\centering
			{\footnotesize			
			$
			\setlength{\extrarowheight}{4mm}
			\begin{NiceArray}{rrrrrrrrrrrrr}%
				[small, hvlines, first-row, first-col,code-before ={\cellcolor{red}{1-3, 1-7, 2-10, 2-13, 3-10, 3-13, 4-3, 4-7, 5-4, 5-6, 8-2, 8-8, 10-2, 10-8, 13-4, 13-6}
					\cellcolor{green}{1-11, 1-12, 2-11, 2-12, 3-11, 3-12, 4-11, 4-12, 5-1, 5-9, 6-2, 6-4, 6-6, 6-8, 7-3, 7-7, 7-10, 7-13, 8-1, 8-9, 10-1, 10-9, 11-3, 11-7, 11-10, 11-13, 12-2, 12-4, 12-6, 12-8, 13-1, 13-9} 
					\rowcolor{yellow}{9}  
					\columncolor{yellow}{5}},
				columns-width = 4mm
				]%
				&   0 &   1 &   2 &   3  &  4 &   5 &   6 &   7 &   8 &   9 &  10 &  11 &  12  \\
				0   &  36 &  12 &  48 &  12 &   2 &  12 &  48 &  12 &  36 &  48 & 144 & 144 &  48  \\
				1   &  36 &  12 &  48 &  12 &   2 &  12 &  48 &  12 &  36 &  48 & 144 & 144 &  48  \\
				2   &  36 &  12 &  48 &  12 &   2 &  12 &  48 &  12 &  36 &  48 & 144 & 144 &  48  \\
				3   &  36 &  12 &  48 &  12 &   2 &  12 &  48 &  12 &  36 &  48 & 144 & 144 &  48  \\
				4   & 144 &  48 &   3 &  48 &   2 &  48 &  12 &  48 & 144 &  12 &  36 &   9 &   3  \\
				5   &  16 & 144 &  36 & 144 &   2 & 144 &  18 & 144 &  16 &  18 &   2 &   4 &  36  \\
				6   &   4 &  36 & 144 &  36 &   2 &  36 & 144 &  36 &   4 & 144 &  16 &  16 & 144  \\
				7   & 144 &  48 &  12 &  48 &   2 &  48 &   6 &  48 & 144 &   6 &  18 &  36 &  12  \\
				8   &   2 &   2 &   2 &   2 &   1 &   2 &   2 &   2 &   2 &   2 &   2 &   2 &   2  \\
				9   & 144 &  48 &   3 &  48 &   2 &  48 &  12 &  48 & 144 &  12 &  36 &   9 &   3  \\
				10  &   4 &  36 & 144 &  36 &   2 &  36 & 144 &  36 &   4 & 144 &  16 &  16 & 144  \\
				11  &  16 & 144 &   9 & 144 &   2 & 144 &  36 & 144 &  16 &  36 &   4 &   1 &   9  \\
				12  & 144 &  48 &  12 &  48 &   2 &  48 &   6 &  48 & 144 &   6 &  18 &  36 &  12  \\
			\end{NiceArray}
			$
			}
			\caption{The size of the sets $\langle x \rangle$ for $x \in \mathbb{E}_{{13}^2}$}\label{Table-E13-b}
		\end{minipage} 
	\end{table}
	\arrayrulecolor{black}

	\arrayrulecolor{cyan}
	\begin{table}[!ht]
		\centering
		{\footnotesize			
		$
		\setlength{\extrarowheight}{4mm}
		\begin{NiceArray}{rrrrrrrrrrrrrrrrrrr}%
			[small, hvlines, first-row, first-col,code-before ={\cellcolor{red}{1-2, 1-5, 1-7, 2-6, 2-13, 2-14, 4-6, 4-13, 4-14, 7-2, 7-5, 7-7, 9-1, 9-15, 9-17, 12-8, 12-9, 12-16, 14-8, 14-9, 14-16, 15-1, 15-15, 15-17, 16-2, 16-5, 16-7, 17-8, 17-9, 17-16, 18-6, 18-13, 18-14, 19-1, 19-15, 19-17} 
				\cellcolor{green}{1-4, 1-10, 1-19, 2-4, 2-10, 2-19, 3-2, 3-5, 3-6, 3-7, 3-13, 3-14, 4-4, 4-10, 4-19, 5-1, 5-8, 5-9, 5-15, 5-16, 5-17, 6-1, 6-8, 6-9, 6-15, 6-16, 6-17, 7-4, 7-10, 7-19, 9-3, 9-12, 9-18, 10-2, 10-5, 10-6, 10-7, 10-13, 10-14, 11-2, 11-5, 11-6, 11-7, 11-13, 11-14, 12-3, 12-12, 12-18, 13-1, 13-8, 13-9, 13-15, 13-16, 13-17, 14-3, 14-12, 14-18, 15-3, 15-12, 15-18, 16-4, 16-10, 16-19, 17-3, 17-12, 17-18, 18-4, 18-10, 18-19, 19-3, 19-12, 19-18} 
				\rowcolor{yellow}{8}  
				\columncolor{yellow}{11}},
			columns-width = 4mm
			]%
			&  0 &  1 &  2 &  3  & 4 &  5 &  6 &  7 &  8 &  9 & 10 & 11 & 12 & 13 & 14 & 15 & 16 & 17 & 18  \\
			0  & 18 & 36 & 18 & 36 & 36 & 12 & 36 &  6 &  2 & 36 &  2 & 18 &  4 & 12 &  9 &  6 & 18 & 18 & 36  \\
			1  &  6 & 12 & 18 & 36 & 12 & 36 &  4 & 18 &  9 & 36 &  2 & 18 & 36 & 36 &  2 & 18 &  6 & 18 & 36  \\
			2  & 18 & 36 &  3 & 12 & 36 & 36 & 36 & 18 & 18 &  4 &  2 &  2 & 36 & 36 & 18 & 18 & 18 &  6 & 12  \\
			3  &  6 & 12 & 18 & 36 &  4 & 36 & 12 & 18 & 18 & 36 &  2 & 18 & 36 & 36 &  6 &  9 &  2 & 18 & 36  \\
			4  & 36 & 18 &  4 &  6 & 18 & 18 & 18 & 36 & 36 &  6 &  2 & 12 & 18 & 18 & 36 & 36 & 36 & 12 &  2  \\
			5  & 36 & 18 & 12 &  2 & 18 & 18 & 18 & 36 & 36 &  6 &  2 & 12 & 18 & 18 & 36 & 36 & 36 &  4 &  6  \\
			6  & 18 & 36 & 18 & 36 & 36 &  4 & 36 &  6 &  6 & 36 &  2 & 18 & 12 & 12 & 18 &  2 &  9 & 18 & 36  \\
			7  &  2 &  2 &  2 &  2 &  2 &  2 &  2 &  2 &  2 &  2 &  1 &  2 &  2 &  2 &  2 &  2 &  2 &  2 &  2  \\
			8  & 36 & 18 & 36 & 18 & 18 &  2 & 18 & 12 & 12 & 18 &  2 & 36 &  6 &  6 & 36 &  4 & 36 & 36 & 18  \\
			9  & 18 & 36 &  6 &  4 & 36 & 36 & 36 & 18 & 18 & 12 &  2 &  6 & 36 & 36 & 18 & 18 & 18 &  1 & 12  \\
			10  & 18 & 36 &  2 & 12 & 36 & 36 & 36 & 18 & 18 & 12 &  2 &  3 & 36 & 36 & 18 & 18 & 18 &  6 &  4  \\
			11  & 12 &  6 & 36 & 18 &  2 & 18 &  6 & 36 & 36 & 18 &  2 & 36 & 18 & 18 & 12 & 36 &  4 & 36 & 18  \\
			12  & 36 & 18 & 12 &  6 & 18 & 18 & 18 & 36 & 36 &  2 &  2 &  4 & 18 & 18 & 36 & 36 & 36 & 12 &  6  \\
			13  & 12 &  6 & 36 & 18 &  6 & 18 &  2 & 36 & 36 & 18 &  2 & 36 & 18 & 18 &  4 & 36 & 12 & 36 & 18  \\
			14  & 36 & 18 & 36 & 18 & 18 &  6 & 18 & 12 &  4 & 18 &  2 & 36 &  2 &  6 & 36 & 12 & 36 & 36 & 18  \\
			15  &  9 & 36 & 18 & 36 & 36 & 12 & 36 &  2 &  6 & 36 &  2 & 18 & 12 &  4 & 18 &  6 & 18 & 18 & 36  \\
			16  &  4 &  2 & 36 & 18 &  6 & 18 &  6 & 36 & 36 & 18 &  2 & 36 & 18 & 18 & 12 & 36 & 12 & 36 & 18  \\
			17  &  2 &  4 & 18 & 36 & 12 & 36 & 12 &  9 & 18 & 36 &  2 & 18 & 36 & 36 &  6 & 18 &  6 & 18 & 36  \\
			18  & 36 & 18 & 36 & 18 & 18 &  6 & 18 &  4 & 12 & 18 &  2 & 36 &  6 &  2 & 36 & 12 & 36 & 36 & 18  \\
		\end{NiceArray}
		$
		}
		\caption{The size of the sets $\langle x \rangle_2$ for $x \in \mathbb{E}_{{19}^2}$}\label{Table-E19-a}
	\end{table}
	\begin{table}[!ht]
		\centering
		{\footnotesize			
		$
		\setlength{\extrarowheight}{4mm}
		\begin{NiceArray}{rrrrrrrrrrrrrrrrrrr}%
			[small, hvlines, first-row, first-col,code-before ={\cellcolor{red}{1-2, 1-5, 1-7, 2-6, 2-13, 2-14, 4-6, 4-13, 4-14, 7-2, 7-5, 7-7, 9-1, 9-15, 9-17, 12-8, 12-9, 12-16, 14-8, 14-9, 14-16, 15-1, 15-15, 15-17, 16-2, 16-5, 16-7, 17-8, 17-9, 17-16, 18-6, 18-13, 18-14, 19-1, 19-15, 19-17}
				\cellcolor{green}{1-4, 1-10, 1-19, 2-4, 2-10, 2-19, 3-2, 3-5, 3-6, 3-7, 3-13, 3-14, 4-4, 4-10, 4-19, 5-1, 5-8, 5-9, 5-15, 5-16, 5-17, 6-1, 6-8, 6-9, 6-15, 6-16, 6-17, 7-4, 7-10, 7-19, 9-3, 9-12, 9-18, 10-2, 10-5, 10-6, 10-7, 10-13, 10-14, 11-2, 11-5, 11-6, 11-7, 11-13, 11-14, 12-3, 12-12, 12-18, 13-1, 13-8, 13-9, 13-15, 13-16, 13-17, 14-3, 14-12, 14-18, 15-3, 15-12, 15-18, 16-4, 16-10, 16-19, 17-3, 17-12, 17-18, 18-4, 18-10, 18-19, 19-3, 19-12, 19-18} \rowcolor{yellow}{8}  \columncolor{yellow}{11}},
			columns-width = 4mm
			]%
			&   0 &   1 &   2 &   3 &   4 &   5 &   6 &   7 &   8 &   9 &  10 &  11 &  12 &  13 &  14 &  15 &  16 &  17 &  18 \\
			0   &  27 & 108 &  81 & 324 & 108 & 108 &  36 &  27 &   9 & 324 &   2 &  81 &  36 & 108 &   9 &  27 &  27 &  81 & 324 \\
			1   &  27 & 108 &  81 & 324 & 108 & 108 &  36 &  27 &   9 & 324 &   2 &  81 &  36 & 108 &   9 &  27 &  27 &  81 & 324 \\
			2   &  81 & 324 &   3 &  36 & 324 & 324 & 324 &  81 &  81 &  12 &   2 &   3 & 324 & 324 &  81 &  81 &  81 &   9 &  12 \\
			3   &  27 & 108 &  81 & 324 &  36 &  36 & 108 &  27 &  27 & 324 &   2 &  81 & 108 & 108 &  27 &   9 &   9 &  81 & 324 \\
			4   & 324 & 162 &  12 &  18 & 162 & 162 & 162 & 324 & 324 &   6 &   2 &  12 & 162 & 162 & 324 & 324 & 324 &  36 &   6 \\
			5   & 324 & 162 &  36 &   2 & 162 & 162 & 162 & 324 & 324 &  18 &   2 &  36 & 162 & 162 & 324 & 324 & 324 &   4 &  18 \\
			6   &  27 & 108 &  81 & 324 &  36 &  36 & 108 &  27 &  27 & 324 &   2 &  81 & 108 & 108 &  27 &   9 &   9 &  81 & 324 \\
			7   &   2 &   2 &   2 &   2 &   2 &   2 &   2 &   2 &   2 &   2 &   1 &   2 &   2 &   2 &   2 &   2 &   2 &   2 &   2 \\
			8   & 108 &  54 & 324 & 162 &  18 &  18 &  54 & 108 & 108 & 162 &   2 & 324 &  54 &  54 & 108 &  36 &  36 & 324 & 162 \\
			9   &  81 & 324 &   9 &   4 & 324 & 324 & 324 &  81 &  81 &  36 &   2 &   9 & 324 & 324 &  81 &  81 &  81 &   1 &  36 \\
			10  &  81 & 324 &   3 &  36 & 324 & 324 & 324 &  81 &  81 &  12 &   2 &   3 & 324 & 324 &  81 &  81 &  81 &   9 &  12 \\
			11  & 108 &  54 & 324 & 162 &  18 &  18 &  54 & 108 & 108 & 162 &   2 & 324 &  54 &  54 & 108 &  36 &  36 & 324 & 162 \\
			12  & 324 & 162 &  12 &  18 & 162 & 162 & 162 & 324 & 324 &   6 &   2 &  12 & 162 & 162 & 324 & 324 & 324 &  36 &   6 \\
			13  & 108 &  54 & 324 & 162 &  54 &  54 &  18 & 108 &  36 & 162 &   2 & 324 &  18 &  54 &  36 & 108 & 108 & 324 & 162 \\
			14  & 108 &  54 & 324 & 162 &  54 &  54 &  18 & 108 &  36 & 162 &   2 & 324 &  18 &  54 &  36 & 108 & 108 & 324 & 162 \\
			15  &   9 &  36 &  81 & 324 & 108 & 108 & 108 &   9 &  27 & 324 &   2 &  81 & 108 &  36 &  27 &  27 &  27 &  81 & 324 \\
			16  &  36 &  18 & 324 & 162 &  54 &  54 &  54 &  36 & 108 & 162 &   2 & 324 &  54 &  18 & 108 & 108 & 108 & 324 & 162 \\
			17  &   9 &  36 &  81 & 324 & 108 & 108 & 108 &   9 &  27 & 324 &   2 &  81 & 108 &  36 &  27 &  27 &  27 &  81 & 324 \\
			18  &  36 &  18 & 324 & 162 &  54 &  54 &  54 &  36 & 108 & 162 &   2 & 324 &  54 &  18 & 108 & 108 & 108 & 324 & 162 \\
		\end{NiceArray}
		$
		}
		\caption{The size of the sets $\langle x \rangle$ for $x \in \mathbb{E}_{{19}^2}$}\label{Table-E19-b}
	\end{table}
	\arrayrulecolor{black}

	\arrayrulecolor{cyan}
	\begin{table}[!ht]
		\centering
		{\footnotesize			
		$
		\setlength{\extrarowheight}{4mm}
		\begin{NiceArray}{rrrrrrrrrrrrrrrrrrrrrrr}%
			[small, hvlines, first-row, first-col, code-before ={\cellcolor{red}{1-16, 2-22, 3-5, 4-11, 5-17, 6-23, 7-6, 8-12, 10-1, 11-7, 12-13, 13-19, 15-8, 16-14, 17-20, 18-3, 20-15, 21-21, 22-4, 23-10} \cellcolor{green}{1-1, 1-7, 1-9, 1-12, 1-13, 1-19, 1-21, 1-22, 1-23, 2-7, 2-9, 2-12, 2-13, 2-16, 2-17, 2-19, 2-21, 2-23, 3-4, 3-6, 3-8, 3-10, 3-11, 3-14, 3-15, 3-18, 3-20, 4-3, 4-4, 4-5, 4-6, 4-8, 4-14, 4-15, 4-18, 4-20, 5-1, 5-7, 5-9, 5-12, 5-13, 5-19, 5-21, 5-22, 5-23, 6-1, 6-7, 6-9, 6-13, 6-16, 6-17, 6-19, 6-21, 6-22, 7-3, 7-4, 7-5, 7-8, 7-10, 7-11, 7-14, 7-15, 7-18, 8-1, 8-7, 8-9, 8-13, 8-16, 8-17, 8-19, 8-21, 8-22, 9-3, 9-4, 9-5, 9-6, 9-8, 9-10, 9-11, 9-14, 9-15, 9-20, 10-7, 10-9, 10-12, 10-13, 10-16, 10-17, 10-19, 10-21, 10-23, 11-1, 11-9, 11-12, 11-13, 11-16, 11-17, 11-19, 11-22, 11-23, 12-1, 12-7, 12-9, 12-12, 12-16, 12-17, 12-21, 12-22, 12-23, 13-1, 13-7, 13-9, 13-12, 13-16, 13-17, 13-21, 13-22, 13-23, 15-3, 15-4, 15-5, 15-6, 15-10, 15-11, 15-15, 15-18, 15-20, 16-3, 16-4, 16-5, 16-6, 16-10, 16-11, 16-15, 16-18, 16-20, 17-3, 17-4, 17-5, 17-8, 17-10, 17-11, 17-14, 17-15, 17-18, 18-4, 18-6, 18-8, 18-10, 18-11, 18-14, 18-15, 18-18, 18-20, 19-1, 19-7, 19-12, 19-13, 19-16, 19-17, 19-19, 19-21, 19-22, 19-23, 20-3, 20-5, 20-6, 20-8, 20-10, 20-11, 20-14, 20-18, 20-20, 21-1, 21-9, 21-12, 21-13, 21-16, 21-17, 21-19, 21-22, 21-23, 22-3, 22-5, 22-6, 22-8, 22-10, 22-11, 22-14, 22-18, 22-20, 23-3, 23-4, 23-5, 23-6, 23-8, 23-14, 23-15, 23-18, 23-20}
				\rowcolor{yellow}{14}  \columncolor{yellow}{2} }, 
			columns-width = 4mm,
			]%
			&  0 &  1 &  2 &  3  & 4 &  5 &  6 &  7 &  8 &  9 & 10 & 11 & 12 & 13 & 14 & 15 & 16 & 17 & 18 & 19 & 20 & 21 & 22  \\
			0  & 44 &  2 & 22 & 22 & 22 & 22 & 44 & 22 & 44 &  2 & 11 & 44 & 44 & 22 & 22 & 44 &  4 & 22 & 44 & 22 & 44 & 44 & 44  \\
			1  &  4 &  2 &  2 & 22 & 11 & 22 & 44 & 22 & 44 & 22 & 22 & 44 & 44 & 22 & 22 & 44 & 44 & 22 & 44 & 22 & 44 & 44 & 44  \\
			2  &  2 &  2 &  4 & 44 & 44 & 44 & 22 & 44 & 22 & 44 & 44 & 22 & 22 & 44 & 44 & 22 & 22 & 44 & 22 & 44 & 22 & 22 & 22  \\
			3  & 22 &  2 & 44 & 44 & 44 & 44 & 22 & 44 & 22 &  4 & 44 & 22 & 22 & 44 & 44 & 22 &  2 & 44 & 22 & 44 & 22 & 22 & 22  \\
			4  & 44 &  2 & 22 & 22 & 22 & 22 & 44 & 22 & 44 & 11 &  2 & 44 & 44 & 22 & 22 &  4 & 44 & 22 & 44 & 22 & 44 & 44 & 44  \\
			5  & 44 &  2 & 22 & 11 & 22 & 22 & 44 & 22 & 44 & 22 & 22 &  4 & 44 & 22 &  2 & 44 & 44 & 22 & 44 & 22 & 44 & 44 & 44  \\
			6  & 22 &  2 & 44 & 44 & 44 & 44 &  2 & 44 & 22 & 44 & 44 & 22 & 22 & 44 & 44 & 22 & 22 & 44 & 22 &  4 & 22 & 22 & 22  \\
			7  & 44 &  2 & 22 &  2 & 22 & 22 & 44 & 22 & 44 & 22 & 22 & 44 & 44 & 22 & 11 & 44 & 44 & 22 & 44 & 22 & 44 & 44 &  4  \\
			8  & 22 &  2 & 44 & 44 & 44 & 44 & 22 & 44 &  2 & 44 & 44 & 22 & 22 & 44 & 44 & 22 & 22 &  4 & 22 & 44 & 22 & 22 & 22  \\
			9  & 44 &  2 & 11 & 22 &  2 & 22 & 44 & 22 & 44 & 22 & 22 & 44 & 44 & 22 & 22 & 44 & 44 & 22 & 44 & 22 & 44 &  4 & 44  \\
			10  & 44 &  2 & 22 & 22 & 22 &  2 & 44 & 22 & 44 & 22 & 22 & 44 & 44 & 22 & 22 & 44 & 44 & 22 & 44 & 11 &  4 & 44 & 44  \\
			11  & 44 &  2 & 22 & 22 & 22 & 22 & 44 &  2 & 44 & 22 & 22 & 44 & 44 & 11 & 22 & 44 & 44 & 22 &  4 & 22 & 44 & 44 & 44  \\
			12  & 44 &  2 & 22 & 22 & 22 & 22 & 44 & 11 & 44 & 22 & 22 & 44 &  4 &  2 & 22 & 44 & 44 & 22 & 44 & 22 & 44 & 44 & 44  \\
			13  &  2 &  1 &  2 &  2 &  2 &  2 &  2 &  2 &  2 &  2 &  2 &  2 &  2 &  2 &  2 &  2 &  2 &  2 &  2 &  2 &  2 &  2 &  2  \\
			14  & 22 &  2 & 44 & 44 & 44 & 44 & 22 & 44 & 22 & 44 & 44 & 22 &  2 &  4 & 44 & 22 & 22 & 44 & 22 & 44 & 22 & 22 & 22  \\
			15  & 22 &  2 & 44 & 44 & 44 & 44 & 22 &  4 & 22 & 44 & 44 & 22 & 22 & 44 & 44 & 22 & 22 & 44 &  2 & 44 & 22 & 22 & 22  \\
			16  & 22 &  2 & 44 & 44 & 44 &  4 & 22 & 44 & 22 & 44 & 44 & 22 & 22 & 44 & 44 & 22 & 22 & 44 & 22 & 44 &  2 & 22 & 22  \\
			17  & 22 &  2 & 44 & 44 &  4 & 44 & 22 & 44 & 22 & 44 & 44 & 22 & 22 & 44 & 44 & 22 & 22 & 44 & 22 & 44 & 22 &  2 & 22  \\
			18  & 44 &  2 & 22 & 22 & 22 & 22 & 44 & 22 &  4 & 22 & 22 & 44 & 44 & 22 & 22 & 44 & 44 &  1 & 44 & 22 & 44 & 44 & 44  \\
			19  & 22 &  2 & 44 &  4 & 44 & 44 & 22 & 44 & 22 & 44 & 44 & 22 & 22 & 44 & 44 & 22 & 22 & 44 & 22 & 44 & 22 & 22 &  2  \\
			20  & 44 &  2 & 22 & 22 & 22 & 11 &  4 & 22 & 44 & 22 & 22 & 44 & 44 & 22 & 22 & 44 & 44 & 22 & 44 &  2 & 44 & 44 & 44  \\
			21  & 22 &  2 & 44 & 44 & 44 & 44 & 22 & 44 & 22 & 44 & 44 &  2 & 22 & 44 &  4 & 22 & 22 & 44 & 22 & 44 & 22 & 22 & 22  \\
			22  & 22 &  2 & 44 & 44 & 44 & 44 & 22 & 44 & 22 & 44 &  4 & 22 & 22 & 44 & 44 &  2 & 22 & 44 & 22 & 44 & 22 & 22 & 22  \\
		\end{NiceArray}
		$
		}
		\caption{The size of the sets $\langle x \rangle_2$ for $x \in \mathbb{E}_{{23}^2}$}\label{Table-E23-a}
	\end{table}
	\begin{table}[!ht]
		\centering
		{\footnotesize			
		$
		\setlength{\extrarowheight}{4mm}
		\begin{NiceArray}{rrrrrrrrrrrrrrrrrrrrrrr}%
			[small, hvlines, first-row, first-col, code-before ={\cellcolor{red}{1-16, 2-22, 3-5, 4-11, 5-17, 6-23, 7-6, 8-12, 10-1, 11-7, 12-13, 13-19, 15-8, 16-14, 17-20, 18-3, 20-15, 21-21, 22-4, 23-10}
			\cellcolor{blue!20}{1-3, 1-4, 1-5, 1-6, 1-8, 1-14, 1-15, 1-18, 1-20, 2-4, 2-6, 2-8, 2-10, 2-11, 2-14, 2-15, 2-18, 2-20, 5-3, 5-4, 5-5, 5-6, 5-8, 5-14, 5-15, 5-18, 5-20, 6-3, 6-5, 6-6, 6-8, 6-10, 6-11, 6-14, 6-18, 6-20, 8-3, 8-5, 8-6, 8-8, 8-10, 8-11, 8-14, 8-18, 8-20, 10-4, 10-6, 10-8, 10-10, 10-11, 10-14, 10-15, 10-18, 10-20, 11-3, 11-4, 11-5, 11-8, 11-10, 11-11, 11-14, 11-15, 11-18, 12-3, 12-4, 12-5, 12-6, 12-10, 12-11, 12-15, 12-18, 12-20, 13-3, 13-4, 13-5, 13-6, 13-10, 13-11, 13-15, 13-18, 13-20, 19-3, 19-4, 19-5, 19-6, 19-8, 19-10, 19-11, 19-14, 19-15, 19-20, 21-3, 21-4, 21-5, 21-8, 21-10, 21-11, 21-14, 21-15, 21-18}
			 \cellcolor{green}{1-1, 1-7, 1-9, 1-12, 1-13, 1-19, 1-21, 1-22, 1-23, 2-7, 2-9, 2-12, 2-13, 2-16, 2-17, 2-19, 2-21, 2-23, 3-4, 3-6, 3-8, 3-10, 3-11, 3-14, 3-15, 3-18, 3-20, 4-3, 4-4, 4-5, 4-6, 4-8, 4-14, 4-15, 4-18, 4-20, 5-1, 5-7, 5-9, 5-12, 5-13, 5-19, 5-21, 5-22, 5-23, 6-1, 6-7, 6-9, 6-13, 6-16, 6-17, 6-19, 6-21, 6-22, 7-3, 7-4, 7-5, 7-8, 7-10, 7-11, 7-14, 7-15, 7-18, 8-1, 8-7, 8-9, 8-13, 8-16, 8-17, 8-19, 8-21, 8-22, 9-3, 9-4, 9-5, 9-6, 9-8, 9-10, 9-11, 9-14, 9-15, 9-20, 10-7, 10-9, 10-12, 10-13, 10-16, 10-17, 10-19, 10-21, 10-23, 11-1, 11-9, 11-12, 11-13, 11-16, 11-17, 11-19, 11-22, 11-23, 12-1, 12-7, 12-9, 12-12, 12-16, 12-17, 12-21, 12-22, 12-23, 13-1, 13-7, 13-9, 13-12, 13-16, 13-17, 13-21, 13-22, 13-23, 15-3, 15-4, 15-5, 15-6, 15-10, 15-11, 15-15, 15-18, 15-20, 16-3, 16-4, 16-5, 16-6, 16-10, 16-11, 16-15, 16-18, 16-20, 17-3, 17-4, 17-5, 17-8, 17-10, 17-11, 17-14, 17-15, 17-18, 18-4, 18-6, 18-8, 18-10, 18-11, 18-14, 18-15, 18-18, 18-20, 19-1, 19-7, 19-12, 19-13, 19-16, 19-17, 19-19, 19-21, 19-22, 19-23, 20-3, 20-5, 20-6, 20-8, 20-10, 20-11, 20-14, 20-18, 20-20, 21-1, 21-9, 21-12, 21-13, 21-16, 21-17, 21-19, 21-22, 21-23, 22-3, 22-5, 22-6, 22-8, 22-10, 22-11, 22-14, 22-18, 22-20, 23-3, 23-4, 23-5, 23-6, 23-8, 23-14, 23-15, 23-18, 23-20} \rowcolor{yellow}{14}  \columncolor{yellow}{2} }, 
			columns-width = 4mm,
			]%
			&   0 &   1 &   2 &   3 &   4 &   5 &   6 &   7 &   8 &   9 &  10 &  11 &  12 &  13 &  14 &  15 &  16 &  17 &  18 &  19 &  20 &  21 &  22 \\
			0   & 484 &   2 & 121 & 121 & 121 & 121 & 484 & 121 & 484 &  11 &  11 & 484 & 484 & 121 & 121 &  44 &  44 & 121 & 484 & 121 & 484 & 484 & 484 \\
			1   &  44 &   2 &  11 & 121 &  11 & 121 & 484 & 121 & 484 & 121 & 121 & 484 & 484 & 121 & 121 & 484 & 484 & 121 & 484 & 121 & 484 &  44 & 484 \\
			2   &  22 &   2 &  44 & 484 &  44 & 484 & 242 & 484 & 242 & 484 & 484 & 242 & 242 & 484 & 484 & 242 & 242 & 484 & 242 & 484 & 242 &  22 & 242 \\
			3   & 242 &   2 & 484 & 484 & 484 & 484 & 242 & 484 & 242 &  44 &  44 & 242 & 242 & 484 & 484 &  22 &  22 & 484 & 242 & 484 & 242 & 242 & 242 \\
			4   & 484 &   2 & 121 & 121 & 121 & 121 & 484 & 121 & 484 &  11 &  11 & 484 & 484 & 121 & 121 &  44 &  44 & 121 & 484 & 121 & 484 & 484 & 484 \\
			5   & 484 &   2 & 121 &  11 & 121 & 121 & 484 & 121 & 484 & 121 & 121 &  44 & 484 & 121 &  11 & 484 & 484 & 121 & 484 & 121 & 484 & 484 &  44 \\
			6   & 242 &   2 & 484 & 484 & 484 &  44 &  22 & 484 & 242 & 484 & 484 & 242 & 242 & 484 & 484 & 242 & 242 & 484 & 242 &  44 &  22 & 242 & 242 \\
			7   & 484 &   2 & 121 &  11 & 121 & 121 & 484 & 121 & 484 & 121 & 121 &  44 & 484 & 121 &  11 & 484 & 484 & 121 & 484 & 121 & 484 & 484 &  44 \\
			8   & 242 &   2 & 484 & 484 & 484 & 484 & 242 & 484 &   2 & 484 & 484 & 242 & 242 & 484 & 484 & 242 & 242 &   4 & 242 & 484 & 242 & 242 & 242 \\
			9   &  44 &   2 &  11 & 121 &  11 & 121 & 484 & 121 & 484 & 121 & 121 & 484 & 484 & 121 & 121 & 484 & 484 & 121 & 484 & 121 & 484 &  44 & 484 \\
			10  & 484 &   2 & 121 & 121 & 121 &  11 &  44 & 121 & 484 & 121 & 121 & 484 & 484 & 121 & 121 & 484 & 484 & 121 & 484 &  11 &  44 & 484 & 484 \\
			11  & 484 &   2 & 121 & 121 & 121 & 121 & 484 &  11 & 484 & 121 & 121 & 484 &  44 &  11 & 121 & 484 & 484 & 121 &  44 & 121 & 484 & 484 & 484 \\
			12  & 484 &   2 & 121 & 121 & 121 & 121 & 484 &  11 & 484 & 121 & 121 & 484 &  44 &  11 & 121 & 484 & 484 & 121 &  44 & 121 & 484 & 484 & 484 \\
			13  &   2 &   1 &   2 &   2 &   2 &   2 &   2 &   2 &   2 &   2 &   2 &   2 &   2 &   2 &   2 &   2 &   2 &   2 &   2 &   2 &   2 &   2 &   2 \\
			14  & 242 &   2 & 484 & 484 & 484 & 484 & 242 &  44 & 242 & 484 & 484 & 242 &  22 &  44 & 484 & 242 & 242 & 484 &  22 & 484 & 242 & 242 & 242 \\
			15  & 242 &   2 & 484 & 484 & 484 & 484 & 242 &  44 & 242 & 484 & 484 & 242 &  22 &  44 & 484 & 242 & 242 & 484 &  22 & 484 & 242 & 242 & 242 \\
			16  & 242 &   2 & 484 & 484 & 484 &  44 &  22 & 484 & 242 & 484 & 484 & 242 & 242 & 484 & 484 & 242 & 242 & 484 & 242 &  44 &  22 & 242 & 242 \\
			17  &  22 &   2 &  44 & 484 &  44 & 484 & 242 & 484 & 242 & 484 & 484 & 242 & 242 & 484 & 484 & 242 & 242 & 484 & 242 & 484 & 242 &  22 & 242 \\
			18  & 484 &   2 & 121 & 121 & 121 & 121 & 484 & 121 &   4 & 121 & 121 & 484 & 484 & 121 & 121 & 484 & 484 &   1 & 484 & 121 & 484 & 484 & 484 \\
			19  & 242 &   2 & 484 &  44 & 484 & 484 & 242 & 484 & 242 & 484 & 484 &  22 & 242 & 484 &  44 & 242 & 242 & 484 & 242 & 484 & 242 & 242 &  22 \\
			20  & 484 &   2 & 121 & 121 & 121 &  11 &  44 & 121 & 484 & 121 & 121 & 484 & 484 & 121 & 121 & 484 & 484 & 121 & 484 &  11 &  44 & 484 & 484 \\
			21  & 242 &   2 & 484 &  44 & 484 & 484 & 242 & 484 & 242 & 484 & 484 &  22 & 242 & 484 &  44 & 242 & 242 & 484 & 242 & 484 & 242 & 242 &  22 \\
			22  & 242 &   2 & 484 & 484 & 484 & 484 & 242 & 484 & 242 &  44 &  44 & 242 & 242 & 484 & 484 &  22 &  22 & 484 & 242 & 484 & 242 & 242 & 242 \\
		\end{NiceArray}
		$
		}
		\caption{The size of the sets $\langle x \rangle$ for $x \in \mathbb{E}_{{23}^2}$}\label{Table-E23-b}
	\end{table}
	\arrayrulecolor{black}

\clearpage
\section{Observation for the dichotomy between even and odd bases}\label{App:Even-Odd-Example}
	Let us use the following Entropoid $\mathbb{E}_{{49223}^2}(a_3=33170, a_8=13052, b_2=12476, b_7=19648)$, which has $\mathbf{0}_* =  (20898, 8427)$, $\mathbf{1}_* = (29739, 25115)$ and $x * y = (x_1, x_2) * (y_1, y_2) = ( 13052 x_2 y_1 + 33170 x_2 + 24201 y_1 + 34725, \ \ \ 19648 x_1 y_2 + 12476 x_1 + 14362 y_2 + 19210)$. For a generator let us use $g = (21287, 34883)$.	
	
	
	For $\mathfrak{b}=3$ and for $i = 2, 3, 4$ one can check that these are the following outcomes:
	\begin{description}
		\item[$i = 2$,] $\mathbb{L}(i) = \{ [0, 0], [1, 0], [0, 1], [1, 1] \}$.\\
		$g^{(3, \mathbb{L}(i), \mathfrak{b})} = \{ (22143, 3374), (22143, 3374), (9735, 2125), (9735, 2125) \} $. As we can see there are only two outcomes: $g_{2,1}= (22143, 3374)$ and $g_{2,2}=(9735, 2125)$, so we get $r_2 = 2$, $\xi_2 = \{ C_{i,1}, C_{i, r_2}  \}$, where $C_{i,1} = \{[0, 0], [1, 0] \} $ and  $C_{i,2} = \{[0, 1], [1, 1] \} $. From this we get $	H_{\infty}(\xi_2)= H_{2}(\xi_2) = H_1(\xi_2) = 1$.
		\item[$i = 3$,] $\mathbb{L}(i) = \{ [0, 0, 0], [1, 0, 0], [0, 1, 0], [1, 1, 0], [0, 0, 1], [1, 0, 1], [0, 1, 1], [1, 1, 1] \}$. \\
		$g^{(9, \mathbb{L}(i), \mathfrak{b})} = \{ (12320, 26593), (12320, 26593), (28416, 42082), (28416, 42082), (28416, 42082),$\\
		$(28416, 42082), (12320, 26593), (12320, 26593) \} $. As we can see there are again only two outcomes: $g_{3,1}= (12320, 26593)$ and $g_{3,2}= (28416, 42082)$. So, again we have $r_3 = 2$, and now $\xi_3 = \{ C_{i,1}, C_{i, r_3}  \}$, where $C_{i,1} = \{ [   [0, 0, 0],      [1, 0, 0],      [0, 1, 1],      [1, 1, 1]] \} $ and  $C_{i,2} = \{ [0, 1, 0],          [1, 1, 0],        [0, 0, 1],      [1, 0, 1] \} $. From this we get $	H_{\infty}(\xi_3)= H_{2}(\xi_3) = H_1(\xi_3) = 1$.
		\item[$i = 4$,] $\mathbb{L}(i) = \{ [0, 0, 0, 0], [1, 0, 0, 0], \ldots, [1, 1, 1, 1]	 \}$. \\
		$g^{(27, \mathbb{L}(i), \mathfrak{b})} = \{ (42159, 1249), (42159, 1249), (46373, 13249), \ldots \} $. One can see that again there are only two outcomes: $g_{4,1}= (42159, 1249)$ and $g_{4,2}= (46373, 13249)$. So, $r_4 = 2$, and now $\xi_4 = \{ C_{i,1}, C_{i, r_4}  \}$, where $|C_{i,1}| = 8 $ and  $|C_{i,2}| = 8$. From this we get $	H_{\infty}(\xi_4)= H_{2}(\xi_4) = H_1(\xi_4) = 1$.	
	\end{description}
	
	For $\mathfrak{b}=4$ and for $i = 2, 3, \ldots 9$ a summary table of the obtained calculations is given in Table \ref{Table:Base-4}. As we can see, the sets $\xi_i$ are partitioned always in 3 subsets, but the entropies tend to 0 as $i$ is increasing.
	\begin{table}[htbp]
		\scriptsize 
		\centering
		\renewcommand{\arraystretch}{1.0}
		\begin{tabular}{|c|r|r|r|r|r|c|r|r|r|r|r||}
			\rowcolor[rgb]{ 1,  .78,  .808} \multicolumn{12}{c}{$\mathfrak{b} = 4$} \bigstrut[b]\\
			\hline
			\multicolumn{6}{|c|}{$i = 2$}                                                                             & \multicolumn{6}{c||}{$i = 3$} \bigstrut\\
			\hline
			$r_i$           & $g_{i,j}$       & $n_{ij}$        & $H_{\infty}$    & $H_2$           & $H_1$           & $r_i$           & $g_{i,j}$       & $n_{ij}$        & $H_{\infty}$    & $H_2$           & $H_1$ \bigstrut\\
			\hline
			\multirow{4}[8]{*}{3} & (2847, 43103)   & 3               & \multirow{4}[8]{*}{1.585} & \multirow{4}[8]{*}{1.585} & \multirow{4}[8]{*}{1.585} & \multirow{4}[8]{*}{3} & (37676, 4224)   & 6               & \multicolumn{1}{c|}{\multirow{4}[8]{*}{0.848}} & \multicolumn{1}{c|}{\multirow{4}[8]{*}{1.295}} & \multicolumn{1}{c||}{\multirow{4}[8]{*}{1.436}} \bigstrut\\
			\cline{2-3}\cline{8-9}                    & (12306, 3250)   & 3               &                 &                 &                 &                 & (14769, 4826)   & 6               &                 &                 &  \bigstrut\\
			\cline{2-3}\cline{8-9}                    & (43283, 29857)  & 3               &                 &                 &                 &                 & (27843, 29019)  & 15              &                 &                 &  \bigstrut\\
			\cline{2-3}\cline{8-9}                    & $\Sigma$        & 9               &                 &                 &                 &                 & $\Sigma$        & 27              &                 &                 &  \bigstrut\\
			\hline
			\multicolumn{6}{|c|}{$i = 4$}                                                                             & \multicolumn{6}{c||}{$i = 5$} \bigstrut\\
			\hline
			$r_i$           & $g_{i,j}$       & $n_{ij}$        & $H_{\infty}$    & $H_2$           & $H_1$           & $r_i$           & $g_{i,j}$       & $n_{ij}$        & $H_{\infty}$    & $H_2$           & $H_1$ \bigstrut\\
			\hline
			\multirow{4}[8]{*}{3} & (9873, 27342)   & 12              & \multirow{4}[8]{*}{0.507} & \multirow{4}[8]{*}{0.891} & \multirow{4}[8]{*}{1.173} & \multirow{4}[8]{*}{3} & (10067, 22108)  & 24              & \multirow{4}[8]{*}{0.317} & \multirow{4}[8]{*}{0.592} & \multirow{4}[8]{*}{0.914} \bigstrut\\
			\cline{2-3}\cline{8-9}                    & (44897, 4336)   & 12              &                 &                 &                 &                 & (6487, 4975)    & 24              &                 &                 &  \bigstrut\\
			\cline{2-3}\cline{8-9}                    & (31057, 15755)  & 57              &                 &                 &                 &                 & (22832, 44737)  & 195             &                 &                 &  \bigstrut\\
			\cline{2-3}\cline{8-9}                    & $\Sigma$        & 81              &                 &                 &                 &                 & $\Sigma$        & 243             &                 &                 &  \bigstrut\\
			\hline
			\multicolumn{6}{|c|}{$i = 6$}                                                                             & \multicolumn{6}{c||}{$i = 7$} \bigstrut\\
			\hline
			$r_i$           & $g_{i,j}$       & $n_{ij}$        & $H_{\infty}$    & $H_2$           & $H_1$           & $r_i$           & $g_{i,j}$       & $n_{ij}$        & $H_{\infty}$    & $H_2$           & $H_1$ \bigstrut\\
			\hline
			\multirow{4}[8]{*}{3} & (19981, 22570)  & 48              & \multirow{4}[8]{*}{0.204} & \multirow{4}[8]{*}{0.391} & \multirow{4}[8]{*}{0.694} & \multirow{4}[8]{*}{3} & (43901, 19938)  & 96              & \multirow{4}[8]{*}{0.133} & \multirow{4}[8]{*}{0.258} & \multirow{4}[8]{*}{0.517} \bigstrut\\
			\cline{2-3}\cline{8-9}                    & (35514, 19869)  & 48              &                 &                 &                 &                 & (2901, 22539)   & 96              &                 &                 &  \bigstrut\\
			\cline{2-3}\cline{8-9}                    & (31074, 13020)  & 633             &                 &                 &                 &                 & (1892, 14331)   & 1995            &                 &                 &  \bigstrut\\
			\cline{2-3}\cline{8-9}                    & $\Sigma$        & 729             &                 &                 &                 &                 & $\Sigma$        & 2187            &                 &                 &  \bigstrut\\
			\hline
			\multicolumn{6}{|c|}{$i = 8$}                                                                             & \multicolumn{6}{c||}{$i = 9$} \bigstrut\\
			\hline
			$r_i$           & $g_{i,j}$       & $n_{ij}$        & $H_{\infty}$    & $H_2$           & $H_1$           & $r_i$           & $g_{i,j}$       & $n_{ij}$        & $H_{\infty}$    & $H_2$           & $H_1$ \bigstrut\\
			\hline
			\multirow{4}[8]{*}{3} & (42829, 25216)  & 192             & \multirow{4}[8]{*}{0.087} & \multirow{4}[8]{*}{0.171} & \multirow{4}[8]{*}{0.380} & \multirow{4}[8]{*}{3} & (9873, 27342)   & 384             & \multirow{4}[8]{*}{0.057} & \multirow{4}[8]{*}{0.114} & \multirow{4}[8]{*}{0.277} \bigstrut\\
			\cline{2-3}\cline{8-9}                    & (18292, 35754)  & 192             &                 &                 &                 &                 & (31057, 15755)  & 384             &                 &                 &  \bigstrut\\
			\cline{2-3}\cline{8-9}                    & (44720, 23968)  & 6177            &                 &                 &                 &                 & (44897, 4336)   & 18915           &                 &                 &  \bigstrut\\
			\cline{2-3}\cline{8-9}                    & $\Sigma$        & 6561            &                 &                 &                 &                 & $\Sigma$        & 19683           &                 &                 &  \bigstrut\\
			\hline
		\end{tabular}%
		\caption{A summary table of calculations with base $\mathfrak{b}=4$}  \label{Table:Base-4}%
	\end{table}%
	
	For odd bases $\mathfrak{b}=5$ and $\mathfrak{b}=7$ things are getting more interesting, as the number $r_i$ of partition parts for the sets $\xi_i$ is increasing as $i$ increases. The entropies $H_{\infty}$, $H_{2}$ and $H_1$ are also increasing.
	
	On the other hand, for even bases $\mathfrak{b}=6$ and $\mathfrak{b}=8$, we see that the number of partitions increases as well, but the entropies $H_{\infty}$, $H_{2}$ and $H_1$, after an initial increase, start to decrease (and trend to zero).
	
	The summary of the calculations is given in Table \ref{Table:Base-5-8}, where in order to keep a reasonable table space we omitted the representatives $g_{i,j}$ of the partitioned classes. 
	\begin{table}[htbp]
		\centering
		\scriptsize 
		\centering
		\renewcommand{\arraystretch}{1.0}    
		\begin{tabular}{|r|r|r|r|r|r|r|r|r|r|r|r|r|r|}
			\hline
			\rowcolor[rgb]{ 1,  .78,  .808} \multicolumn{7}{|c|}{$\mathfrak{b} = 5$}                                                                                               & \multicolumn{7}{c|}{$\mathfrak{b} = 6$} \bigstrut\\
			\hline
			$i$             & $r_i$           & $\min n_{ij}$   & $\max n_{ij}$   & $H_{\infty}$    & $H_2$           & $H_1$           & $i$             & $r_i$           & $\min n_{ij}$   & $\max n_{ij}$   & $H_{\infty}$    & $H_2$           & $H_1$ \bigstrut\\
			\hline
			\hline
			2               & 4               & 4               & 5               & 2.000           & 2.000           & 2.000           & 2               & 5               & 5               & 5               & 2.322           & 2.322           & 2.322 \bigstrut\\
			\hline
			3               & 6               & 8               & 16              & 2.000           & 2.415           & 2.500           & 3               & 7               & 10              & 45              & 1.474           & 2.276           & 2.543 \bigstrut\\
			\hline
			4               & 8               & 16              & 48              & 2.415           & 2.678           & 2.811           & 4               & 9               & 20              & 305             & 1.035           & 1.841           & 2.439 \bigstrut\\
			\hline
			5               & 10              & 32              & 192             & 2.415           & 2.871           & 3.031           & 5               & 11              & 40              & 1845            & 0.760           & 1.429           & 2.218 \bigstrut\\
			\hline
			6               & 12              & 64              & 640             & 2.678           & 3.023           & 3.198           & 6               & 13              & 80              & 10505           & 0.573           & 1.104           & 1.961 \bigstrut\\
			\hline
			7               & 14              & 128             & 2560            & 2.678           & 3.148           & 3.333           & 7               & 15              & 160             & 57645           & 0.439           & 0.857           & 1.704 \bigstrut\\
			\hline
			8               & 16              & 256             & 8960            & 2.871           & 3.255           & 3.447           & 8               & 17              & 320             & 308705          & 0.340           & 0.669           & 1.464 \bigstrut\\
			\hline
			9               & 18              & 512             & 35840           & 2.871           & 3.348           & 3.544           & 9               & 19              & 640             & 1625445         & 0.265           & 0.524           & 1.247 \bigstrut\\
			\hline
			\rowcolor[rgb]{ 1,  .78,  .808} \multicolumn{7}{|c|}{$\mathfrak{b} = 7$}                                                                                               & \multicolumn{7}{c|}{$\mathfrak{b} = 8$} \bigstrut\\
			\hline
			$i$             & $r_i$           & $\min n_{ij}$   & $\max n_{ij}$   & $H_{\infty}$    & $H_2$           & $H_1$           & $i$             & $r_i$           & $\min n_{ij}$   & $\max n_{ij}$   & $H_{\infty}$    & $H_2$           & $H_1$ \bigstrut\\
			\hline
			\hline
			2               & 6               & 6               & 6               & 2.585           & 2.585           & 2.585           & 2               & 7               & 7               & 7               & 2.807           & 2.807           & 2.807 \bigstrut\\
			\hline
			3               & 12              & 12              & 24              & 3.170           & 3.433           & 3.503           & 3               & 13              & 14              & 91              & 1.914           & 3.054           & 3.408 \bigstrut\\
			\hline
			4               & 20              & 24              & 144             & 3.170           & 3.971           & 4.124           & 4               & 21              & 28              & 889             & 1.433           & 2.622           & 3.548 \bigstrut\\
			\hline
			5               & 30              & 48              & 576             & 3.755           & 4.360           & 4.580           & 5               & 31              & 56              & 7735            & 1.120           & 2.146           & 3.467 \bigstrut\\
			\hline
			6               & 42              & 96              & 2880            & 4.018           & 4.666           & 4.933           & 6               & 43              & 112             & 63217           & 0.896           & 1.751           & 3.278 \bigstrut\\
			\hline
			7               & 56              & 192             & 17280           & 4.018           & 4.918           & 5.220           & 7               & 57              & 224             & 496951          & 0.729           & 1.437           & 3.039 \bigstrut\\
			\hline
			8               & 72              & 384             & 80640           & 4.380           & 5.132           & 5.459           & 8               & 73              & 448             & 3805249         & 0.599           & 1.188           & 2.780 \bigstrut\\
			\hline
			9               & 90              & 768             & 430080          & 4.550           & 5.318           & 5.664           & 9               & 91              & 896             & 28596295        & 0.497           & 0.988           & 2.521 \bigstrut\\
			\hline
		\end{tabular}%
		\caption{A summary table of calculations with bases $\mathfrak{b}=5, 6, 7, 8$}  \label{Table:Base-5-8}%
	\end{table}%

	\begin{figure}[h!]
		\centering
		\includegraphics[scale=0.5]{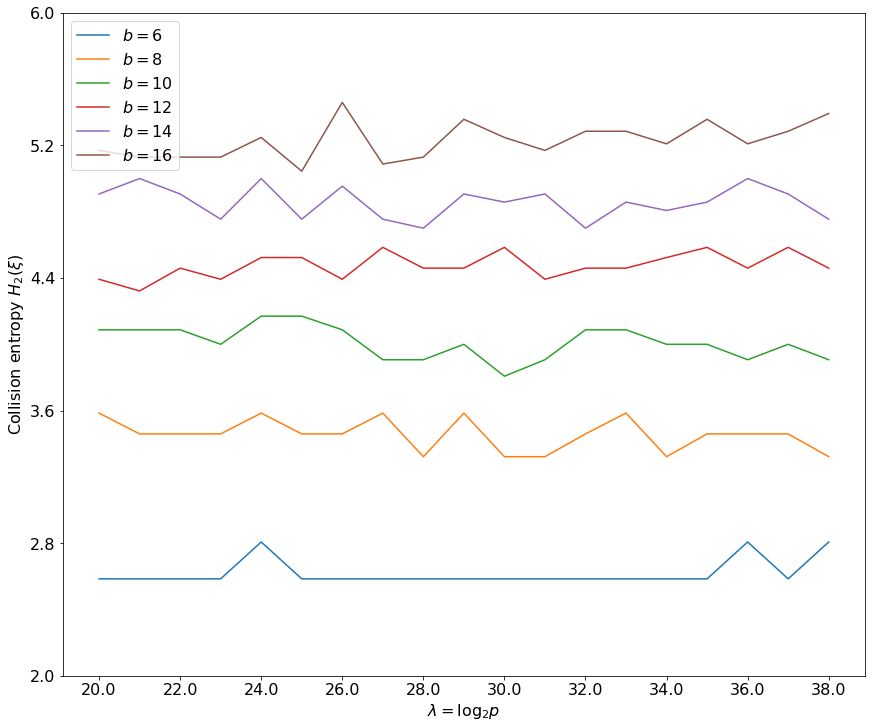}
		\caption{Collision entropy $H_2(\xi)$ experimentally calculated for bases $\mathfrak{b} = 6, 8, \ldots, 16$, for different entropoids $\mathbb{E}_p$ where $\lambda = \log_2 p$ varies in the interval $[20, 38]$. The collision entropy values were computed as an average of 100 experiments. }
		\label{Figure:EvenBases}
	\end{figure}
	\begin{figure}[h!]
		\centering
		\includegraphics[scale=0.5]{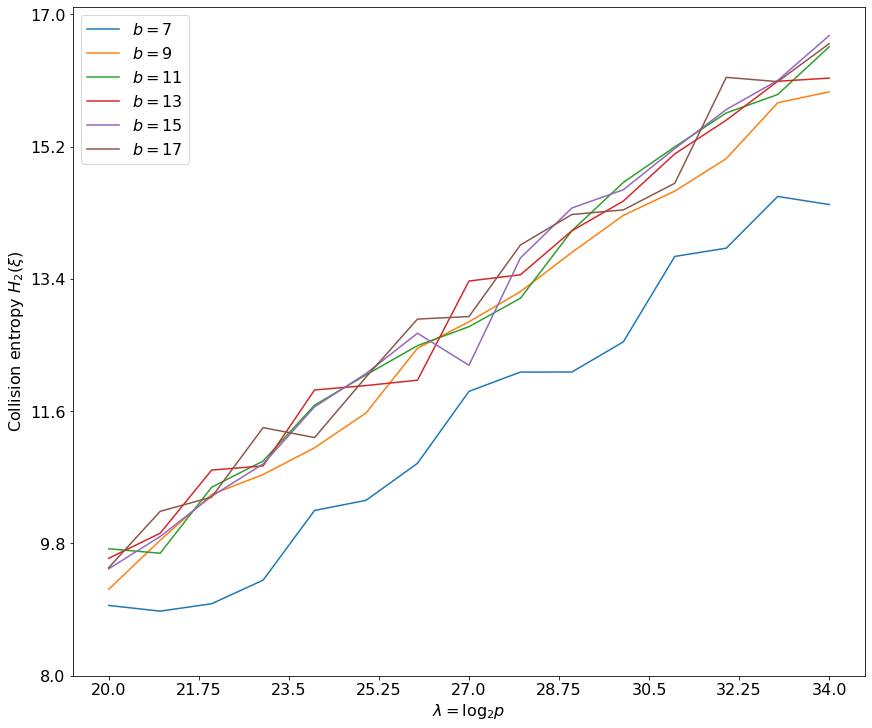}
		\caption{Collision entropy $H_2(\xi)$ experimentally calculated for bases $\mathfrak{b} = 7, 9, \ldots, 17$, for different entropoids $\mathbb{E}_p$ where $\lambda = \log_2 p$ varies in the interval $[20, 34]$. The collision entropy values were computed as an average of 10 experiments.}
		\label{Figure:OddBases}
	\end{figure}

\clearpage
\section{Proof-of-concept SageMath Jupyter implementation of the algorithms given in "Entropoid Based Cryptography"}\label{App:Proof-of-concept-SageMath}
The file \texttt{Proof\_of\_concept\_SageMath\_Jupyter\_implementation.ipynb} is provided in the folder \texttt{/anc/} with this Arxiv submission.

\end{document}